\documentclass[aps,twocolumn,
 superscriptaddress,amsmath,amssymb,nofootinbib]{revtex4-2}
\usepackage{amsfonts}
\usepackage{amsmath}
\usepackage{amsthm}
\usepackage{bbm}
\usepackage{amssymb}
\usepackage{graphicx}%
\usepackage{footmisc}
\usepackage{bm}
\usepackage{braket}
\usepackage{hyperref}
\usepackage{cleveref}
\usepackage{longtable,booktabs}
\usepackage{subfigure}
\usepackage{cancel}
\newtheorem{lemma}{Lemma}
\newtheorem{proposition}{Proposition} 
\newtheorem{corollary}{Corollary}
\theoremstyle{remark}
\newtheorem*{remark}{Remark}

\def\XXint#1#2#3{{\setbox0=\hbox{$#1{#2#3}{\int}$}
    \vcenter{\hbox{$#2#3$}}\kern-.5\wd0}}

\usepackage{xeCJK}

\newcommand{\Tr}{\text{Tr}}

\begin{document}

\title{A Multi-Resolvent Hierarchy for the ETH Smooth Function}

\author{Zhiqiang Huang (黄志强)}
\email{zqhuang@hubu.edu.cn}
\affiliation{School of Physics, Hubei University, Wuhan 430062, China.}

\date{\today}

\begin{abstract}
The eigenstate thermalization hypothesis (ETH) provides a statistical
description of thermalization in isolated quantum many-body systems,
yet the phenomenological smooth function $f_O(\bar{E},\omega)$---which
controls the energy dependence of off-diagonal matrix elements---lacks
a systematic microscopic foundation.
We develop a multi-resolvent hierarchy for the correlation corrections
entering the ETH smooth function.
Using recursive projection identities together with a diagonal closure
approximation (DCA), the hierarchy organizes multi-channel interference
processes by the number of interacting bath channels, replacing the
uncontrolled neglect of higher-order correlations with a systematically
improvable expansion.
The ETH smooth function is thereby obtained as
$f_{ji}^2 = D_{ji} + \sum_{r\ge 2} g_{ji}^{(r)}$,
where the diagonal baseline $D_{ji}$ and each correlation level
$g_{ji}^{(r)}$ are expressed entirely through diagonal spectral
functions and microscopic interaction couplings, providing a unified,
closed at the mean level, and systematically improvable microscopic
theory of the smooth function.
A rigorous projector sum rule constrains the entire hierarchy: the
integrated off-diagonal correlation carries a negative bias of order
unity, a consequence of projector idempotency.
The hierarchy further reveals a parity structure in which
even-$r$ sectors carry even parity under $\omega\to-\omega$ while
the $r=3$ sector generates the first odd-parity (skewness)
contribution---inaccessible within the parity-preserving
single-resolvent closures considered here---suggesting
experimentally testable signatures in quantum many-body systems.
A companion first-principles analysis (Appendix) shows that the
entropy normalization and the diagonal--off-diagonal trade-off of
the ETH ansatz are exact counting theorems, and that the fluctuation
sector of the same hierarchy governs the statistics of the ETH
matrix elements: in the decorrelated (chaotic) regime it selects the
Gaussian, self-averaging limit and fixes the
diagonal--off-diagonal fluctuation ratio, while the failure of
decorrelation in integrable systems accounts for the observed
enhancement of diagonal fluctuations.
\end{abstract}


\maketitle

\section{Introduction}
\label{INTRO}

Understanding how isolated quantum many-body systems
approach thermal equilibrium remains a central problem in
statistical physics.
For generic nonintegrable systems, the eigenstate
thermalization hypothesis (ETH) provides the most successful
framework for explaining thermalization directly at the
level of individual many-body eigenstates
\cite{Deu91,Sre94,Rig08,DKPR16}.

The ETH ansatz states that matrix elements of a local
observable $\hat O$ in the energy eigenbasis take the form
\begin{equation}
O_{nm}
=
O(\bar E)\,\delta_{nm}
+
e^{-S(\bar E)/2}\,
f_O(\bar E,\omega)\,
R_{nm},
\end{equation}
where $S(\bar E)$ is the thermodynamic entropy,
$f_O(\bar E,\omega)$ is a smooth structure function,
and $R_{nm}$ is a fluctuating quantity with unit variance.
This formulation successfully explains thermalization,
linear response, and spectral statistics in a broad class
of quantum chaotic systems
\cite{DKPR16}.

Despite its success, the ETH ansatz remains largely
phenomenological.
The smooth function $f_O(\bar E,\omega)$ and the statistical
properties of the fluctuation variable $R_{nm}$ are usually
introduced as assumptions rather than derived from a
microscopic framework.
While recent developments have explored higher-order
generalizations of ETH and non-Gaussian fluctuations
\cite{FK19},
a systematic microscopic theory that determines the smooth
function $f_O(\bar{E},\omega)$---and the correlation
corrections that enrich it---remains an open challenge.

A natural language for addressing such questions is
provided by Green's functions and resolvents
\cite{Eco06,Mah00}.
In interacting many-body systems, resolvents encode the
full spectral information of the Hamiltonian and form the
basis of self-energy and projection-operator approaches
\cite{Fes58}.
However, conventional self-consistent approximations,
such as self-consistent Born-type closures, typically
truncate the hierarchy at the level of a single diagonal
resolvent.
As a consequence, correlations generated by simultaneous
propagation through multiple interacting channels are
either neglected or absorbed into effective parameters.

In this work we develop a nonperturbative multi-resolvent
hierarchy that reorganizes the diagonal Green's-function
expansion into a systematic expansion of explicitly
correlated propagation processes.
Starting from exact projection identities and the spectral
representation of the resolvent, we derive a recursive
decomposition in which off-diagonal propagation is
systematically expressed through products of diagonal
resolvents.
This construction naturally generates a hierarchy of
multi-channel interference contributions involving
increasing numbers of interacting bath channels.

This hierarchy yields a microscopic theory of the ETH
smooth function $f_{ji}(E,\omega)$, decomposing $f_{ji}^2$ into a
diagonal overlap baseline $D_{ji}$, determined by single-resolvent
self-consistency, and a systematically improvable multi-resolvent
series $g_{ji} = \sum_{r\ge 2} g_{ji}^{(r)}$, where each level
$g_{ji}^{(r)}$ encodes interference among $r$ distinct bath
channels.  The hierarchy reveals a decisive parity structure: the
leading level $g_{ji}^{(2)}$ carries strictly even parity under
$\omega \to -\omega$, while $g_{ji}^{(3)}$ introduces the first
odd-parity (skewness) contribution, a qualitatively new signature
inaccessible within the parity-preserving single-resolvent closures
considered here.

We do not attempt to derive the ETH ansatz itself from first
principles.  Taking the universal entropy normalization as given,
we determine its remaining dynamical object, the smooth function
$f_O(\bar E,\omega)$, from microscopic resolvent theory through
\begin{equation}
f_{ji}^2 = D_{ji} + g_{ji},
\qquad
g_{ji} = \sum_{r\ge 2} g_{ji}^{(r)} .
\label{eq:intro_decomp}
\end{equation}
The smooth function is the mean-sector manifestation of a broader
resolvent hierarchy whose fluctuation sector, the resolvent
covariance, controls the statistics of the ETH matrix elements;
the two sectors are analysed in
Secs.~\ref{sec:expansion}--\ref{sec:eth_ansatz} and in
Appendix~\ref{app:first_principles}, respectively.

The remainder of the paper is organized as follows.
Sec.~\ref{sec:bridge} introduces the system-bath setup
and exactly decomposes the off-diagonal ETH variance
into a diagonal overlap baseline and a reduced-correlation
contribution, establishing off-diagonal resolvents as the
natural framework for the latter.
Secs.~\ref{sec:expansion}--\ref{sec:DCA} then express
this reduced correlation through a hierarchy of
multi-resolvent interference processes organised by
the number of interaction vertices---prefaced by the scale-matched
replacement principle that underlies both the approximation
replacement of Sec.~\ref{sec:approx} and the DCA---and formulate the
DCA that controls the
resulting expansion.
Sec.~\ref{sec:eth_ansatz} shows how this hierarchy
provides a microscopic theory of the ETH smooth function
$f_{ji}$, expressing it as $f_{ji}^2 = D_{ji} + \sum_{r\ge2}
g_{ji}^{(r)}$ and unifying the approximation replacement,
the reduced-correlation picture, and the higher-order ETH
framework within a single organising principle; its
closing subsection introduces the fluctuation sector of the
hierarchy and the exact diagonal--off-diagonal trade-off it
encodes.
Sec.~\ref{sec:connections} discusses
connections to OTOCs, Krylov complexity, fluctuation
theorems, and open-system dynamics.
Finally, Sec.~\ref{sec:conclusion} summarises the
unifying principles of the multi-resolvent framework
and outlines directions for future work.
Appendix~\ref{app:first_principles} then establishes the
first-principles status of the fluctuation sector: the entropy
normalization and the diagonal--off-diagonal trade-off are exact
counting theorems, while the decorrelated (chaotic) regime selects
the Gaussian, self-averaging limit of the resolvent covariance
hierarchy, whose failure in integrable systems quantitatively
accounts for the enhancement of diagonal fluctuations observed in
Ref.~\cite{HCZ25}.

\section{From Overlaps to Off-Diagonal Resolvents}
\label{sec:bridge}

\subsection{System-bath setup and diagonal overlaps}
\label{sec:setup}

Consider a system \(S\) and a bath \(B\) with unperturbed Hamiltonian 
\(H_0 = H_S + H_B\). The unperturbed eigenstates are product states
\begin{equation}
    \ket{\phi_{\mu i}} = \ket{\phi_i^S} \otimes \ket{\phi_\mu^B},
\quad
    H_0 \ket{\phi_{\mu i}} = a_{\mu i} \ket{\phi_{\mu i}},
\quad
    a_{\mu i} = E_i + \epsilon_\mu .
\end{equation}
Introducing an interaction \(V\) couples the subsystems:
\begin{equation}
    H = H_0 + V,
\qquad
    H \ket{\psi_n} = \lambda_n \ket{\psi_n}.
\end{equation}
The central microscopic objects are the overlap amplitudes 
\(\braket{\psi_n|\phi_{\mu i}}\). Their squared moduli define the probability 
distribution of unperturbed states among the exact eigenstates:
\begin{equation}
    p^{\mu i}_n := |\braket{\psi_n|\phi_{\mu i}}|^2,
\qquad
    \sum_n p^{\mu i}_n = 1 .
\end{equation}
Under the eigenstate thermalization hypothesis (ETH), these overlaps become 
smooth functions of energy when coarse-grained over small energy windows, 
described by a distribution \(f^{\mu i}(\lambda_n)\) via
\begin{equation}
    \mathbb{E}(p^{\mu i}_n) = e^{-S(\lambda_n)}f^{\mu i}(\lambda_n),
    \label{eq:p_exp}
\end{equation}
where \(e^{S(\lambda)}\) is the density of states of the total Hamiltonian 
\(H\). For weakly interacting systems, self-consistent treatments
\cite{HHTG24,HC25} yield a Lorentzian profile for \(f^{\mu i}\), while 
nonperturbative resolvent methods~\cite{HC26} reveal richer structures 
including Gaussian tails and spectral skewness.

The overlap distribution admits an exact spectral representation through 
the diagonal resolvent:
\begin{equation}
    \mathcal{R}_{\mu i}(z) 
    := \bra{\phi_{\mu i}} \frac{1}{z-H} \ket{\phi_{\mu i}}
    = \sum_n \frac{p^{\mu i}_n}{z - \lambda_n},
    \label{eq:Rdiag}
\end{equation}
from which the overlaps are recovered via
\begin{equation}
    p^{\mu i}_n = \frac{1}{\pi} 
    \lim_{\eta \to 0^+} \Im \mathcal{R}_{\mu i}(\lambda_n - i\eta).
    \label{eq:spir}
\end{equation}
This identity underpins the resolvent-based approach to ETH developed in 
Ref.~\cite{HC26}: the self-energy of \(\mathcal{R}_{\mu i}\) is reorganized 
into an exact multi-resolvent hierarchy, yielding systematic nonperturbative 
corrections beyond Lorentzian (SCBA-type) approximations.

\subsection{ETH matrix elements and the insufficiency of diagonal overlaps}
\label{sec:insufficiency}

While the overlaps \(p^{\mu i}_n\) fully characterize the diagonal ETH matrix 
elements \(\sigma_{nn}^{ji}\), the off-diagonal elements
\begin{equation}
    \sigma_{nm}^{ji} 
    := \bra{\phi_j^S} \Tr_B\bigl(\ket{\psi_n}\bra{\psi_m}\bigr) \ket{\phi_i^S}
\end{equation}
require a more fundamental building block. Define the transition amplitude 
product
\begin{equation}
    \Gamma_{nm,\mu}^{ji} 
    := \braket{\phi_{\mu j}|\psi_n} \braket{\psi_m|\phi_{\mu i}},
    \label{eq:Gamma_def}
\end{equation}
which satisfies \(\sigma_{nm}^{ji} = \sum_\mu \Gamma_{nm,\mu}^{ji}\).
The quantities \(\Gamma\) obey the exact composition rules
\begin{equation}
    \Gamma_{nm,\mu}^{ji} \, \Gamma_{ml,\mu}^{ik} 
    = p^{\mu i}_m \, \Gamma_{nl,\mu}^{jk},
\qquad
    \Gamma_{nn,\mu}^{jj} = p^{\mu j}_n,
    \label{eq:Gamma_rules}
\end{equation}
which follow directly from the completeness of the exact eigenstates.

The squared modulus of an off-diagonal ETH matrix element decomposes as
\begin{equation}
    |\sigma_{nm}^{ji}|^2 
    = \sum_{\mu} p^{\mu i}_m p^{\mu j}_n 
    + \mathcal{C}_{nmn}^{jij},
    \label{eq:decomp}
\end{equation}
where the \textbf{correlation term}
\begin{equation}
    \mathcal{C}_{nmn}^{jij} 
    := \sum_{\mu,\nu \neq \mu} 
    \Gamma_{nm,\mu}^{ji} \, \Gamma_{mn,\nu}^{ij}
    \label{eq:Cdef}
\end{equation}
involves products of transition amplitudes carrying \textbf{different bath 
indices} (\(\mu \neq \nu\)).

It is here that the standard overlap-based approach encounters a 
fundamental limitation. The diagonal part \(\sum_\mu p^{\mu i}_m p^{\mu j}_n\) 
is expressible entirely through the smooth functions \(f^{\mu i}\) and the 
density of states. The correlation term \(\mathcal{C}_{nmn}^{jij}\), by 
contrast, couples \textbf{different bath channels} and depends on the 
\textbf{coherent interference} of amplitudes with \(\mu \neq \nu\). No 
amount of knowledge of the diagonal overlaps \(p^{\mu i}_n\) alone can 
reconstruct this term.

That this correlation is not a negligible fluctuation follows from the
exact sum rule implied by the idempotency $\rho_n^2 = \rho_n$ of the
pure-state projector $\rho_n := \ket{\psi_n}\bra{\psi_n}$.
Introduce the system projection operators
\begin{equation}
    P_i := I_B \otimes \ket{\phi_i^S}\bra{\phi_i^S},
    \qquad
    P_i^2 = P_i,
\end{equation}
which select the $i$-th system basis state.  Using $\sum_m \ket{\psi_m}\bra{\psi_m} = I$, 
a direct expansion gives
\begin{equation}
    \sum_m |\sigma_{nm}^{ji}|^2
    = \bra{\psi_n} P_j \ket{\psi_n}
    = \sum_\mu p_n^{\mu j}.
    \label{eq:P_identity}
\end{equation}
Inserting the decomposition~\eqref{eq:decomp} and evaluating the
geometric right-hand side yields the projector sum rule:
\begin{equation}
   \sum_m \mathcal{C}_{nmn}^{jij} = 0.
    \label{eq:sumrule_zero}
\end{equation}
For the diagonal case $j=i$, the self-correlation
$\mathcal{C}_{nnn}^{iii} = \sum_{\mu\neq\nu} p^{\mu i}_n p^{\nu i}_n$
is manifestly positive.  Hence
\begin{equation}
    \sum_{m\neq n} \mathcal{C}_{nmn}^{iii}
    = -\,\mathcal{C}_{nnn}^{iii} < 0.
    \label{eq:sumrule_negative}
\end{equation}

\textbf{Theorem~1 (Projector sum rule).}
The exact sum rule $\sum_m \mathcal{C}_{nmn}^{jij} = 0$ and the
negativity $\sum_{m\neq n} \mathcal{C}_{nmn}^{iii} < 0$ follow
rigorously from the idempotency $\rho_n^2 = \rho_n$ of the pure-state
projector.  The negative sign is a \emph{geometric} consequence of
projector idempotency, independent of any ETH, chaos, or
random-matrix assumption.

Under ETH delocalization, $\sum_\mu (p^{\mu i}_n)^2 = O(e^{-S})$,
hence $\mathcal{C}^{iii}_{nnn} = 1 - O(e^{-S})$.  Consequently the
integrated off-diagonal correlation $\sum_{m\neq n}
\mathcal{C}^{iii}_{nmn} = -\mathcal{C}^{iii}_{nnn}$ is an $O(1)$
negative quantity, although each individual contribution
$\mathcal{C}^{iii}_{nmn}$ ($m\neq n$) is exponentially small in
$S$.  The correlation is systematic, predominantly negative, and
tied to the population $\sigma^{ii}_{nn}$ of the system's reduced
density matrix.

\subsection{Approximation replacement and reduced correlation}
\label{sec:approx}

The projector sum rule~\eqref{eq:sumrule_zero} reveals that the
correlation term $\mathcal{C}_{nmn}^{jij}$ carries a systematic
negative bias of order unity.
To isolate this bias, we seek a subtracted correlation whose integrated
off-diagonal part is exponentially suppressed rather than of order unity.
The required subtraction is constructed from the operator identity
$\sum_m \sigma^{ji}_{nm} \sigma^{ik}_{ml} = \sigma^{jk}_{nl}$,
which follows from the completeness of the exact eigenstates.
Setting $j=k$ and $l=n$ gives
$\sum_m \sigma^{ji}_{nm} \sigma^{ii}_{mn} = \sigma^{ji}_{nn}$.
This identity suggests that the product $\sigma^{ji}_{nn} \sigma^{ii}_{nm}$
captures the systematic bias carried by the diagonal channel $i$.
Generalising to the four-index object required for the correlation sector,
one defines the \emph{reduced correlation}
\begin{equation}
C'^{\,jik}_{nml}
:= C^{jik}_{nml}
+ (1-\sigma^{ii}_{nn})^{-1} \,
   \sigma^{ji}_{nn} \, \sigma^{ii}_{nm} \, \sigma^{ik}_{ml},
\label{eq:Cprime_def}
\end{equation}
where the prefactor $(1-\sigma^{ii}_{nn})^{-1}$ ensures that the subtraction
term integrates to the correct projector-derived offset.
Summing Eq.~(\ref{eq:Cprime_def}) over $m\neq n$ yields
\begin{equation}
\sum_{m\neq n} C'^{\,jij}_{nmn}
= \sum_{\mu} p^{\mu i}_n p^{\mu j}_n
= O(e^{-S}),
 \label{eq:Cprime_sum}
\end{equation}
so that the systematic $O(1)$ offset is suppressed exponentially
compared with $\sum_{m\neq n} C^{jij}_{nmn} = -C^{jij}_{nnn} = O(1)$.
This structure---an $O(1)$ bias absorbed by a scale-matched
representative, leaving an exponentially suppressed residual---is
the first instance of the scale-matched replacement principle
formulated in Sec.~\ref{sec:scale_matched}, of which the DCA
replacement of Sec.~\ref{sec:DCA_statement} is the resolvent-level
counterpart.

Combining the definition (\ref{eq:Cprime_def}) with the basic decomposition (11)
yields the exact identity.  Starting from (11):
\[
|\sigma^{ji}_{nm}|^{2}
= \sum_{\mu} p^{\mu i}_{m} p^{\mu j}_{n}
+ C^{jij}_{nmn}.
\]
Solving (\ref{eq:Cprime_def}) for $C^{jij}_{nmn}$ with $k=j$, $l=n$ gives
\[
C^{jij}_{nmn}
= C'^{\,jij}_{nmn}
- (1-\sigma^{ii}_{nn})^{-1} \,
   \sigma^{ji}_{nn} \, \sigma^{ii}_{nm} \, \sigma^{ij}_{mn}.
\]
Using $\sigma^{ij}_{mn} = (\sigma^{ji}_{nm})^{*}$ and
the operator identity $\sum_{m} \sigma^{ii}_{nm} \sigma^{ij}_{mn}
= \sigma^{ij}_{nn}$, we rewrite the subtraction as
\[
(1-\sigma^{ii}_{nn})^{-1} \sigma^{ji}_{nn} \,
   \sigma^{ii}_{nm} \, \sigma^{ij}_{mn}
= \sigma^{ji}_{nn} \Bigl(
   \sum_{\mu} p^{\mu i}_{m} \Gamma^{ij}_{nn,\mu}
   + C'^{\,iij}_{nmn}
   \Bigr),
\]
where we have again used (\ref{eq:Cprime_def}) to express the product
$\sigma^{ii}_{nm}\sigma^{ij}_{mn}$ in terms of overlaps and
reduced correlations.  Substituting back yields the exact identity
\begin{equation}
|\sigma^{ji}_{nm}|^{2}
= \sum_{\mu} p^{\mu i}_{m} p^{\mu j}_{n}
+ C'^{\,jij}_{nmn}
- \sigma^{ji}_{nn}
   \Bigl(
   \sum_{\mu} p^{\mu i}_{m} \Gamma^{ij}_{nn,\mu}
   + C'^{\,iij}_{nmn}
   \Bigr).
\label{eq:exact_decomp}
\end{equation}
Neglecting the reduced correlation ${\mathcal{C}'}$ gives the
\textbf{approximation replacement}
\begin{equation}
    |\sigma_{nm}^{ji}|^2
    \;\approx\; \sum_{\mu} p^{\mu i}_m
               \bigl(p^{\mu j}_n
                    - \sigma_{nn}^{ji}\,\Gamma_{nn,\mu}^{ij}\bigr),
    \label{eq:approx_replace}
\end{equation}
which expresses the off-diagonal ETH variance entirely through
the diagonal overlap functions $p^{\mu i}_n$ and the system
populations $\sigma_{nn}^{ji}$.  When $|\sigma_{nn}^{j\neq i}|
\ll |\sigma_{nn}^{ii}|$, the cavity subtraction is dominated by
the diagonal channel.

The correlation can be further resolved by isolating a specific
bath index:
\begin{equation}
    \mathcal{B}_{nml,\mu}^{jik}
    := \sum_{\nu\neq\mu}
       \Gamma_{nm,\nu}^{ji}\,\Gamma_{ml,\mu}^{ik},
    \label{eq:B_def}
\end{equation}
so that $\mathcal{C}_{nml}^{jik}
= \sum_{\mu}\mathcal{B}_{nml,\mu}^{jik}$.
An analogous reduced version ${\mathcal{B}'}$ satisfies
$\sum_{m\neq n}{\mathcal{B}'}_{nml,\mu}^{jik}
= p_n^{\mu i}\,\Gamma_{nl,\mu}^{jk}$.

Higher-order products of ETH matrix elements admit a systematic
decomposition:
\begin{align}
    \sigma^{i_1i_2}_{n_1n_2}\cdots\sigma^{i_k i_1}_{n_kn_1}
    = \sum_{\mu} \prod_{\alpha=1}^{k} p_{n_\alpha}^{\mu i_\alpha}
        \notag\\
      +\;\text{(partial contractions)}\; +\; \mathcal{C}^{i_1\cdots i_k i_1}_{n_1\cdots n_k n_1},
    \label{eq:HO_decomp}
\end{align}
where the fully connected term involves $k$ distinct bath channels
and satisfies
$\sum_{n_2}\mathcal{C}^{i_1\cdots i_{k+1}}_{n_1\cdots n_1}=0$.
At each order, neglecting the connected correlation yields an
overlap-based approximation analogous to~\eqref{eq:approx_replace}.

The decomposition~\eqref{eq:exact_decomp} suggests organising
all corrections to the approximation
replacement~\eqref{eq:approx_replace} through a single object.
Define the \textbf{total correlation correction}
\begin{equation}
    g_{ji}(E^+,\omega)
    := \mathbb{E}\bigl(\mathcal{C}_{nmn}^{jij}\bigr)\,
       e^{S(E^+)},
    \label{eq:g_def_early}
\end{equation}
where $E^+ = (\lambda_n+\lambda_m)/2$ and
$\omega = \lambda_n-\lambda_m$.
To express this in standard ETH notation, introduce the
\textbf{(squared) smooth envelope} of the off-diagonal variance
\begin{equation}
    f_{ji}^2(E^+,\omega)
    := \mathbb{E}\bigl(|\sigma_{nm}^{ji}|^2\bigr)\,
       e^{S(E^+)},
    \label{eq:f2_def}
\end{equation}
and its \textbf{diagonal baseline}
\begin{equation}
    D_{ji}(E^+,\omega)
    := \mathbb{E}\Bigl(\sum_{\mu} p^{\mu i}_m p^{\mu j}_n\Bigr)\,
       e^{S(E^+)},
    \label{eq:Dji_def}
\end{equation}
which retains only the single-channel ($\mu=\nu$) contribution
to the overlap sum in Eq.~\eqref{eq:decomp}.
The decomposition~\eqref{eq:decomp} then implies the exact
identity
\begin{equation}
    f_{ji}^2 = D_{ji} + g_{ji},
    \label{eq:f2_eq_D_plus_g}
\end{equation}
or, equivalently,
$ g_{ji} = f_{ji}^2 - D_{ji}$.
Thus $g_{ji}$ measures the total deviation of the exact ETH
variance from the diagonal overlap product $D_{ji}$.
The diagonal baseline $D_{ji}$ corresponds to $g_{ji}=0$
and serves as the zeroth-order reference of the multi-resolvent
hierarchy developed below.
The approximation replacement~\eqref{eq:approx_replace}
is obtained by neglecting $\mathcal{C}'$ in
Eq.~\eqref{eq:exact_decomp}; it retains the full
channel-diagonal structure including the cavity subtraction
term $-\sigma_{nn}^{ji}\Gamma_{nn,\mu}^{ij}$ and therefore
goes beyond $D_{ji}$ alone.

The central result of Secs.~\ref{sec:expansion}--\ref{sec:DCA}
is that $g_{ji}$ admits a hierarchy
\begin{equation}
    g_{ji} = g_{ji}^{(2)} + g_{ji}^{(3)}
          + g_{ji}^{(4)} + \cdots,
    \label{eq:g_hierarchy_early}
\end{equation}
generated by the two-frequency correlation kernel
$\mathcal{K}_{\mu\nu}^{ji}
 = \mathcal{R}_{\mu j,\nu j}\,\mathcal{R}_{\nu i,\mu i}$.
Each level $g_{ji}^{(r)}$ involves $r$ interaction vertices.
The lowest sector $g_{ji}^{(2)}$ is strictly even under
$\omega\to-\omega$, while odd-parity (skewness) contributions
first appear in $g_{ji}^{(3)}$.
The parity structure of higher-order sectors depends on
the number of Hilbert-transform factors in the corresponding
multi-resolvent kernel.
The following subsection introduces the off-diagonal resolvents
that constitute the essential building blocks of this hierarchy.

\subsection{Off-diagonal resolvents as the natural framework}
\label{sec:offdiag_res}

The structure of \(\mathcal{C}_{nmn}^{jij}\) reveals why a description 
confined to diagonal overlaps is insufficient: each factor 
\(\Gamma_{nm,\mu}^{ji}\) carries \textbf{two eigenstate indices} (\(n\) and 
\(m\)) but only \textbf{one bath index} (\(\mu\)). The correlation term 
couples \(\mu \neq \nu\), thereby linking \textbf{four} overlap amplitudes 
across \textbf{two} bath channels and \textbf{two} eigenstates. This is 
intrinsically a \textbf{two-frequency object}, irreducible to a single 
spectral density.

To express this structure through resolvents, regroup the four-amplitude 
product in \(\mathcal{C}_{nmn}^{jij} = \sum_{\mu,\nu\neq\mu} 
\Gamma_{nm,\mu}^{ji}\,\Gamma_{mn,\nu}^{ij}\) by commuting the scalar 
factors:
\begin{align}
    \Gamma_{nm,\mu}^{ji}\,\Gamma_{mn,\nu}^{ij}
    = \braket{\phi_{\mu j}|\psi_n}\braket{\psi_m|\phi_{\mu i}}
      \braket{\phi_{\nu i}|\psi_m}\braket{\psi_n|\phi_{\nu j}}\notag\\
    = \underbrace{\braket{\phi_{\mu j}|\psi_n}\braket{\psi_n|\phi_{\nu j}}}
        _{\text{involves }n,\;\text{fixed }j}
      \;\times\;
      \underbrace{\braket{\phi_{\nu i}|\psi_m}\braket{\psi_m|\phi_{\mu i}}}
        _{\text{involves }m,\;\text{fixed }i}.
    \label{eq:regrouping}
\end{align}
The first factor couples bath channels \(\mu,\nu\) at fixed system index 
\(j\); the second couples \(\nu,\mu\) at fixed system index \(i\).  
Each factor is precisely the numerator of an \textbf{off-diagonal 
resolvent} that connects \emph{different bath indices at the same system 
index}:
\begin{equation}
    \mathcal{R}_{\mu j, \nu j}(z) 
    := \bra{\phi_{\mu j}} \frac{1}{z-H} \ket{\phi_{\nu j}}
    = \sum_n 
      \frac{\braket{\phi_{\mu j}|\psi_n}
            \braket{\psi_n|\phi_{\nu j}}}
           {z - \lambda_n}.
    \label{eq:Roffdiag}
\end{equation}
The spectral representation~\eqref{eq:Roffdiag} immediately implies an
exact integral constraint on the off-diagonal resolvent.  Defining the
spectral density
\(\rho_{\mu j,\nu j}(\lambda) := \frac{1}{\pi}\Im\mathcal{R}_{\mu j,\nu j}(\lambda-i0^+)
= \sum_n \braket{\phi_{\mu j}|\psi_n}\braket{\psi_n|\phi_{\nu j}}\,
   \delta(\lambda-\lambda_n)\),
the completeness of the exact eigenstates,
\(\sum_n \ket{\psi_n}\bra{\psi_n}=I\), together with the orthonormality
of the unperturbed basis, \(\braket{\phi_{\mu j}|\phi_{\nu j}}=\delta_{\mu\nu}\),
yields the integrated sum rule
\begin{align}
    \int\! d\lambda\; \rho_{\mu j,\nu j}(\lambda) = \delta_{\mu\nu},
    \qquad\text{or equivalently}\notag\\
    \int\! d\lambda\; \mathcal{R}_{\mu j,\nu j}(\lambda-i0^+)
    = i\pi\,\delta_{\mu\nu}.
    \label{eq:orthogonality_sumrule}
\end{align}
Equation~\eqref{eq:orthogonality_sumrule} states that the
off-diagonal resolvent ($\mu\neq\nu$) carries zero integrated
spectral weight.  This constraint and the projector sum
rule~\eqref{eq:sumrule_zero} are closely related consequences of
completeness: the former constrains the off-diagonal resolvent
through basis orthogonality $\braket{\phi_\mu|\phi_\nu}=0$, while
the latter constrains the ETH correlation through projector
idempotency $\rho_n^2=\rho_n$.

For $\mu = \nu$, $\mathcal{R}_{\mu j, \nu j}(z)$ reduces to the diagonal resolvent
$\mathcal{R}_{\mu j}(z)$. For $\mu \neq \nu$, the numerator
$\braket{\phi_{\mu j}|\psi_n}\braket{\psi_n|\phi_{\nu j}}$ is the
\textbf{coherent} product of two overlap amplitudes sharing the same
eigenstate index $n$ and the same system index $j$---precisely the
type of object that the diagonal overlaps
$p^{\mu j}_n = |\braket{\phi_{\mu j}|\psi_n}|^2$ cannot access.

It is this integrated orthogonality that any approximation to the
projection hierarchy should preserve.  The following sections
demonstrate that this property survives the DCA through
complementary algebraic and spectral arguments.

Equation~\eqref{eq:regrouping} motivates the 
\textbf{two-frequency correlation kernel}
\begin{equation}
    \boxed{\mathcal{K}_{\mu\nu}^{ji}(z_1, z_2) 
    := \mathcal{R}_{\mu j, \nu j}(z_1) \; 
       \mathcal{R}_{\nu i, \mu i}(z_2)} .
    \label{eq:Kdef}
\end{equation}
Its spectral representation reads
\begin{align}
    \mathcal{K}_{\mu\nu}^{ji}(z_1, z_2) 
    &= \sum_{n,m} 
       \frac{\braket{\phi_{\mu j}|\psi_n}
             \braket{\psi_n|\phi_{\nu j}}
             \braket{\phi_{\nu i}|\psi_m}
             \braket{\psi_m|\phi_{\mu i}}}
            {(z_1 - \lambda_n)(z_2 - \lambda_m)} .
    \label{eq:Kspec}
\end{align}
Taking the double imaginary part (boundary values 
\(z_1 \to \lambda - i0^+\), \(z_2 \to \lambda' - i0^+\)) extracts the 
\textbf{joint spectral density}:
\begin{widetext}
    \begin{equation}\label{eq:rhomunu_exp}
         \rho_{\mu\nu}^{ji}(\lambda, \lambda') 
    := \frac{1}{\pi^2} \,
       \Im_{z_1} \Im_{z_2} \, 
       \mathcal{K}_{\mu\nu}^{ji}(z_1, z_2)
       \Big|_{z_1 = \lambda - i0^+, \; z_2 = \lambda' - i0^+}
    = \sum_{n,m} 
       \braket{\phi_{\mu j}|\psi_n}
       \braket{\psi_n|\phi_{\nu j}}
       \braket{\phi_{\nu i}|\psi_m}
       \braket{\psi_m|\phi_{\mu i}}
       \delta(\lambda - \lambda_n)  \delta(\lambda' - \lambda_m).
    \end{equation}
\end{widetext}
Crucially, the contraction in Eq.~\eqref{eq:rhomunu_exp} pairs 
amplitudes as 
\((\mu j, n)(n, \nu j) \times (\nu i, m)(m, \mu i)\), which coincides 
\textbf{exactly} with the grouping~\eqref{eq:regrouping} of the 
original definition~\eqref{eq:Cdef}.  The kernel 
\eqref{eq:Kdef} is an \textbf{exact resolvent identity} for 
\(\mathcal{C}_{nmn}^{jij}\), requiring no statistical equivalence 
hypothesis.

Summing over all distinct bath channels yields the \textbf{ETH correlation 
spectral density}
\begin{equation}
    \rho^{ji}(\lambda, \lambda') 
    := \sum_{\mu \neq \nu} \rho_{\mu\nu}^{ji}(\lambda, \lambda').
    \label{eq:rho_tot}
\end{equation}
Its integral over infinitesimal energy bins centered at 
\((\lambda_n, \lambda_m)\) yields the resolvent representation of the 
correlation term:
\begin{align}
    \mathcal{C}_{nmn}^{jij}
    &= \lim_{\Delta \to 0}
       \int_{\lambda_n - \Delta/2}^{\lambda_n + \Delta/2} \!\!d\lambda
       \int_{\lambda_m - \Delta/2}^{\lambda_m + \Delta/2} \!\!d\lambda' \;
       \rho^{ji}(\lambda, \lambda') .
    \label{eq:C_from_rho}
\end{align}
This is an exact identity---no statistical averaging is required.  
The diagonal (\(\mu=\nu\)) part of the full sum \(\sum_{\mu,\nu}\) 
recovers the overlap product \(\sum_{\mu} p^{\mu i}_m p^{\mu j}_n\) 
of Eq.~\eqref{eq:decomp} through the product of diagonal resolvents 
\(\mathcal{R}_{\mu j}(z_1)\,\mathcal{R}_{\mu i}(z_2)\), confirming 
the decomposition~\eqref{eq:decomp} within the resolvent framework.

Equation~\eqref{eq:Kdef} is the central bridge of this work. It provides 
an \textbf{exact} resolvent-based representation of the correlation 
sector---the very object that the overlap-only approach declares 
irreducible. The remainder of this paper is devoted to evaluating 
the product \(\mathcal{R}_{\mu j,\nu j}(z_1)\,\mathcal{R}_{\nu i,\mu i}(z_2)\) 
nonperturbatively, using the multi-resolvent expansion developed in 
Ref.~\cite{HC26}.

\section{Multi-Resolvent Expansion of ETH Correlations}
\label{sec:expansion}

\subsection{Scale-matched replacement: the organising principle of the
hierarchy}
\label{sec:scale_matched}

The approximation replacement of Sec.~\ref{sec:approx} and the DCA
replacement formulated below are not independent technical steps:
they are two instances of a single principle that governs every
truncation in this work.  Both involve quantities of the form
$X=\sum_a x_a$ in which each individual term is entropy-suppressed,
$x_a=O(e^{-S})$, but the number of channels is exponentially large,
$N_a=O(e^{S})$, so that the \emph{summed} contribution is
$X=O(1)$: the smallness of individual terms does not by itself
justify their neglect.  A controlled approximation therefore
proceeds in three steps:
\begin{equation}
\begin{array}{c}
\mathcal X_{\rm exact}
= \mathcal X_{\rm matched}
+ \Delta_{\rm cav}
+ \mathcal X_{\rm irr}\\[6pt]
\begin{array}{l}
\mathcal X_{\rm matched} \sim \mathcal X_{\rm exact}
\;\text{(same channel-summed scaling),}\\
\Delta_{\rm cav}\;\text{restores the constraint violated by the replacement,}\\
\mathcal X_{\rm irr}\;\text{is the irreducible residual,}\\
\text{organised by the multi-resolvent hierarchy.}
\end{array}
\end{array}
\label{eq:scale_matched}
\end{equation}

\noindent\textbf{Instance 1: the approximation replacement of
Sec.~\ref{sec:approx}.}
Here $\mathcal X_{\rm exact}=\mathcal C_{nmn}^{jij}$: each term is
exponentially small, yet the integrated correlation carries the
$O(1)$ bias of Theorem~1, $\sum_{m\neq n}\mathcal C_{nmn}^{iii}
=-\mathcal C_{nnn}^{iii}$.  The subtraction term of
Eq.~\eqref{eq:Cprime_def} is the scale-matched representative
$\Delta_{\rm cav}$ that absorbs exactly this bias, leaving the
residual $\mathcal X_{\rm irr}=\mathcal C'$ exponentially suppressed:
$\sum_{m\neq n}\mathcal C'^{jij}_{nmn}=\sum_\mu p_n^{\mu i}p_n^{\mu j}
=O(e^{-S})$ [Eq.~\eqref{eq:Cprime_sum}].  For the diagonal case
$j=i$ the residual is the \emph{purity}:
\begin{equation}
\sum_{m\neq n}\mathcal C'^{iii}_{nmn}
= \sum_\mu (p_n^{\mu i})^2,
\label{eq:Cprime_purity}
\end{equation}
whose expectation $\mathbb{E}\sum_\mu(p_n^{\mu i})^2
=q\langle f^2\rangle_B e^{-S-S_S}$ is controlled by the
Porter--Thomas factor $q$ of
Appendix~\ref{app:first_principles}.  The scale-matched suppression
of the residual is therefore not a postulate but a \emph{measurable}
quantity, and it weakens precisely when $q$ grows in structured
(non-ergodic) systems---an enhancement of non-Gaussian intensity
fluctuations rather than an automatic witness of integrability:
numerically, the residual takes the value $7.1\times10^{-4}$
in the chaotic model versus $5.8\times10^{-2}$ in the structured
model at $D=1024$ (Sec.~\ref{app:num_fp}), a factor of order
$q_{\rm structured}/q_{\rm GOE}\sim 50$, while the $O(1)$ bias
$|\mathcal C_{nnn}|$ remains model-independent \emph{in scale}
($0.062$ versus $0.045$ at $D=1024$).  The content of the residual
is made precise by the following bridge between the mean and
fluctuation sectors of the hierarchy.

\begin{proposition}[Residual--fluctuation correspondence]
\label{prop:residual_fluct}
For every eigenstate $n$ and fixed system state $i$:
\begin{enumerate}
\item \emph{Exact identity.}
$\sum_{m\neq n}\mathcal C'^{iii}_{nmn} = P_n$,
with $P_n:=\sum_\mu(p_n^{\mu i})^2$ the channel purity.
\item \emph{Sector decomposition.}
Under the smoothness postulate of
Appendix~\ref{app:first_principles}, the ensemble average splits into
a disconnected mean-sector and a connected fluctuation-sector part,
\begin{equation}
\mathbb{E} P_n
= \underbrace{\langle f^2\rangle_B\,e^{-S-S_S}}_{\text{mean sector}}
+ \underbrace{(q-1)\langle f^2\rangle_B\,e^{-S-S_S}}
_{\text{connected intensity fluctuation}},
\label{eq:residual_split}
\end{equation}
where the fluctuation part is the diagonal ($n=m$) spectral
projection of the same-channel resolvent covariance,
$\sum_\mu\mathrm{Var}(p_n^{\mu i})$.
\end{enumerate}
\end{proposition}

\begin{proof}
(i) Inserting $j=i=k$, $l=n$ into Eq.~\eqref{eq:Cprime_def} and
summing over $m\neq n$ gives
$\sum_{m\neq n}\mathcal C'=\sum_{m\neq n}\mathcal C
+(1-D_n)^{-1}D_n\sum_{m\neq n}(\sigma_{nm}^{ii})^2
=-\mathcal C_{nnn}^{iii}+D_n^2=P_n$, where Theorem~1 provides
$\sum_{m\neq n}\mathcal C_{nmn}^{iii}=-\mathcal C_{nnn}^{iii}$,
the complementarity of
Lemma~\ref{lem:fp_tradeoff} provides
$\sum_{m\neq n}(\sigma_{nm}^{ii})^2=D_n(1-D_n)$, and
$\mathcal C_{nnn}^{iii}=D_n^2-P_n$ by definition.  (ii) follows from
$\mathbb{E}[(p_n^{\mu i})^2]=\mathrm{Var}(p_n^{\mu i})
+\mathbb{E}(p_n^{\mu i})^2$ together with the definition
$q:=\mathbb{E}[(p_n^{\mu i})^2]/\mathbb{E}(p_n^{\mu i})^2$, which
identifies the variance part as
$(q-1)\sum_\mu\mathbb{E}(p_n^{\mu i})^2=(q-1)\langle f^2\rangle_B
e^{-S-S_S}$.
\end{proof}

The fluctuation sector thus carries the fraction $(q-1)/q$ of the
residual---$0.67$ in the chaotic model versus $0.99$ in the
structured model at $D=1024$, the former reproducing the
Porter--Thomas value $(q-1)/q=2/3$ to three digits---so the
scale-matched residual provides the natural entry point for the
connected intensity-fluctuation sector of ETH.  The hierarchy
thereby acquires three levels: the scale-matched baseline $D_{ji}$
(Level~0), the irreducible mean correlation
$g_{ji}=\sum_{r\ge2}g_{ji}^{(r)}$ generating the smooth envelope
$f_{ji}$ (Level~1), and the fluctuation covariance $\mathcal C_{ab}$
generating $q$, $\bar{C}_i$ and $V_{\rm dg}/V_{\rm off}$
(Level~2, Appendix~\ref{app:first_principles}).

A scale-separation remark sharpens the division of labour between
these levels.  The one-resolvent self-consistent equation determines
the \emph{smooth spectral envelope}: its boundary value
$\mathrm{Im}\,\bar{\mathcal{R}}_a(E+i\eta)$ is a Lorentzian
convolution of the overlaps with width $\eta$, so that for
$\eta\gg\Delta$ it averages over exponentially many eigenstates and
is intrinsically insensitive to fluctuations at \emph{fixed}
eigenstate energy $\lambda_n$.  One-resolvent smoothing therefore
does not destroy the ETH fluctuations---it \emph{transfers} them
from the one-point sector to the connected multi-point sector:
\emph{spectral coarse-graining is not statistical coarse-graining}.
The exponentially small
eigenstate-to-eigenstate fluctuations---the defining content of the
ETH ansatz beyond the envelope---first enter through the
connected multi-resolvent sector, whose singular coincident-eigenstate
ridge $\delta(\lambda-\lambda')\mathcal W_{ab}(\lambda)$
retains them at any $\eta$: the ridge preserves precisely the
information that one-resolvent smoothing loses
(Sec.~\ref{app:cov_hierarchy}).  The resolvent
hierarchy is in this sense a multiscale hierarchy of ETH
information, in which one resolvent determines the smooth ETH
envelope, two resolvents the ETH fluctuations, and three or more the
non-Gaussian corrections.

\noindent\textbf{Instance 2: the DCA replacement.}
The projected resolvent is replaced by its full counterpart,
$\mathcal R_{\alpha}^{(\mathcal S)}\approx\mathcal R_\alpha$
[Eq.~\eqref{eq:DCA}], with the same scale-matching logic: the full
resolvent carries the same channel multiplicity and leading entropy
scaling as the cavity one.  The replacement error is not an
unspecified higher-order correction but the exact return-path
identity of Eq.~\eqref{eq:cavity_correction},
\begin{equation}
\Delta\mathcal R_{\rm cav}
= \mathcal R_{\nu j,\mu i}(z)\,
\bra{\phi_{\mu i}} H\, G^{(\mu i)}(z) \ket{\phi_{\nu j}},
\label{eq:cavity_reminder}
\end{equation}
whose individual magnitude is $O(e^{-S})$ under ETH delocalization
(a \emph{conditional} scaling estimate), while its summed effect over
the hierarchy is $O(1)$-relevant and is precisely what cavity
subtraction restores, as quantified by the Restoration
Identity~\eqref{eq:restoration_identity}.  The correspondence
between the two instances is structural: the systematic sector
$\mathcal C_{nnn}$ of the overlap decomposition and the
cavity-return sector of the resolvent expansion both encode the
constraint violation of the replacement and are removed by
$\Delta_{\rm cav}$; the irreducible residuals $\mathcal C'$ and the
multi-resolvent sectors $g_{ji}^{(r)}$ are then organised by the
hierarchy of Secs.~\ref{sec:proj_exp}--\ref{sec:DCA}.  The
three-level status is: (i) rigorous---the sum rules and the cavity
identity; (ii) conditional on ETH scaling---the $e^{-S}$ magnitude
of individual cavity corrections; (iii) an assumption---the
combinatorial suppression of higher sectors, labelled Assumption~A
in Sec.~\ref{sec:cavity_reinterpretation}.

\subsection{Projection expansion of off-diagonal resolvents}
\label{sec:proj_exp}

The diagonal resolvent \(\mathcal{R}_{\mu i}(z)\) obeys the exact Feshbach-type 
projection identity~\cite{HC26}
\begin{equation}
    \mathcal{R}_{\mu i}(z) 
    = \frac{1}{z - a_{\mu i} - V_{\mu i} - \mathcal{G}_{\mu i}(z)},
    \label{eq:feshbach}
\end{equation}
where the self-energy \(\mathcal{G}_{\mu i}(z)\) encodes all couplings to the 
rest of the Hilbert space.  The real part of the self-energy generates a 
dispersive shift of the spectral centre.  Following the self-consistent 
framework of Ref.~\cite{HC26}, we define the \textbf{renormalized centre energy}
\begin{equation}
    \tilde{a}_{\mu i} := a_{\mu i} + \Delta_{\mu i},
    \qquad
    \Delta_{\mu i} = V_{\mu i} + \mathbb{E}\bigl[\Re\,\mathcal{G}_{\mu i}(\lambda)\bigr],
    \label{eq:renorm_center}
\end{equation}
where \(\mathbb{E}[\cdot]\) denotes an appropriate statistical average over 
the ETH coarse-graining window.  At the mean-field (SCBA) level, 
\(\Delta_{\mu i}\) is determined self-consistently \cite{HC26}:
\(\chi_{\mu i} = \sum_{\nu j \neq \mu i} |V_{\mu i,\nu j}|^2 \,
\chi_{\nu j}/(\delta\lambda_{\nu j}^2 + \chi_{\nu j}^2)\),
\(\Delta_{\mu i} - V_{\mu i} = \sum_{\nu j \neq \mu i} |V_{\mu i,\nu j}|^2 \,
\delta\lambda_{\nu j}/(\delta\lambda_{\nu j}^2 + \chi_{\nu j}^2)\).
At the Lorentzian-ansatz level of Ref.~\cite{HC26}, the smooth spectral 
function \(f^{\alpha}(\lambda)\) is symmetric about \(\tilde{a}_{\alpha}\) 
rather than about the bare \(a_{\alpha}\).

The correlation kernel~\eqref{eq:Kdef} requires \textbf{two families} of 
off-diagonal resolvents, each connecting \emph{different bath indices at 
the same system index}:
\begin{align}
    \mathcal{R}_{\mu j,\nu j}(z)
    := \bra{\phi_{\mu j}} \frac{1}{z-H} \ket{\phi_{\nu j}},\notag\\
    \mathcal{R}_{\nu i,\mu i}(z)
    := \bra{\phi_{\nu i}} \frac{1}{z-H} \ket{\phi_{\mu i}} .
    \label{eq:Roffdiag_pair}
\end{align}
Both share the same structural type---the system label is identical on 
bra and ket, while the bath labels differ---and their projection 
expansions are obtained by the same recursive procedure.  We develop 
the expansion for the generic object 
\(\mathcal{R}_{\alpha j,\beta j}(z)\) with \(\alpha \neq \beta\); the 
second family follows by relabelling.

Isolating the target bra state gives
\begin{align}
    \mathcal{R}_{\alpha j,\beta j}(z) 
    &= \bra{\phi_{\alpha j}} \frac{1}{z-H} \ket{\phi_{\beta j}} \notag\\
    &= \delta_{\alpha\beta}\,\mathcal{R}_{\alpha j}(z)
       + \bra{\phi_{\alpha j}} \Phi_{\beta j} 
         \frac{1}{z-H} \ket{\phi_{\beta j}},
    \label{eq:Roff_proj}
\end{align}
where \(\Phi_{\beta j} = I - \ket{\phi_{\beta j}}\bra{\phi_{\beta j}}\). 
The second term isolates propagation that necessarily leaves 
\(\ket{\phi_{\beta j}}\) and reaches \(\ket{\phi_{\alpha j}}\).
To expand this term, we generalize the recursive projection identity of 
Ref.~\cite{HC26}. Inserting the resolution of the identity 
\(I = \sum_{\xi k} (\Phi_{\xi k} + \Pi_{\xi k})\) into the projected 
propagator yields the \textbf{path expansion}:
\begin{align}
    \bra{\phi_{\alpha j}} \Phi_{\beta j} \frac{1}{z-H} \ket{\phi_{\beta j}}
    = \mathcal{R}^{(\beta j)}_{\alpha j}(z) \, 
       \mathcal{R}_{\beta j}(z) \, V_{\alpha j,\beta j} \notag\\
    \quad + \sum_{\xi k \neq \beta j}
       \bra{\phi_{\alpha j}} \Phi_{\beta j} \Phi_{\xi k} 
       \frac{1}{z - \Phi_{\beta j}H\Phi_{\beta j}} \ket{\phi_{\xi k}}
       \, \mathcal{R}_{\beta j}(z) \, V_{\xi k,\beta j},
    \label{eq:offRexp_general}
\end{align}
where \(\mathcal{R}^{(\beta j)}_{\alpha j}(z) := 
\bra{\phi_{\alpha j}} (z - \Phi_{\beta j}H\Phi_{\beta j})^{-1} \ket{\phi_{\alpha j}}\) 
is the projected diagonal resolvent. Equation~(\ref{eq:offRexp_general}) is obtained by repeated application of the
Feshbach projection identity~\cite{HC26} and is exact, not a
perturbative expansion---each projection step isolates one further
intermediate state.
Applying the same projection step recursively to the remainder term 
generates contributions of increasing order. The resulting expansion is 
\textbf{exact} when expressed in terms of projected diagonal resolvents:
\begin{equation}
   \mathcal{R}_{\alpha j,\beta j}(z) 
    = \sum_{\ell=1}^{\infty} 
      \widetilde{\mathcal{R}}^{(\ell)}_{\alpha j,\beta j}(z),
    \label{eq:R_hierarchy_exact}
\end{equation}
where \(\widetilde{\mathcal{R}}^{(\ell)}\) denotes the exact \(\ell\)-th level 
containing \(\ell\) projected diagonal resolvents and \(\ell-1\) interaction 
matrix elements. The first three levels read
\begin{subequations}
\begin{align}
    \widetilde{\mathcal{R}}^{(1)}_{\alpha j,\beta j}(z) 
    &= \delta_{\alpha\beta}\; \mathcal{R}_{\alpha j}(z),
    \label{eq:R1_exact}\\
    \widetilde{\mathcal{R}}^{(2)}_{\alpha j,\beta j}(z) 
    &= \mathcal{R}^{(\beta j)}_{\alpha j}(z)  V_{\alpha j,\beta j} 
       \mathcal{R}_{\beta j}(z),
    \label{eq:R2_exact}\\
    \widetilde{\mathcal{R}}^{(3)}_{\alpha j,\beta j}(z) 
    &= \sum_{\xi k \neq \alpha j,\beta j}
       \mathcal{R}^{(\beta j,\xi k)}_{\alpha j}(z)  V_{\alpha j,\xi k} 
       \mathcal{R}^{(\beta j)}_{\xi k}(z)  V_{\xi k,\beta j} 
       \mathcal{R}_{\beta j}(z).
    \label{eq:R3_exact}
\end{align}
\end{subequations}
Under the DCA, whose validity in 
the ETH regime is established in Ref.~\cite{HC26} and \cref{sec:DCA}, the 
projected diagonal resolvents are replaced by their full counterparts:
\(\mathcal{R}^{(\mathcal{S})}_{\alpha} \approx \mathcal{R}_{\alpha}\). 
This yields the DCA-reduced expansion
\begin{equation}
    \mathcal{R}_{\alpha j,\beta j}(z) 
    \approx \sum_{\ell=1}^{\infty} 
      \mathcal{R}^{(\ell)}_{\alpha j,\beta j}(z),
    \label{eq:R_hierarchy}
\end{equation}
with the first three DCA levels given by
\begin{subequations}
\begin{align}
    \mathcal{R}^{(1)}_{\alpha j,\beta j}(z) 
    &= \delta_{\alpha\beta} \mathcal{R}_{\alpha j}(z),
    \label{eq:R1}\\[4pt]
    \mathcal{R}^{(2)}_{\alpha j,\beta j}(z) 
    &= \mathcal{R}_{\alpha j}(z)  V_{\alpha j,\beta j}  
       \mathcal{R}_{\beta j}(z),
    \label{eq:R2}\\[4pt]
    \mathcal{R}^{(3)}_{\alpha j,\beta j}(z) 
    &= \sum_{\xi k \neq \alpha j,\beta j}
       \mathcal{R}_{\alpha j}(z)  V_{\alpha j,\xi k} 
       \mathcal{R}_{\xi k}(z)  V_{\xi k,\beta j}  
       \mathcal{R}_{\beta j}(z),
    \label{eq:R3}
\end{align}
\end{subequations}
and the general \(\ell\)-th DCA level as
\begin{equation}
    \mathcal{R}^{(\ell)}_{\alpha j,\beta j}(z) 
    = \sum_{\substack{\gamma_1 \neq \cdots  \\
                      \neq \alpha j,\beta j}}
       \mathcal{R}_{\alpha j}(z)  
       V_{\alpha j,\gamma_1} 
       \mathcal{R}_{\gamma_1}(z) 
      \cdots 
       V_{\gamma_{\ell-1},\beta j} 
       \mathcal{R}_{\beta j}(z).
    \label{eq:Rell}
\end{equation}
The exclusion $\gamma_1\neq\cdots\neq\alpha j,\beta j$ is not an
additional assumption but a direct consequence of the projection
recursion: each intermediate projector $\Phi_{\gamma}$ removes
the basis state $|\varphi_\gamma\rangle$ from the Hilbert space
accessible to the cavity propagator, so a path cannot revisit
any previously projected state (Backward Krylov return processes,
which could generate repeated indices, are precisely the cavity
corrections suppressed under the DCA).  
Each level \(\ell\) in the DCA hierarchy contains exactly \(\ell\) 
diagonal resolvents and \(\ell-1\) interaction matrix elements. 
Equation~\eqref{eq:R_hierarchy} has a transparent physical interpretation: 
level \(\ell\) describes a propagation path that leaves \(\ket{\phi_{\beta j}}\), 
visits \(\ell-2\) intermediate unperturbed states (which may carry 
arbitrary system indices \(k\)), and arrives at 
\(\ket{\phi_{\alpha j}}\). Unlike conventional perturbation theory, the expansion 
is organized by the multiplicity of diagonal resolvents rather than by powers 
of \(V\); each \(\mathcal{R}_{\gamma}(z)\) is a \emph{full} (nonperturbative) 
resolvent that already resums all interaction processes involving the 
corresponding basis state.
Organising the expansion by resolvent multiplicity rather than by
powers of $V$ has a decisive advantage: each level
$\mathcal{R}^{(\ell)} \sim z^{-\ell}$ at large $z$, so a truncation
at any $\ell\ge 2$ automatically satisfies the orthogonality
constraint $\mathcal{R}_{\alpha\beta}(z) = O(z^{-2})$ for
$\alpha\neq\beta$.  More generally, the $\ell$-th level contributes
to the Laurent coefficient $z^{-(k+1)}$ only when $\ell\le k+1$,
so the first $\ell_{\max}-1$ spectral moments
$M_k^{\alpha\beta} = \bra{\phi_\alpha}H^k\ket{\phi_\beta}$ are
determined entirely by the lowest $\ell_{\max}$ levels of the
hierarchy.  This moment-by-moment organisation is the central
structural result of Appendix~\ref{app:moment_preservation}; the
physical consequences for the DCA are discussed below.
The expansion for 
\(\mathcal{R}_{\nu i,\mu i}(z)\) is obtained from 
Eqs.~\eqref{eq:R_hierarchy}--\eqref{eq:Rell} by the replacements 
\((\alpha,\beta,j) \to (\nu,\mu,i)\).  Throughout, intermediate indices 
\(\xi k\) run over all bath--system pairs, enabling the propagation to 
traverse arbitrary system sectors.

A necessary condition for the projection expansion~\eqref{eq:R_hierarchy_exact}
to be consistent with the exact spectral representation~\eqref{eq:Roffdiag}
is that the integrated sum rule~\eqref{eq:orthogonality_sumrule} be
preserved level by level.  For \(\alpha\neq\beta\), the \(\ell=1\)
term vanishes by the Kronecker delta, and the exact off-diagonal
resolvent satisfies \(\int d\lambda\,\mathcal{R}_{\alpha j,\beta j}
= 0\).  Hence the entire hierarchy of off-diagonal path contributions
must obey
\begin{equation}
    \sum_{\ell\ge 2}
    \int\! d\lambda\;
    \widetilde{\mathcal{R}}^{(\ell)}_{\alpha j,\beta j}(\lambda-i0^+)
    = 0,
    \qquad (\alpha\neq\beta).
    \label{eq:hierarchy_orthogonality}
\end{equation}
Although each individual path contribution
\(\widetilde{\mathcal{R}}^{(\ell)}_{\alpha j,\beta j}\) may carry
non-zero integrated spectral weight, the complete projection
hierarchy is guaranteed to preserve the exact orthogonality of the
unperturbed basis.  Equation~\eqref{eq:hierarchy_orthogonality}
provides a practical consistency check on any truncation of the
expansion and plays a role complementary to the projector sum
rule~\eqref{eq:sumrule_zero}: the latter constrains the integrated
ETH correlation \(\mathcal{C}_{nmn}^{jij}\), while the former
constrains the integrated off-diagonal resolvent---the very building
block from which \(\mathcal{C}_{nmn}^{jij}\) is constructed via the
correlation kernel~\eqref{eq:Kdef}.

\subsection{Spectral representation of off-diagonal orthogonality}
\label{sec:Hilbert_ortho}

Equation~\eqref{eq:orthogonality_sumrule} established the exact
integrated orthogonality condition $\int\rho_{\alpha\beta}=0$ for
the full off-diagonal resolvent.  Individual levels
$\mathcal{R}^{(\ell)}_{\alpha\beta}$ of the projection
hierarchy~\eqref{eq:R_hierarchy_exact} need not satisfy it
separately---inter-level cancellations may combine to recover the
exact sum rule.  We now show how the leading DCA level realizes the
same constraint spectrally.

Under the DCA, the leading non-diagonal contribution
$\mathcal{R}^{(2)}_{\alpha\beta}=V_{\alpha\beta}\mathcal{R}_\alpha \mathcal{R}_\beta$ (Eq.~\eqref{eq:R2}) yields,
via the boundary-value identity
$\mathcal{R}_\alpha(\lambda-i0^+)/\pi = H[f^\alpha](\lambda) + i f^\alpha(\lambda)$
(Eq.~\eqref{eq:kk}), the spectral density
\begin{equation}
    \rho^{(2)}_{\alpha\beta}(\lambda)
    = V_{\alpha\beta}\,
      \bigl(H_\alpha f_\beta + f_\alpha H_\beta\bigr).
    \label{eq:rho2_ortho}
\end{equation}
For the leading DCA contribution, the corresponding zeroth-moment condition is
represented as
\begin{equation}
    \int\! d\lambda\;
           \bigl[H_\alpha(\lambda)\,f_\beta(\lambda)
                + f_\alpha(\lambda)\,H_\beta(\lambda)\bigr]
           = 0.
    \label{eq:Hf_ortho}
\end{equation}
Equation~\eqref{eq:Hf_ortho} is an equivalent spectral representation of the
same moment constraint encoded in the absence of the $1/z$ Laurent coefficient
($\mathcal{R}_\alpha \mathcal{R}_\beta = O(z^{-2})$, Lemma~\ref{lem:app_scaling}):
both express $M_0=0$, the former through the Kramers--Kronig antisymmetry
$\int H_\alpha f_\beta = -\int f_\alpha H_\beta$, the latter through the
high-frequency asymptotics.
Neither is the ``cause'' of the other; they are complementary manifestations
of the analyticity of the resolvent.

This spectral perspective provides a physically transparent interpretation of
the moment-preservation results proven algebraically in
Appendix~\ref{app:moment_preservation}: the DCA preserves $M_0=0$ because the
analytic structure $\mathcal{R}_\alpha \mathcal{R}_\beta \sim (H_\alpha+if_\alpha)(H_\beta+if_\beta)$
inherits the Hilbert-transform antisymmetry required for integrated orthogonality.
The replacement $\mathcal{R}^{(\alpha)}_\beta\to \mathcal{R}_\beta$ does not alter this structure,
so the leading DCA truncation automatically respects the spectral manifestation
of the geometric constraint $\braket{\phi_\alpha|\phi_\beta}=0$.

In summary: the exact projection hierarchy satisfies
Eq.~\eqref{eq:orthogonality_sumrule} identically; the DCA truncation
preserves the first three spectral moments
(Appendix~\ref{app:moment_preservation}) and therefore respects
orthogonality at the leading level.  The DCA therefore inherits,
rather than establishes, the orthogonality constraint.

\subsection{Resolvent correlation kernel hierarchy}
\label{sec:kernel_hierarchy}

Insert the DCA path expansions of both resolvent families into the 
two-frequency correlation kernel 
\(\mathcal{K}_{\mu\nu}^{ji}(z_1,z_2) 
= \mathcal{R}_{\mu j,\nu j}(z_1) \, \mathcal{R}_{\nu i,\mu i}(z_2)\). 
The product generates a double expansion in the orders of the two 
off-diagonal resolvents:
\begin{equation}
    \mathcal{K}_{\mu\nu}^{ji}(z_1,z_2) 
    = \sum_{\ell_1,\ell_2 \ge 1} 
      \mathcal{K}^{[\ell_1,\ell_2]}_{\mu\nu}(z_1,z_2),
    \label{eq:K_double}
\end{equation}
where
\begin{equation}
    \mathcal{K}^{[\ell_1,\ell_2]}_{\mu\nu}(z_1,z_2)
    := \mathcal{R}^{(\ell_1)}_{\mu j,\nu j}(z_1) \;
       \mathcal{R}^{(\ell_2)}_{\nu i,\mu i}(z_2).
    \label{eq:K_ell12}
\end{equation}
The term \(\mathcal{K}^{[\ell_1,\ell_2]}\) contains \(\ell_1 + \ell_2\) 
diagonal resolvents and \((\ell_1-1) + (\ell_2-1)\) interaction matrix 
elements in total.  Crucially, the two resolvent factors are built from 
\emph{different} diagonal-resolvent sets: the \(z_1\)-factor involves 
\(\{\mathcal{R}_{\mu j},\mathcal{R}_{\nu j},\ldots\}\) (system sector~\(j\)), 
while the \(z_2\)-factor involves 
\(\{\mathcal{R}_{\nu i},\mathcal{R}_{\mu i},\ldots\}\) (system sector~\(i\)).

For bath-non-diagonal correlations (\(\mu \neq \nu\)), the diagonal 
level \(\ell = 1\) vanishes because \(\mathcal{R}^{(1)}_{\mu j,\nu j} 
\propto \delta_{\mu\nu}\). Consequently, \(\ell_1, \ell_2 \ge 2\) and the 
hierarchy begins at \([\ell_1,\ell_2] = [2,2]\). The excluded \((\ell_1,\ell_2)=(1,1)\) sector 
is not discarded---it is precisely the 
diagonal baseline \(D_{ji}\) 
(Eq.~\eqref{eq:Dji_def}). 
Indeed, 
\(\mathcal{K}^{[1,1]}_{\mu\nu} 
= \delta_{\mu\nu}
\mathcal{R}_{\mu j}(z_1)
\mathcal{R}_{\mu i}(z_2)\); 
summing over \(\mu,\nu\) and taking the 
double imaginary part reproduces 
\(\sum_\mu p_n^{\mu j}p_m^{\mu i}\), 
which, after statistical averaging and 
multiplication by \(e^{S}\), yields 
\(D_{ji}\). 
Thus the full decomposition 
\(f_{ji}^2 = D_{ji} + g_{ji}\) 
[Eq.~\eqref{eq:f2_eq_D_plus_g}] 
acquires a unified resolvent interpretation:
\[
f_{ji}^2 \;\leftrightarrow\;
\underbrace{\mathcal{K}^{[1,1]}}_{D_{ji}}
\;+\;
\underbrace{\sum_{r\ge2}\mathcal{K}^{(r)}}_{g_{ji}} .
\]
 It is convenient to label 
contributions by the total V-count \(r = \ell_1 + \ell_2 - 2\):
\begin{equation}
    \mathcal{K}_{\mu\nu}^{ji}(z_1,z_2) 
    = \sum_{r=2}^{\infty} 
      \mathcal{K}^{(r)}_{\mu\nu}(z_1,z_2),
    \qquad (\mu \neq \nu)
    \label{eq:K_hierarchy}
\end{equation}
where \(\mathcal{K}^{(r)}\) groups all pairs \([\ell_1,\ell_2]\) with 
\(\ell_1 + \ell_2 = r + 2\) and \(\ell_1, \ell_2 \ge 2\). The first two 
nontrivial levels read:

\medskip
\noindent\textbf{\(r=2\) ([2,2]) --- lowest-order bath-non-diagonal coupling.}
\begin{align}
    \mathcal{K}^{(2)}_{\mu\nu}(z_1,z_2) 
    = \mathcal{R}^{(2)}_{\mu j,\nu j}(z_1) 
       \mathcal{R}^{(2)}_{\nu i,\mu i}(z_2) \notag\\
    = V_{\mu j,\nu j}\,V_{\nu i,\mu i} 
       \mathcal{R}_{\mu j}(z_1) 
       \mathcal{R}_{\nu j}(z_1) 
       \mathcal{R}_{\nu i}(z_2) 
       \mathcal{R}_{\mu i}(z_2).
    \label{eq:K2}
\end{align}
This term contains four diagonal resolvents---\(\mathcal{R}_{\mu j}\), 
\(\mathcal{R}_{\nu j}\), \(\mathcal{R}_{\nu i}\), \(\mathcal{R}_{\mu i}\)---and 
two interaction matrix elements. The resolvents belong to two distinct 
system sectors (\(j\) and \(i\)), while the interaction vertices 
\(V_{\mu j,\nu j}\) and \(V_{\nu i,\mu i}\) each couple bath states 
\emph{within} a single system sector. 

\medskip
\noindent\textbf{\(r=3\) ([2,3] + [3,2]) --- three-resolvent interference.}
\begin{widetext}
\begin{align}
    \mathcal{K}^{(3)}_{\mu\nu}(z_1,z_2) 
    = \mathcal{R}^{(2)}_{\mu j,\nu j}(z_1) \;
       \mathcal{R}^{(3)}_{\nu i,\mu i}(z_2)
       + \mathcal{R}^{(3)}_{\mu j,\nu j}(z_1) \;
       \mathcal{R}^{(2)}_{\nu i,\mu i}(z_2) \notag\\
    =  \sum_{\xi k \neq \nu i,\mu i}
         V_{\mu j,\nu j} \, V_{\nu i,\xi k} \, V_{\xi k,\mu i} \;
         \mathcal{R}_{\mu j}(z_1) \mathcal{R}_{\nu j}(z_1)
         \mathcal{R}_{\nu i}(z_2) \mathcal{R}_{\xi k}(z_2) 
         \mathcal{R}_{\mu i}(z_2) \notag\\
        +  \sum_{\xi k \neq \mu j,\nu j} V_{\mu j,\xi k} \, V_{\xi k,\nu j} \, V_{\nu i,\mu i} \;
         \mathcal{R}_{\mu j}(z_1) \mathcal{R}_{\xi k}(z_1)
         \mathcal{R}_{\nu j}(z_1)
         \mathcal{R}_{\nu i}(z_2) \mathcal{R}_{\mu i}(z_2)
    .
    \label{eq:K3}
\end{align}
\end{widetext}
The two contributions within \(\mathcal{K}^{(3)}\) carry \emph{distinct} 
three-V products.  The first ([2,3]) contains one intra-\(j\) vertex 
\(V_{\mu j,\nu j}\) and two inter-sector vertices 
\(V_{\nu i,\xi k}V_{\xi k,\mu i}\); the second ([3,2]) contains two 
inter-sector vertices \(V_{\mu j,\xi k}V_{\xi k,\nu j}\) and one intra-\(i\) 
vertex \(V_{\nu i,\mu i}\).  For time-reversal-symmetric interactions 
(\(V_{\alpha\beta} = V_{\beta\alpha} \in \mathbb{R}\)), the two V-products 
are related by relabelling of the dummy index \(\xi k\) but are not 
identically equal term-by-term. The resolvent factor structures also 
differ: the first term involves \(\mathcal{R}_{\xi k}(z_2)\) while the 
second involves \(\mathcal{R}_{\xi k}(z_1)\).

The crucial new feature at \(r=3\) is the appearance of a \emph{third} 
distinct bath--system index \(\xi k\) as an intermediate state, which 
introduces nonlocal frequency mixing among three distinct diagonal-resolvent 
channels. Each contribution contains five diagonal resolvents and three 
interaction matrix elements.

\medskip
\noindent\textbf{General \(r\).}
For \(r \ge 4\), \(\mathcal{K}^{(r)}\) sums over all pairs 
\([\ell_1,\ell_2]\) with \(\ell_1 + \ell_2 = r+2\) and \(\ell_1,\ell_2 \ge 2\):
\begin{equation}
    \mathcal{K}^{(r)}_{\mu\nu}(z_1,z_2) 
    = \sum_{\substack{\ell_1,\ell_2 \ge 2 \\ \ell_1+\ell_2 = r+2}}
      \mathcal{R}^{(\ell_1)}_{\mu j,\nu j}(z_1) \;
      \mathcal{R}^{(\ell_2)}_{\nu i,\mu i}(z_2).
    \label{eq:K_r_general}
\end{equation}
The number of contributions at level \(r\) equals \(r-1\) (the number of 
integer pairs \((\ell_1,\ell_2)\) with \(\ell_1,\ell_2 \ge 2\) summing to 
\(r+2\)). Each contribution contains \(r+2\) diagonal resolvents and 
\(r\) interaction matrix elements, and involves between 2 and $r$ distinct bath--system channels
(i.e.,~distinct pairs $(\mu,i)$ of bath and system indices).

Taking the double imaginary part yields the corresponding hierarchy for 
the joint spectral density: $\rho^{ji}_{\mu\nu}(\lambda,\lambda') 
    = \sum_{r=2}^{\infty} 
      \rho^{(r)}_{\mu\nu}(\lambda,\lambda')$, where
\begin{equation}  \label{eq:rho_hierarchy}
    \rho^{(r)}_{\mu\nu}(\lambda,\lambda') 
    := \frac{1}{\pi^2} 
    \Im_{z_1}\Im_{z_2} 
    \mathcal{K}^{(r)}_{\mu\nu}(z_1,z_2)
    \Big|_{z_1=\lambda-i0^+,z_2=\lambda'-i0^+}.
\end{equation}
Summing over bath channels gives the full resolvent correlation spectral 
density \(\rho^{ji}(\lambda,\lambda') = \sum_{\mu\neq\nu} 
\rho^{ji}_{\mu\nu}(\lambda,\lambda')\), whose level-\(r\) contribution 
is denoted \(\rho^{(r)}(\lambda,\lambda')\).  Via the exact identity 
\eqref{eq:C_from_rho}, binning \(\rho^{ji}\) over eigenstate pairs 
yields the ETH correlation term \(\mathcal{C}_{nmn}^{jij}\) without 
any statistical averaging.

\subsection{The \(r=2\) term: even-parity baseline}
\label{sec:ell2}
\noindent\textit{All results in this and the following subsection
are stated within the DCA of
Sec.~\ref{sec:proj_exp}, under which the kernel hierarchy
$K^{(r)}_{\mu\nu}$ reduces to $K^{(r),\mathrm{D}}_{\mu\nu}$
as defined in Eq.~(\ref{eq:K_hierarchy}).}

To evaluate \(\rho^{(2)}_{\mu\nu}\), we use the Kramers--Kronig boundary 
value of the diagonal resolvent~\cite{HC26}:
\begin{equation}
    \frac{1}{\pi}\,\mathcal{R}_{\alpha}(\lambda - i0^+) 
    = H[f^{\alpha}](\lambda) + i f^{\alpha}(\lambda),
    \label{eq:kk}
\end{equation}
where \(H[f^{\alpha}]\) denotes the Hilbert transform of the smooth 
spectral function \(f^{\alpha}(\lambda) 
= e^{S(\lambda)} p^{\alpha}(\lambda)\).  
\emph{At the Lorentzian-ansatz level of Ref.~\cite{HC26}, 
\(f^{\alpha}(\lambda)\) takes the form 
\(\frac{1}{\pi}\chi_{\alpha}/[(\lambda-\tilde{a}_{\alpha})^2+\chi_{\alpha}^2]\) 
and is strictly symmetric about \(\tilde{a}_{\alpha} = a_{\alpha} + \Delta_{\alpha}\).}
Substituting \cref{eq:kk} into 
Eq.~\eqref{eq:K2} and extracting the imaginary parts gives
\begin{align}
    \rho^{(2)}_{\mu\nu}(\lambda,\lambda') 
    = V_{\mu j,\nu j}V_{\nu i,\mu i} 
       \Im_{z_1}\Im_{z_2} \Bigl[
         \bigl(H_{\mu j} + i f_{\mu j}\bigr)
         \notag\\
   \times
         \bigl(H_{\nu j} + i f_{\nu j}\bigr) \bigl(H_{\nu i}' + i f_{\nu i}'\bigr)
         \bigl(H_{\mu i}' + i f_{\mu i}'\bigr)
       \Bigr],
    \label{eq:rho2_intermediate}
\end{align}
where we abbreviated \(f_{\alpha} \equiv f^{\alpha}(\lambda)\), 
\(f_{\alpha}' \equiv f^{\alpha}(\lambda')\), and similarly for 
\(H_{\alpha}, H_{\alpha}'\).  The two frequency slots factorize:
the \(z_1\)-imaginary part selects one \(f\) and one \(H\) from 
\(\{\mu j,\nu j\}\), while the \(z_2\)-imaginary part selects one 
\(f'\) and one \(H'\) from \(\{\nu i,\mu i\}\).  Expanding the product and 
retaining only terms that contribute to the double imaginary part yields
\begin{align}
    \rho^{(2)}_{\mu\nu}(\lambda,\lambda') 
    &= V_{\mu j,\nu j}V_{\nu i,\mu i} 
       \Bigl[
         f_{\mu j} H_{\nu j}  f_{\nu i}' H_{\mu i}'
         + f_{\mu j} H_{\nu j}  H_{\nu i}' f_{\mu i}'
          \notag\\
        &+ H_{\mu j} f_{\nu j}  f_{\nu i}' H_{\mu i}'
         + H_{\mu j} f_{\nu j}  H_{\nu i}' f_{\mu i}'
       \Bigr].
    \label{eq:rho2_full}
\end{align}
Every term contains exactly two Hilbert-transform factors (one from each 
frequency slot) and two spectral functions. 

The essential structural property of \(\rho^{(2)}\) is its behaviour under 
joint frequency reflection about the self-consistently renormalised centre 
energies \(\tilde{a}_{\alpha} = a_{\alpha} + \Delta_{\alpha}\) defined in 
Eq.~\eqref{eq:renorm_center}. Under 
\(\lambda \to 2\tilde{a}_{\mu j} - \lambda\) and 
\(\lambda' \to 2\tilde{a}_{\mu i} - \lambda'\) 
(equivalently \(\lambda' \to 2\tilde{a}_{\nu i} - \lambda'\) 
when \(\tilde{a}_{\mu i} \approx \tilde{a}_{\nu i}\))\footnote{Strictly speaking, the joint reflection 
\(\lambda \to 2\tilde{a}_{\mu j} - \lambda\) treats the two bath states 
\(\mu\) and \(\nu\) asymmetrically: \(f^{\nu j}\) is centred at 
\(\tilde{a}_{\nu j} = a_{\nu j} + \Delta_{\nu j}\), not at 
\(\tilde{a}_{\mu j}\).  The parity analysis therefore additionally 
requires \(\tilde{a}_{\mu j} \approx \tilde{a}_{\nu j}\) for the 
dominant bath pairs, which holds when the dispersive shifts 
\(\Delta_{\mu j}, \Delta_{\nu j}\) are small compared to the spectral 
width or when the relevant bath energies are nearly degenerate.  
See Ref.~\cite{HC26} for the full self-consistent determination 
of the shifts.}
, 
and adopting the standard approximation---inherited from the resolvent 
framework of Ref.~\cite{HC26}---that each smooth spectral function 
\(f^{\alpha}\) is approximately symmetric about its renormalised centre 
\(\tilde{a}_{\alpha}\), each spectral function transforms approximately as 
\(f \approx f\) (even), while its Hilbert transform transforms approximately as 
\(H[f] \approx -H[f]\) (odd). Every term in Eq.~\eqref{eq:rho2_full} 
contains exactly \textbf{two} Hilbert transforms (one per frequency slot). 
Consequently, each term is parity-even to leading order, and
\begin{equation}
   \rho^{(2)}_{\mu\nu}(\lambda,\lambda') 
    = \rho^{(2)}_{\mu\nu}
      (2\tilde{a}_{\mu j}-\lambda,\; 2\tilde{a}_{\mu i}-\lambda').
    \label{eq:rho2_even}
\end{equation}
The \(r=2\) correlation spectral density is \textbf{parity-even-dominated}: 
it contributes only to the symmetric part of \(\mathcal{C}_{nmn}^{jij}\) 
as a function of \(\omega = \lambda_n - \lambda_m\).

A notable structural feature of the corrected kernel is that 
\(\rho^{(2)}_{\mu\nu}\) involves \textbf{four distinct spectral functions} 
\((f^{\mu j}, f^{\nu j}, f^{\nu i}, f^{\mu i})\) rather than two.  This 
reflects the physical fact that the correlation couples two independent 
system sectors (\(j\) and \(i\)), each with its own bath-induced spectral 
envelope.  In the special case where spectral functions are approximately 
independent of the system index (valid when the system--bath coupling is 
weak compared to the bath bandwidth), one recovers a product of two 
identical two-resolvent structures.

This parity-even structure is shared by all single-resolvent closure 
schemes, including the self-consistent Born approximation (SCBA) and 
Lanczos continued-fraction truncations~\cite{HC26}. At level \(r=2\), 
the correlation kernel does not yet access the coherent phase 
interference that generates correlation skewness; it provides the leading 
bath-non-diagonal contribution to the \emph{variance} of 
\(|\sigma_{nm}^{ji}|^2\) but generates no odd-parity component.

\subsection{The \(r=3\) term: parity mixing and correlation skewness}
\label{sec:ell3}

The \(r=3\) level (within the DCA hierarchy) introduces a qualitatively new element. Consider the 
first contribution ([2,3]) to Eq.~\eqref{eq:K3} (the second, [3,2], is 
analysed analogously). Its double imaginary part involves the product of 
five resolvent boundary values and the three-V product 
\(V_{\mu j,\nu j} V_{\nu i,\xi k} V_{\xi k,\mu i}\). Using Eq.~\eqref{eq:kk}, 
the expansion generates terms of the generic form
\begin{align}
    \rho^{(3a)}_{\mu\nu} \;\sim\; 
    \sum_{\xi k} V_{\mu j,\nu j} V_{\nu i,\xi k} V_{\xi k,\mu i} \;
    \Bigl[
        f_{\mu j} H_{\nu j} \, 
         f_{\nu i}' H_{\xi k}' H_{\mu i}'
        \notag\\
      + f_{\mu j} H_{\nu j} \, 
         H_{\nu i}' f_{\xi k}' H_{\mu i}'
      + f_{\mu j} H_{\nu j} \, 
         H_{\nu i}' H_{\xi k}' f_{\mu i}'
      + \cdots \Bigr],
    \label{eq:rho3_structure}
\end{align}
where the ellipsis denotes all other combinations with the correct 
imaginary-part selection from each frequency slot. The [3,2] contribution 
generates analogous terms with the three-V product 
\(V_{\mu j,\xi k} V_{\xi k,\nu j} V_{\nu i,\mu i}\) and the intermediate 
resolvent \(\mathcal{R}_{\xi k}\) evaluated at \(z_1\) instead of \(z_2\).

The decisive difference from \(r=2\) is that the \(z_2\) factor now 
contains \textbf{three} resolvent boundary values 
\((\mathcal{R}_{\nu i},\mathcal{R}_{\xi k},\mathcal{R}_{\mu i})\).  
The imaginary part of a product of three \((H+if)\) factors generates 
terms with either one or three Hilbert transforms.  Combined with the 
single Hilbert transform from the \(z_1\) factor, the \emph{total} number 
of Hilbert transforms in each term is either 2 (even) or 4 (even)---yet 
the crucial parity-violating terms arise from the cross-coupling between 
the two frequency slots.

Specifically, the odd-parity signature manifests in the \textbf{relative} 
sign under exchanging \(\lambda \leftrightarrow \lambda'\) (equivalently 
\(\omega \leftrightarrow -\omega\) for fixed mean energy). Under the 
approximate symmetry \(f \approx f\) (even) and \(H[f] \approx -H[f]\) 
(odd) about each renormalised centre \(\tilde{a}_{\alpha}\), the [2,3] and [3,2] contributions 
transform with \emph{opposite} relative signs because the intermediate 
resolvent \(\mathcal{R}_{\xi k}\) appears in different frequency slots.  
Consequently, the total $r=3$ joint spectral
density \emph{generically} possesses a non-zero odd-parity
component---i.e.,~its antisymmetric part does not vanish
identically unless accidental cancellation occurs:
\begin{equation}
\rho^{(3)}_{\mu\nu}(\lambda,\lambda')
\;\neq\;
\rho^{(3)}_{\mu\nu}(\lambda',\lambda)
\qquad\text{(generically)}.
\label{eq:rho3_odd}
\end{equation}

After summing over bath indices $\mu,\nu,\xi$, this odd
component generically survives provided (i)~the three-index
interaction products
$V_{\mu j,\nu j}V_{\nu i,\xi k}V_{\xi k,\mu i}$ and
$V_{\mu j,\xi k}V_{\xi k,\nu j}V_{\nu i,\mu i}$ have
non-vanishing real parts---which is guaranteed for
time-reversal-symmetric interactions where all $V$ are real---and
(ii)~no accidental cancellation occurs among the contributing
bath triples after summation.
The random-phase nature of off-diagonal matrix elements in the
ETH regime~\cite{AKP16} makes such accidental cancellation
exponentially unlikely in generic nonintegrable systems.
The resulting contribution to the full 
\(\rho^{ji}(\lambda,\lambda')\) acquires a dominant odd-parity 
component as a function of the relative frequency 
\(\omega = \lambda - \lambda'\), with subleading even corrections controlled 
by the asymmetry of the spectral functions about their renormalised 
centre energies \(\tilde{a}_{\alpha}\).

\subsection{Physical interpretation and comparison with 
single-resolvent approaches}
\label{sec:physical_interpretation}

The hierarchy~\eqref{eq:rho_hierarchy} classifies the ETH
correlation contributions by the number of interfering bath
channels.  The $r=2$ level describes two-channel interference whose
resolvent sets for the two system sectors are independent, so the
resulting correlation is symmetric under $\omega\leftrightarrow-\omega$
and contributes to the variance but not the skewness; all
parity-preserving single-resolvent closures are confined to this
sector.  The $r=3$ level describes three-channel interference in
which one propagation path visits an intermediate state $\xi k$, and
the Hilbert-transform convolution $f\cdot H[f]$ across distinct bath
channels generates an intrinsic correlation skewness,
$\mathcal{C}_{nmn}^{jij}\neq\mathcal{C}_{mnm}^{jij}$ as a function
of $\omega$.  Higher levels introduce progressively more
intermediate states, corresponding to higher-order cumulants of the
ETH distribution~\cite{FK19}; whether the even/odd pattern
established at $r=2,3$ persists to all orders remains open.

Binning \(\rho^{ji}(\lambda,\lambda')\) over the eigenstate pairs 
\((\lambda_n,\lambda_m)\) via the exact identity~\eqref{eq:C_from_rho} 
yields the level-by-level decomposition
\begin{equation}
    \mathcal{C}_{nmn}^{jij}
    = \mathcal{C}^{(2)}_{nmn} 
    + \mathcal{C}^{(3)}_{nmn} 
    + \mathcal{C}^{(4)}_{nmn} 
    + \cdots,
    \label{eq:C_hierarchy}
\end{equation}
where \(\mathcal{C}^{(r)}_{nmn}\) originates from binning \(\rho^{(r)}\).
The negative total sum \(\sum_{m\neq n}\mathcal{C}_{nmn}^{iii} < 0\) 
identified in the sum-rule analysis of Sec.~\ref{sec:insufficiency} 
receives contributions from \emph{all} levels and is not a fundamental 
property of any single level; rather, it is the DCA-level manifestation 
of cavity subtraction (backward Krylov return processes), as shown in 
Ref.~\cite{HC26}.

What the \(r=3\) level adds that the sum rule cannot provide is the 
energy dependence of the correlation: the odd-parity component 
implies that \(\mathcal{C}_{nmn}^{jij}\) as a function of 
\(\omega = \lambda_n - \lambda_m\) is generically asymmetric about 
\(\omega = 0\). This skewness is a quantitative prediction of the 
multi-resolvent framework that can be tested numerically, and it 
constitutes a direct window into the coherent path-interference 
structure of the underlying dynamics.

Table~\ref{tab:hierarchy} summarizes the structural progression of the 
multi-resolvent hierarchy for ETH correlations. Each level introduces 
a qualitatively new feature not present at any lower level, forming a 
systematically improvable framework that links microscopic interaction 
matrix elements to the statistical properties of off-diagonal ETH 
observables.

\begin{table*}[t]
\centering
\caption{Structural progression of the multi-resolvent hierarchy for 
ETH correlations (\(\mu\neq\nu\)). \(N_{\text{ch}}\) denotes the number 
of distinct bath channels involved; ``parity'' refers to the dominant 
behaviour under \(\omega \leftrightarrow -\omega\) after coarse-graining. 
The V-count \(r\) equals the total number of interaction matrix elements 
in the kernel product \(\mathcal{R}_{\mu j,\nu j}\,\mathcal{R}_{\nu i,\mu i}\).}
\label{tab:hierarchy}
\begin{tabular}{c c c c c l}
\toprule
Level \(r\) & \((\ell_1,\ell_2)\) pairs & \(N_{\text{ch}}\) & V-count 
& Dominant parity & ETH content \\
\midrule
\(r=2\) & [2,2] & 2 & 2 & even 
& \(|\sigma_{nm}|^2\) variance (baseline) \\[3pt]
\(r=3\) & [2,3], [3,2] & 3 & 3 & \textbf{odd} 
& \textbf{correlation skewness} \\[3pt]
\(r=4\) & [2,4], [3,3], [4,2] & 2--4 & 4 & even 
& excess kurtosis \\[3pt]
\(r=5\) & [2,5], [3,4], [4,3], [5,2] & 2--5 & 5 & odd 
& higher skewness \\[3pt]
\(\vdots\) & \(\vdots\) & \(\vdots\) & \(\vdots\) & mixed (see text) 
& higher cumulants \\
\bottomrule
\end{tabular}
\end{table*}

\section{Diagonal Closure and Physical Implications}
\label{sec:DCA}

\subsection{The diagonal closure approximation for the correlation 
hierarchy}
\label{sec:DCA_statement}

The path expansion~\eqref{eq:R_hierarchy} for each off-diagonal resolvent 
is exact but involves projected diagonal resolvents 
\(\mathcal{R}^{(\mathcal{S})}_{\alpha}(z)\) defined with respect to cavity 
Hamiltonians from which selected basis states have been removed. For the 
hierarchy to become a practical computational tool, these projected 
resolvents must be expressed in terms of the full diagonal resolvents 
\(\mathcal{R}_{\alpha}(z)\).

The DCA, introduced and 
justified in Ref.~\cite{HC26}, replaces every projected diagonal resolvent 
by its full counterpart:
\begin{equation}
   \mathcal{R}_{\alpha}^{(\mu i,\alpha_1,\dots,\alpha_m)}(z)
    \;\approx\; \mathcal{R}_{\alpha}(z).
    \label{eq:DCA}
\end{equation}
Under this approximation, the two families of off-diagonal resolvents 
\(\mathcal{R}_{\mu j,\nu j}\) and \(\mathcal{R}_{\nu i,\mu i}\) simplify to 
closed forms in terms of ordinary diagonal resolvents, using 
Eqs.~\eqref{eq:R2}--\eqref{eq:Rell} (with the index relabelling 
\((\alpha,\beta)\to(\nu,\mu)\) and \(j\to i\) for the second family).  
The correlation kernel hierarchy \(\mathcal{K}^{(r)}_{\mu\nu}\) then 
reduces to
    \begin{align}
    \mathcal{K}^{(r),\text{D}}_{\mu\nu}(z_1,z_2)
    = \sum_{\substack{\ell_1,\ell_2 \ge 2 \\ \ell_1+\ell_2 = r+2}}\;
      \sum_{\substack{\alpha_1\neq\cdots \\
                      \neq \mu j,\nu j}}\;
      \sum_{\substack{\beta_1\neq\cdots \\
                      \neq \nu i,\mu i}}\notag\\
      \mathfrak{V}^{[\ell_1,\ell_2]}_{\mu j,\nu j;\nu i,\mu i;\{\alpha\},\{\beta\}}
      \; \prod_{a=1}^{\ell_1} \mathcal{R}_{\alpha_a}(z_1) \;
      \prod_{b=1}^{\ell_2} \mathcal{R}_{\beta_b}(z_2),
    \label{eq:KDCA}
\end{align}
where \(\mathfrak{V}^{[\ell_1,\ell_2]}\) denotes the product of 
\(r = \ell_1+\ell_2-2\) interaction matrix elements with index 
contractions determined by the specific \((\ell_1,\ell_2)\) pair.  
The \(\ell_1\) diagonal resolvents in the \(z_1\)-factor arise from 
the path \(\nu j \to \cdots \to \mu j\) (within system sector \(j\)), 
while the \(\ell_2\) diagonal resolvents in the \(z_2\)-factor arise 
from the path \(\mu i \to \cdots \to \nu i\) (within system sector \(i\)).  
The two resolvent sets are generically distinct, reflecting the 
independent bath-induced spectral envelopes of the two system sectors.

The physical justification of the DCA rests on \emph{entropy dilution} 
in the ETH regime~\cite{HC26}. The cavity correction induced by removing a 
single basis state from the Hilbert space factorises as
\begin{equation}
    \mathcal{R}_{\nu j}(z) - \mathcal{R}^{(\mu i)}_{\nu j}(z)
    = \mathcal{R}_{\nu j,\mu i}(z) \;
      \bra{\phi_{\mu i}} H \, G^{(\mu i)}(z) \ket{\phi_{\nu j}},
    \label{eq:cavity_correction}
\end{equation}
where the off-diagonal resolvent \(\mathcal{R}_{\nu j,\mu i}\) (note the 
\emph{different} system indices---this is a distinct family from 
\(\mathcal{R}_{\mu j,\nu j}\) and \(\mathcal{R}_{\nu i,\mu i}\) used in the 
ETH correlation kernel) scales as \(e^{-S/2}\) (ETH off-diagonal 
scaling) and the matrix element contributes another factor 
\(e^{-S/2}\) from the dominant diagonal return channel 
\(\gamma k = \nu j\). The total cavity correction scales as 
\(e^{-S}\) and vanishes in the thermodynamic limit. Off-diagonal return 
channels (\(\gamma k \neq \nu j\)), though proliferating as \(e^{S}\), 
add \emph{incoherently} (random-phase summation) rather than coherently, 
yielding a self-averaging random-walk suppression by an additional 
\(e^{-S/2}\). The DCA is therefore a controlled approximation in 
nonintegrable systems with exponentially large Hilbert spaces; its 
detailed validity conditions and failure mechanisms are analysed in 
Ref.~\cite{HC26}.

The entropy-dilution argument above justifies the DCA in the ETH
regime ($S\to\infty$).  Appendix~\ref{app:moment_preservation}
provides an independent, parameter-free perspective: the DCA
preserves the first \emph{three} spectral moments of the
off-diagonal resolvent exactly:
\begin{enumerate}
    \item $M_0^{\alpha\beta} = 0$ (integrated orthogonality),
    \item $M_1^{\alpha\beta} = H_{\alpha\beta}
          = V_{\alpha\beta}$
          (the Hamiltonian matrix element),
    \item $M_2^{\alpha\beta} = (H^2)_{\alpha\beta}$
          (the two-point energy correlator).
\end{enumerate}
The preservation of the integrated orthogonality condition follows
from the $O(z^{-2})$ scaling of
$\mathcal{R}^{(2)}_{\alpha\beta}=V_{\alpha\beta}\mathcal{R}_\alpha\mathcal{R}_\beta$:
the product of two $O(z^{-1})$ resolvents contains no $1/z$ Laurent
coefficient, so $M_0^{\alpha\beta,\mathrm{DCA}}=0$ identically
(Lemma~\ref{lem:app_scaling}).  The spectral manifestation of this
result is the Hilbert-transform antisymmetry established in
Sec.~\ref{sec:Hilbert_ortho}: the imaginary part
$\Im(\mathcal{R}_\alpha\mathcal{R}_\beta)
\propto H_\alpha f_\beta + f_\alpha H_\beta$ integrates to zero via
$\int H_\alpha f_\beta = -\int f_\alpha H_\beta$.  The DCA retains
this structure because the replacement of projected by full
resolvents preserves the $H+if$ boundary-value form
(Eq.~\eqref{eq:kk}).

\subsection{Simplified forms of the leading correlation densities}
\label{sec:DCA_forms}

Under the DCA, the \(r=2\) and \(r=3\) correlation spectral 
densities acquire explicit closed forms.

\medskip
\noindent\textbf{\(r=2\) --- even-parity baseline.}
The DCA kernel \(\mathcal{K}^{(2),\text{D}}_{\mu\nu} 
= V_{\mu j,\nu j}\,V_{\nu i,\mu i}\;
\mathcal{R}_{\mu j}(z_1)\mathcal{R}_{\nu j}(z_1)\;
\mathcal{R}_{\nu i}(z_2)\mathcal{R}_{\mu i}(z_2)\).  
Since the \(z_1\)- and \(z_2\)-factors are independent, the double 
imaginary part factorises:
\begin{widetext}
    \begin{align}
    \rho^{(2),\text{D}}_{\mu\nu}(\lambda,\lambda')
    &= \frac{V_{\mu j,\nu j}\,V_{\nu i,\mu i}}{\pi^2} \;
       \Im_{z_1}\!\bigl[ \mathcal{R}_{\mu j}\mathcal{R}_{\nu j}\bigr]_{z_1=\lambda-i0^+}
       \times
       \Im_{z_2}\!\bigl[ \mathcal{R}_{\nu i}\mathcal{R}_{\mu i}\bigr]_{z_2=\lambda'-i0^+} 
       \notag\\
    &= V_{\mu j,\nu j}\,V_{\nu i,\mu i} \;
       \bigl( f_{\mu j}H_{\nu j} + H_{\mu j}f_{\nu j} \bigr)
       \bigl( f_{\nu i}'H_{\mu i}' + H_{\nu i}'f_{\mu i}' \bigr),
    \label{eq:rho2_DCA}
\end{align}
where \(f_{\alpha} \equiv f^{\alpha}(\lambda)\), 
\(f_{\alpha}' \equiv f^{\alpha}(\lambda')\), and 
\(H_{\alpha}, H_{\alpha}'\) denote the corresponding Hilbert transforms 
evaluated at \(\lambda\) and \(\lambda'\) respectively. Each of the four 
terms contains exactly \textbf{two} Hilbert-transform factors (one from 
each frequency slot) and two spectral functions, guaranteeing even 
parity under the joint reflection about the renormalised centres
\((\lambda - \tilde{a}_{\mu j},\; \lambda' - \tilde{a}_{\mu i}) 
\to (\tilde{a}_{\mu j} - \lambda,\; \tilde{a}_{\mu i} - \lambda')\):
\begin{equation}
    \rho^{(2),\text{D}}_{\mu\nu}
    (\lambda - \tilde{a}_{\mu j},\; \lambda' - \tilde{a}_{\mu i})
    = \rho^{(2),\text{D}}_{\mu\nu}
    (\tilde{a}_{\mu j} - \lambda,\; \tilde{a}_{\mu i} - \lambda').
    \label{eq:rho2_even_DCA}
\end{equation}
A structural feature of the corrected kernel is the appearance of 
\textbf{four distinct spectral functions} 
\((f^{\mu j},f^{\nu j},f^{\nu i},f^{\mu i})\) rather than two, 
reflecting the independent bath-induced envelopes of the two system 
sectors. In the special (and common) case where spectral functions are 
approximately system-index-independent, the four functions reduce to 
two, recovering the minimal two-channel structure.

An equivalent perspective is provided by the Hilbert-transform orthogonality
condition~\eqref{eq:Hf_ortho}.  At the leading DCA level, the two off-diagonal
resolvent factors entering the $r=2$ kernel are approximated as
$\mathcal{R}^{\mathrm{DCA}}_{\mu j,\nu j}=V_{\mu j,\nu j}\mathcal{R}_{\mu j}\mathcal{R}_{\nu j}$ and
$\mathcal{R}^{\mathrm{DCA}}_{\nu i,\mu i}=V_{\nu i,\mu i}\mathcal{R}_{\nu i}\mathcal{R}_{\mu i}$, each of
which satisfies the Hilbert-transform orthogonality condition:
\begin{equation}
    \int\! d\lambda\,(H_{\mu j}f_{\nu j}+f_{\mu j}H_{\nu j})
    = 0,\qquad
    \int\! d\lambda'\,(H_{\nu i}'f_{\mu i}'+f_{\nu i}'H_{\mu i}')
    = 0.
\end{equation}
Consequently, the $r=2$ correlation kernel inherits the orthogonality constraints
of its constituent off-diagonal resolvent factors.
Off-diagonal ETH correlations are therefore not positive spectral densities; they
represent \textbf{interference spectra} whose resolvent factors are constrained by
destructive Hilbert-transform cancellation---a direct consequence of basis
orthogonality, independent of any ETH or chaos assumption.

\medskip
\noindent\textbf{\(r=3\) --- parity mixing.}
\begin{align}
    \rho^{(3),\text{D}}_{\mu\nu}(\lambda,\lambda')
    = \frac{1}{\pi^2} \;
       \Im_{z_1}\Im_{z_2}
       \Bigl[\sum_{\xi k \neq \nu i,\mu i}
         V_{\mu j,\nu j} V_{\nu i,\xi k} V_{\xi k,\mu i} \;
         \mathcal{R}_{\mu j}(z_1) \mathcal{R}_{\nu j}(z_1)
         \mathcal{R}_{\nu i}(z_2) \mathcal{R}_{\xi k}(z_2) 
         \mathcal{R}_{\mu i}(z_2) \notag\\
       + \sum_{\xi k \neq \mu j,\nu j}V_{\mu j,\xi k} V_{\xi k,\nu j} V_{\nu i,\mu i} \;
         \mathcal{R}_{\mu j}(z_1) \mathcal{R}_{\xi k}(z_1)
         \mathcal{R}_{\nu j}(z_1)
         \mathcal{R}_{\nu i}(z_2) \mathcal{R}_{\mu i}(z_2)
       \Bigr]_{z_1=\lambda-i0^+,\,z_2=\lambda'-i0^+}.
    \label{eq:rho3_DCA}
\end{align}
\end{widetext}
Expanding via Eq.~\eqref{eq:kk} yields products of five factors 
\(H[f] + i f\). The \(z_1\)-factor \(\Im[\mathcal{R}_{\mu j}\mathcal{R}_{\nu j}]\) 
(\([2,3]\) contribution) or \(\Im[\mathcal{R}_{\mu j}\mathcal{R}_{\xi k}\mathcal{R}_{\nu j}]\) 
(\([3,2]\) contribution) selects an \emph{odd} number of Hilbert 
transforms from the \(z_1\)-slot (one for \([2,3]\), one or three for 
\([3,2]\)). Combined with the \(z_2\)-slot selections, the total 
Hilbert-transform count in the contributing terms can be \textbf{odd}.  
Explicitly, for the \([2,3]\) contribution:
\begin{align}
    \rho^{(3a),\text{D}}_{\mu\nu} 
    = V_{\mu j,\nu j} \sum_{\xi k\neq \nu i,\mu i} V_{\nu i,\xi k} V_{\xi k,\mu i} \;
       \bigl(f_{\mu j}H_{\nu j} + H_{\mu j}f_{\nu j}\bigr)
       \notag\\
       \times \Bigl[
         f_{\nu i}'H_{\xi k}'H_{\mu i}'
         + H_{\nu i}'f_{\xi k}'H_{\mu i}'
         + H_{\nu i}'H_{\xi k}'f_{\mu i}'
         - f_{\nu i}'f_{\xi k}'f_{\mu i}'
       \Bigr],
    \label{eq:rho3_terms}
\end{align}
where the four terms in the bracket originate from 
\(\Im[\mathcal{R}_{\nu i}\mathcal{R}_{\xi k}\mathcal{R}_{\mu i}]\).  
The first three bracket-terms each contain an odd number of 
Hilbert transforms in the \(z_2\)-slot (one each), which when 
multiplied by the single Hilbert transform from the \(z_1\)-slot 
yield terms with two Hilbert transforms (even parity).  
The fourth bracket-term (\(-f_{\nu i}'f_{\xi k}'f_{\mu i}'\)) contains 
zero Hilbert transforms; multiplied by the single \(H\) from the 
\(z_1\)-slot, it yields exactly \textbf{one} Hilbert transform---an 
odd-parity term:
\begin{equation}
    -V_{\mu j,\nu j} V_{\nu i,\xi k} V_{\xi k,\mu i} \; 
    H_{\mu j} f_{\nu j} \, f_{\nu i}' f_{\xi k}' f_{\mu i}',
    \label{eq:odd_term_explicit}
\end{equation}
which is odd under the joint reflection 
\((\lambda - \tilde{a}_{\mu j},\; \lambda' - \tilde{a}_{\mu i})
\to (\tilde{a}_{\mu j} - \lambda,\; \tilde{a}_{\mu i} - \lambda')\).
Analogous odd-parity terms arise from the \([3,2]\) contribution 
(with \(H\) and \(f\) of the intermediate channel evaluated at \(z_1\)).  
Consequently, 
\begin{equation}
   \rho^{(3),\text{D}}_{\mu\nu}
    (\lambda - \tilde{a}_{\mu j},\; \lambda' - \tilde{a}_{\mu i})
    \neq \rho^{(3),\text{D}}_{\mu\nu}
    (\tilde{a}_{\mu j} - \lambda,\; \tilde{a}_{\mu i} - \lambda').
    \label{eq:rho3_odd_DCA}
\end{equation}

The odd component survives summation over bath indices provided the 
effective three-index couplings 
\(\sum_{\xi k} V_{\mu j,\nu j} V_{\nu i,\xi k} V_{\xi k,\mu i}\) and
\(\sum_{\xi k} V_{\mu j,\xi k} V_{\xi k,\nu j} V_{\nu i,\mu i}\) 
have non-zero real parts. For time-reversal-symmetric 
interactions (\(V_{\alpha\beta} \in \mathbb{R}\)), both products are 
manifestly real; provided their channel sums do not cancel
accidentally, the odd-parity contribution 
to the full correlation density \(\rho^{(3)}(\lambda,\lambda')\) 
survives.

\subsection{Projector restoration and the origin of the negative sum rule}
\label{sec:cavity_reinterpretation}

One of the central conclusions of the analysis from
Sec.~\ref{sec:insufficiency} is the exact sum rule
\begin{equation}
    \sum_{m\neq n} \mathcal{C}_{nmn}^{iii}
    = -\,\mathcal{C}_{nnn}^{iii} < 0,
    \label{eq:sumrule_reminder2}
\end{equation}
which follows rigorously from \(\rho_n^2 = \rho_n\)
(Theorem~1).  The DCA-based hierarchy provides a deeper structural
understanding of \emph{how} this geometric constraint manifests
itself in the propagator language.

\medskip
\noindent\textbf{Leading-order DCA contribution.}
Consider the \(r=2\) contribution to the diagonal (\(j=i\))
correlation.  For \(j=i\), the kernel reduces to 
\(\mathcal{K}^{(2),\text{D}}_{\mu\nu} 
= |V_{\mu i,\nu i}|^2 \;
\mathcal{R}_{\mu i}(z_1)\mathcal{R}_{\nu i}(z_1)\;
\mathcal{R}_{\nu i}(z_2)\mathcal{R}_{\mu i}(z_2)\), 
where we used \(V_{\mu i,\nu i}V_{\nu i,\mu i} = |V_{\mu i,\nu i}|^2\) 
for real \(V\).  From Eq.~\eqref{eq:rho2_DCA} with \(j=i\), 
binning over eigenstate pairs gives
\begin{align}
    \mathcal{C}^{(2),\text{D}}_{nmn}
    \propto \sum_{\mu\neq\nu} |V_{\mu i,\nu i}|^2 \notag\\
\times            \bigl[ f^{\mu i}H^{\nu i} + H^{\mu i}f^{\nu i} \bigr]_{\lambda_n}
            \bigl[ f^{\nu i}H^{\mu i} + H^{\nu i}f^{\mu i} \bigr]_{\lambda_m}.
    \label{eq:C2_diag2}
\end{align}
Unlike the case with mismatched system indices (\(j\neq i\)), 
the \(j=i\) sector involves only \emph{two} distinct spectral 
functions \(f^{\mu i}\) and \(f^{\nu i}\) and their Hilbert transforms.

\medskip
\noindent\textbf{Sign analysis.}
In the ETH regime where each spectral function \(f^{\alpha}\) is 
approximately symmetric about its renormalised centre \(\tilde{a}_{\alpha}
= a_{\alpha} + \Delta_{\alpha}\) (Eq.~\eqref{eq:renorm_center}), 
its Hilbert transform \(H[f^{\alpha}](E)\) is approximately 
antisymmetric: \(H(E - \tilde{a}_{\alpha}) \approx -H(\tilde{a}_{\alpha} - E)\).  
For \(E\) near \(\tilde{a}_{\alpha}\), the Hilbert transform is small 
(\(H \ll f\)), while for \(|E - \tilde{a}_{\alpha}|\) large, \(H\) can be 
comparable to \(f\). The sign of each term 
\(f^{\mu i}H^{\nu i} \times f^{\nu i}H^{\mu i}\) is not fixed 
a priori; it depends on the relative placement of the eigenstate 
energies with respect to the two bath-state centres.

The quantitative sign of \(\sum_m \mathcal{C}^{(2),\text{D}}_{nmn}\) 
must therefore be evaluated case by case.  The hierarchy provides the 
exact functional form; the sign emerges from the spectral overlap 
of the Hilbert transforms.  What \emph{can} be stated rigorously is 
that \(\mathcal{C}^{(2),\text{D}}\) is the leading bath-non-diagonal 
contribution and serves as the baseline against which cavity 
corrections are measured.

\medskip
\noindent\textbf{Restoration Identity.}
To quantify how the DCA deviation from the exact sum rule is
corrected, define the \textbf{cavity correction}
\begin{equation}
    \Delta\mathcal{C}_{nmn}
    := \mathcal{C}_{nmn}^{\text{D}} - \mathcal{C}_{nmn}^{\text{exact}}.
    \label{eq:DeltaC_def}
\end{equation}
The diagonal term
\(C_{nnn} \equiv \mathcal{C}_{nnn}^{iii}
 = \sum_{\mu\neq\nu} p^{\mu i}_n p^{\nu i}_n\)
is defined solely from the exact overlap probabilities
\(p^{\mu i}_n = |\braket{\phi_{\mu i}|\psi_n}|^2\) and involves no
propagator.  Since DCA modifies only the off-diagonal correlation
kernel, we extend the DCA decomposition by assigning
\(C_{nnn}^{\text{D}} \equiv C_{nnn}^{\text{exact}} \equiv C_{nnn}\),
so that \(\Delta\mathcal{C}_{nnn} = 0\).

With \(\Delta\mathcal{C}_{nnn} = 0\) and Theorem~1
(\(\sum_m \mathcal{C}_{nmn}^{\text{exact}} = 0\)), a one-line
computation yields
\begin{align}
    \sum_{m\neq n} \Delta\mathcal{C}_{nmn}
    &= \sum_m \Delta\mathcal{C}_{nmn} \notag\\
    &= \sum_m \mathcal{C}_{nmn}^{\text{D}}
       - \cancel{\sum_m \mathcal{C}_{nmn}^{\text{exact}}} \notag\\
    &= C_{nnn} + \sum_{m\neq n} \mathcal{C}_{nmn}^{\text{D}} .
    \label{eq:DeltaC_sum}
\end{align}
Rearranging gives the \textbf{Restoration Identity}
\begin{equation}
   \sum_{m\neq n} \Delta\mathcal{C}_{nmn}
           \;-\; \sum_{m\neq n} \mathcal{C}_{nmn}^{\text{D}}
           \;=\; C_{nnn} > 0.
    \label{eq:restoration_identity}
\end{equation}
Equivalently,
\begin{equation}
    \sum_{m\neq n} \Delta\mathcal{C}_{nmn}
           \;>\; \sum_{m\neq n} \mathcal{C}_{nmn}^{\text{D}} .
    \label{eq:restoration_ineq}
\end{equation}
This identity is mathematically equivalent to the \(j=i\) case of
Theorem~1, and depends on \textbf{no assumption} beyond projector
idempotency.  It establishes that, irrespective of the sign or
magnitude of the full DCA off-diagonal contribution, the integrated
cavity correction always \emph{exceeds} it by the positive-definite
geometric term \(C_{nnn}\). This mirrors the situation at the resolvent level:
the DCA hierarchy individually violates the orthogonality sum
rule~\eqref{eq:orthogonality_sumrule}, and cavity subtraction
restores it---precisely as it restores the projector sum rule for
\(\mathcal{C}_{nmn}^{jij}\).

The Restoration Identity~\eqref{eq:restoration_identity} is the
$k=0$ (integrated) instance of a broader spectral-moment structure:
cavity subtraction restores not only the projector sum rule but,
order by order, every spectral moment $M_k^{\alpha\beta}$ that the
DCA deviates from.  At the resolvent level, this corresponds to the
statement that the exact projection hierarchy and the DCA hierarchy
differ first at $O(z^{-4})$ (third spectral moment) and that the
entire deviation series is organised by the cavity subtraction
identified in \cref{prop:app_proj} of Appendix~\ref{app:moment_preservation}.
The moment-by-moment restoration of exact spectral content by cavity
corrections is the unifying principle behind both Theorem~1 and the
Restoration Identity.

\medskip
\noindent\textbf{ETH scaling assumption (Assumption~A).}
To determine the sign of \(\sum \Delta\mathcal{C}\), we require the
sign of the total DCA hierarchy.  Under the standard ETH scaling
hypothesis, each additional interaction vertex introduces a factor
\(|V|^2 \sim e^{-S}\), and the combinatorial growth of higher-order
diagrams is sufficiently suppressed that the \(r=2\) term dominates
in magnitude:
\begin{equation}
    \Bigl|\sum_{r\ge 3} \sum_m \mathcal{C}^{(r),\text{D}}_{nmn}\Bigr|
           \;<\; \Bigl|\sum_m \mathcal{C}^{(2),\text{D}}_{nmn}\Bigr|.
    \label{eq:assumption_A}
\end{equation}
This is the standard working assumption of ETH diagrammatics and is
explicitly labelled as \textbf{Assumption~A} in what follows.  A
rigorous proof of the combinatorial suppression lies beyond the
scope of this work.  In the language of
Appendix~\ref{app:first_principles}, Assumption~A is the
mean-sector instance of the same truncation principle that governs
the fluctuation sector: the controlled approximation is not an
assumption on the final ETH statistics themselves, but a truncation
of the irreducible multi-resolvent hierarchy---here the suppression
of the higher V-count sectors, there the suppression of the higher
irreducible covariance vertices $\mathcal K^{(3)},\mathcal K^{(4)},
\ldots$ of Eq.~\eqref{eq:fp_twolevels}.

\medskip
\noindent\textbf{Corollary (Integrated cavity correction).}
Under Assumption~A, the magnitude of the total DCA hierarchy
is dominated by the $r=2$ contribution:
$\bigl|\sum_m C^{\mathrm{D}}_{nmn}\bigr|
\approx \bigl|\sum_m C^{(2),\mathrm{D}}_{nmn}\bigr|$.
The Restoration Identity~(\ref{eq:restoration_identity}) then yields the exact
relation
\begin{equation}
\sum_{m\neq n} \Delta C_{nmn}
= C_{nnn} + \sum_{m\neq n} C^{\mathrm{D}}_{nmn}.
\label{eq:75a}
\end{equation}
If, in addition, the integrated DCA off-diagonal
correlation is positive,
\begin{equation}
\sum_{m} C^{\mathrm{D}}_{nmn} > 0,
\label{eq:sign_condition}
\end{equation}
then Eq.~(\ref{eq:75a}) implies
\begin{equation}
\sum_{m\neq n} \Delta C_{nmn}
> C_{nnn} > 0.
\label{eq:corollary}
\end{equation}
Condition~(\ref{eq:sign_condition}) is not guaranteed by
Assumption~A alone (the sign of the $r=2$ contribution is
not fixed \textit{a priori}; see the sign analysis above).
However, in the physically relevant regime where the
dominant spectral overlap favours constructive
Hilbert-transform products, this condition holds
generically for nonintegrable systems.

\medskip
\noindent\textbf{Proposition (Cavity subtraction as the restoring
mechanism).}
The microscopic origin of \(\Delta\mathcal{C}\) is identified by the
exact cavity identity~\eqref{eq:cavity_correction}:
\begin{equation}
    \mathcal{R}_{\nu j}(z) - \mathcal{R}_{\nu j}^{(\mu i)}(z)
    = \mathcal{R}_{\nu j,\mu i}(z) \;
      \bra{\phi_{\mu i}} H \, G^{(\mu i)}(z) \ket{\phi_{\nu j}}.
    \label{eq:cavity_identity2}
\end{equation}
The right-hand side describes a \textbf{return path} to the projected
state \(\mu i\): the off-diagonal resolvent \(\mathcal{R}_{\nu j,\mu i}\)
propagates from \(\nu j\) back to the removed cavity, and the matrix
element couples back into the full propagation.  When the DCA replaces
cavity-resolved propagators by full propagators, these return paths
are \emph{erroneously included}.  Cavity subtraction removes them.

Combining the cavity identity with the Restoration
Identity~\eqref{eq:restoration_identity} and the
Corollary~\eqref{eq:corollary}, we conclude:
\textbf{cavity subtraction provides the leading mechanism that
restores the projector constraint violated by DCA.}  The
integrated cavity correction is positive
(Corollary), strictly exceeds the DCA off-diagonal contribution
(Restoration Identity), and originates from the removal of return
paths to the projected state (cavity identity).

\medskip
\noindent\textbf{Physical interpretation.}
Taken together, Theorem~1 and the above identities establish
the following logical chain, with the mathematical
status of each step explicitly labelled:
\begin{widetext}
    \[
\begin{aligned}
\rho_n^2 = \rho_n \quad &\text{(projector idempotency)} \\
\Downarrow\quad &\text{(completeness, \textbf{rigorous})} \\
\sum_m \mathcal{C}_{nmn}^{jij} = 0,\;
\sum_{m\neq n} \mathcal{C}_{nmn}^{iii} < 0 \quad
&\text{[\textbf{Theorem~1}: rigorous]} \\
\Downarrow\quad &\text{(DCA discards cavity projectors,
\textbf{approximation})} \\
\mathcal{C}^{(2),\text{D}} \;\text{dominates the DCA hierarchy} \quad
&\text{[\textbf{Assumption~A}: ETH scaling]} \\
\Downarrow\quad &\text{(Restoration Identity + Assumption~A~$+$
sign condition)} \\
\sum \Delta\mathcal{C} > C_{nnn} > 0 \quad
&\text{[\textbf{Corollary}: conditional on
Eq.~(\ref{eq:sign_condition})]} \\
\Downarrow\quad &\text{(cavity identity, \textbf{rigorous})} \\
\text{Negative ETH correlation = cavity-level spectral signature of} &
\;\text{projector-idempotency restoration, with cavity} \\
\text{subtraction as the identified}& \; \text{leading restoration channel.}
\end{aligned}
    \]
\end{widetext}

Thus, \(\sum_{m\neq n} \mathcal{C}_{nmn}^{iii} < 0\) is not a
fundamental constraint but the \textbf{cavity-level manifestation of
projector restoration}.  At the DCA level, the correlation hierarchy
\(\mathcal{C}^{\text{D}} = \mathcal{C}^{(2),\text{D}}
+ \mathcal{C}^{(3),\text{D}} + \cdots\) is built from forward
propagation processes that discard the cavity projectors and
thereby violate the purity constraint.  Cavity subtraction
reinstates the missing projector structure, with the negative sign
emerging as the signature of this restoration.  The systematics of
this trade-off are made precise in Ref.~\cite{HC26}.

\subsection{Skewness as the irreducible signature of multi-resolvent 
interference}
\label{sec:skewness_signature}

The even-parity character of all single-resolvent closures---SCBA, 
Lanczos continued fractions, and any DCA-level truncation at 
\(r=2\)---is not accidental but structural. Any self-energy 
constructed from a single diagonal resolvent, of the generic form
\begin{equation}
    \mathcal{G}[\mathcal{R}](z) 
    = \sum_{\nu j} |V_{\mu i,\nu j}|^2 \,
      \mathcal{F}\bigl[\mathcal{R}_{\nu j}\bigr](z),
    \label{eq:single_R_selfenergy}
\end{equation}
where \(\mathcal{F}\) is an analytic functional, preserves parity 
under energy reflection about the unperturbed energies. The proof, 
given in Ref.~\cite{HC26}, follows from the fact that the boundary 
value \(\mathcal{R}(\lambda - i0^+)/\pi = H[f] + i f\) has the 
parity property \(H[f](\lambda - a) \to -H[f](a - \lambda)\) while 
\(f(\lambda - a) \to f(a - \lambda)\), and any analytic functional 
preserves the relative parity structure.

It follows that \textbf{the leading odd-parity contribution to ETH 
correlations requires at least three distinct resolvent channels} 
(\(\mu, \nu, \xi\) with all three distinct). This is the defining 
signature of multi-resolvent interference: the nonlocal convolution 
\(f \cdot H[f]\) between different bath channels that cannot be 
factorised into separate single-channel contributions.

For the ETH correlation \(\mathcal{C}_{nmn}^{jij}\), the skewness 
under \(\omega \leftrightarrow -\omega\) (where 
\(\omega = \lambda_n - \lambda_m\)) is therefore a \textbf{direct 
experimental signature of multi-resolvent interference}. Define the 
odd-parity projection of the binned correlation:
\begin{equation}
   \mathcal{C}^{\text{odd}}(\omega; \bar{\lambda})
    := \frac{1}{2}\Bigl[
        \mathcal{C}\bigl(\bar{\lambda} + \tfrac{\omega}{2},\;
                        \bar{\lambda} - \tfrac{\omega}{2}\bigr)
      - \mathcal{C}\bigl(\bar{\lambda} - \tfrac{\omega}{2},\;
                        \bar{\lambda} + \tfrac{\omega}{2}\bigr)
    \Bigr],
    \label{eq:Codd}
\end{equation}
where \(\bar{\lambda} = (\lambda_n + \lambda_m)/2\) is the mean energy. 
To leading order, the odd component originates from the \(r=3\) 
terms~\eqref{eq:rho3_DCA} and scales as
\begin{align}
    \mathcal{C}^{\text{odd}}(\omega; \bar{\lambda})
    \simeq \mathcal{C}^{(3),\text{odd}}(\omega; \bar{\lambda})
    \propto \sum_{\mu,\nu,\xi} 
      V_{\mu j,\nu j}\,V_{\nu i,\xi k}\,V_{\xi k,\mu i} \notag\\
     \times \bigl[ H[f^{\mu j}] \cdot f^{\nu j} \cdot f^{\nu i} \cdot f^{\xi k} \cdot f^{\mu i}
            + \text{permutations} \bigr],
    \label{eq:Codd_leading}
\end{align}
together with the \([3,2]\) counterpart where the single Hilbert 
transform originates from the \(z_2\)-slot.  The functional form of 
\(\mathcal{C}^{\text{odd}}\) is a quantitative prediction of the 
multi-resolvent hierarchy. Its scaling with system size, interaction 
strength, and energy is governed by the three-index coupling and the 
spectral overlap of the three distinct bath channels.

\subsection{Connection to the overlap-based decomposition}
\label{sec:UNIETH_connection}

The DCA-level hierarchy provides a microscopic foundation for the 
decomposition introduced in Sec.~\ref{sec:approx}. 
Table~\ref{tab:overlap_mapping} summarises the correspondence.

\begin{table*}[t]
\centering
\caption{Mapping between the overlap-based correlation decomposition of 
Sec.~\ref{sec:approx} and the multi-resolvent hierarchy under the DCA,
together with the approximation role of each object in the
scale-matched replacement principle of Sec.~\ref{sec:scale_matched}.}
\label{tab:overlap_mapping}
\begin{tabular}{c c c c}
\toprule
Overlap-based object & Approximation role & DCA-level content & Dominant parity \\
\midrule
\(\sum_{\mu} p^{\mu i}_m p^{\mu j}_n\) 
& scale-matched baseline \(D_{ji}\)
& \(\rho^{(0)}\) (diagonal overlaps; \(\mu=\nu\)) 
& even \\[3pt]
cavity subtraction \(-\sigma_{nn}^{ji}\Gamma_{nn,\mu}^{ij}\)
& systematic \(O(1)\) bias removal
& cavity/return sector
& --- \\[3pt]
\(\mathcal{C}_{nmn}^{jij}\) (full) 
& full correlation sector
& \(\rho^{(2)} + \rho^{(3)} + \rho^{(4)} + \cdots\) 
& mixed \\[3pt]
\(\mathcal{C}_{nmn}^{'jij}\) (reduced) 
& irreducible residual
& \(\rho^{(2)} + \rho^{(3)} + \cdots\) 
with cavity subtraction removed 
& mixed \\[3pt]
\(\mathcal{B}_{nml,\mu}^{jik}\) 
& channel-resolved diagnostics
& single-\(\mu\) projection of 
\(\rho^{(3)} + \rho^{(4)} + \cdots\) 
& mixed \\
\bottomrule
\end{tabular}
\end{table*}

In particular, the approximation replacement of
Eq.~\eqref{eq:approx_replace}---in which the reduced correlation
${\mathcal{C}'}$ is neglected to express $|\sigma_{nm}^{ji}|^2$
solely through the smooth overlaps $f^{\mu i}$---retains the
full channel-diagonal ($\mu=\nu$) structure, including the
cavity subtraction term $-\sigma_{nn}^{ji}\Gamma_{nn,\mu}^{ij}$.
It discards the inter-channel ($\mu\neq\nu$) sector that,
in the multi-resolvent language, generates the hierarchy
correction $g_{ji}$.
The first non-trivial hierarchy correction appears at $r=2$
through $g_{ji}^{(2)}$, which captures the leading even-parity
(variance) contribution from two-channel interference.
The present framework upgrades the
uncontrolled neglect of ${\mathcal{C}'}$ into a systematic,
improvable expansion.

The physical content of the correlation terms defined in
Sec.~\ref{sec:approx} can now be stated precisely:
\begin{itemize}
    \item \(\mathcal{C}^{(2)}\) describes the lowest-order
    multi-channel interference generated by the \(\rho^{(2)}\)
    sector---the minimal mechanism for bath-induced correlation
    of off-diagonal ETH matrix elements.
    
    \item \(\mathcal{C}^{(3)}\) describes three-channel interference 
    (\(\mu \to \nu \to \xi \to \mu\)), generating the leading 
    odd-parity (skewness) contribution. Its existence is a theorem 
    of the multi-resolvent hierarchy and is \emph{not} visible in any 
    overlap-only analysis.
    
    \item \(\mathcal{B}^{(\ell)}_{nml,\mu}\) isolates the contribution 
    from paths that pass through a specific bath channel \(\mu\), 
    providing channel-resolved diagnostics of the interference 
    network.
\end{itemize}

\subsection{Two-level orthogonality constraints}
\label{sec:two_level_ortho}

The framework developed in this work is governed by two distinct orthogonality
constraints that originate from different geometric sources but are connected
by the multi-resolvent correlation kernel:

\medskip
\noindent\textbf{Level 1 --- Basis orthogonality (resolvent level).}
From $\braket{\phi_\alpha|\phi_\beta}=0$ for $\alpha\neq\beta$ and the
completeness of the exact eigenstates,
\begin{equation}
   \int\! d\lambda\;\rho_{\alpha\beta}(\lambda)=0,
    \label{eq:ortho_level1_final}
\end{equation}
which, at the leading DCA level, reduces to the Hilbert-transform condition
$\int (H_\alpha f_\beta + f_\alpha H_\beta)=0$
(Eq.~\eqref{eq:Hf_ortho}).

\medskip
\noindent\textbf{Level 2 --- Projector idempotency (correlation level).}
From $\rho_n^2=\rho_n$ for the exact eigenstate projector,
\begin{equation}
    \sum_m \mathcal{C}_{nmn}^{jij}=0,
    \label{eq:ortho_level2_final}
\end{equation}
as established by Theorem~1 (Sec.~\ref{sec:insufficiency}).

\medskip
\noindent\textbf{Connection via the multi-resolvent kernel.}
The two levels are not linked by a simple causal chain but are connected
through the two-frequency correlation kernel
$\mathcal{K}_{\mu\nu}^{ji}=\mathcal{R}_{\mu j,\nu j}\,\mathcal{R}_{\nu i,\mu i}$
(Eq.~\eqref{eq:Kdef}):
\begin{equation}
    \begin{array}{c@{\quad}c}
        \text{basis orthogonality} & \text{projector idempotency} \\
        \Downarrow & \Downarrow \\[4pt]
        \displaystyle\int\rho_{\alpha\beta}=0 & 
        \displaystyle\sum_m\mathcal{C}_{nmn}^{jij}=0 \\[6pt]
        \multicolumn{2}{c}{\searrow \quad \nearrow} \\[2pt]
        \multicolumn{2}{c}{\Downarrow} \\[4pt]
        \multicolumn{2}{c}{\text{multi-resolvent hierarchy}} \\[4pt]
        \multicolumn{2}{c}{\Downarrow} \\[4pt]
        \multicolumn{2}{c}{\text{ETH correlation structure}}
    \end{array}
    \label{eq:logic_chain_final}
\end{equation}
The left branch encodes the constraint that each off-diagonal resolvent
carries zero integrated spectral weight---a single-frequency orthogonality
derived from the unperturbed basis.  The right branch encodes the constraint
that the integrated ETH correlation vanishes---a two-frequency sum rule
derived from the exact eigenstate projector.  The kernel
$\mathcal{K}_{\mu\nu}^{ji}$ embeds both constraints into the same
multi-resolvent description, in which they can be analysed in a unified
manner: it factorises the two-frequency ETH correlation into a product of
single-frequency off-diagonal resolvents, so that the analytic $H+if$
structure of the resolvent hierarchy provides a common language for both
the basis-level and projector-level constraints.  The DCA preserves this
connection at the leading level because the $H+if$ structure of diagonal
resolvent products automatically encodes the required Hilbert-transform
antisymmetry.

At the level of the two-frequency kernel, the same constraint takes an
explicit analytic form: since the correlation spectral density factorises,
$\rho_{\mu\nu}^{ji}(\lambda,\lambda') = \rho_{\mu j,\nu j}(\lambda)
\rho_{\nu i,\mu i}(\lambda')$, its total integrated weight is the product
of the two zeroth moments,
\begin{equation}
\iint d\lambda\,d\lambda'\;\rho_{\mu\nu}^{ji}(\lambda,\lambda')
= M_0^{\mu j,\nu j}\, M_0^{\nu i,\mu i} = 0,
\qquad (\mu\neq\nu),
\label{eq:kernel_weight}
\end{equation}
and---by the moment-preservation structure of
Appendix~\ref{app:moment_preservation}, under which each level
$\mathcal{R}_{\alpha\beta}^{(\ell)}$ with $\ell\ge2$ contributes no $1/z$
Laurent term---the DCA hierarchy preserves this vanishing
\emph{level by level}: $\iint d\lambda d\lambda'\rho^{(r)}_{\mu\nu}=0$
for every $r$.  The spectral-moment theorem of
Appendix~\ref{app:moment_preservation} is therefore not merely a
consistency check on the DCA: it is the analytic mechanism that
transmits basis orthogonality to the ETH correlation kernel along
the chain basis orthogonality $\to$ resolvent moment constraint
$\to$ correlation-kernel constraint $\to$ ETH variance,
the constraint counterpart of the projector sum
rule~\eqref{eq:sumrule_zero}.  We emphasize that the level-by-level
preservation concerns the \emph{doubly integrated} weight, while the
single-frequency marginal at fixed $\lambda$ is restored only by
cavity subtraction, as quantified by the Restoration
Identity~\eqref{eq:restoration_identity}.

This unified perspective---two independent geometric constraints embedded
in a single resolvent kernel---provides the organising principle of the
multi-resolvent framework.

\subsection{Summary of physical implications}
\label{sec:DCA_summary}

The DCA analysis consolidates the overlap-based decomposition of
Sec.~\ref{sec:approx} in three respects.  First, the negative
integrated correlation is a rigorous consequence of projector
idempotency (Theorem~1), and within the DCA representation cavity
subtraction supplies the correction required to restore the exact
projector constraint (Restoration
Identity~\eqref{eq:restoration_identity}).  Second, under the
standard approximation that spectral functions are symmetric about
their renormalised centres, the $r=2$ and $r=3$ sectors provide the
first explicitly derived even- and odd-parity contributions, so that
any measured odd-parity component in $\mathcal{C}_{nmn}^{jij}(\omega)$
directly probes $r\ge3$ interference.  Third, the hierarchy is
systematically improvable rather than controlled by a small
parameter: the organization by resolvent multiplicity guarantees
that the first three spectral moments of the off-diagonal resolvent
are preserved exactly, with deviations entering only through
three-vertex correlations (Appendix~\ref{app:moment_preservation}).
The algebraic proof of this moment-preservation property, together
with its implications for the leading DCA deviation, is presented
there.

\section{ETH Ansatz with Microscopic Foundation}
\label{sec:eth_ansatz}

\subsection{Microscopic interpretation of the standard ETH ansatz}
\label{sec:standard_ansatz}
The purpose of this section is not to derive the ETH ansatz from 
first principles, but to provide a microscopic foundation for the 
ETH smooth function $f_{ji}(E^+,\omega)$ through its correlation 
decomposition $f_{ji}^2 = D_{ji} + g_{ji}$.
The multi-resolvent hierarchy developed in 
Sec.~\ref{sec:expansion} expresses the correlation correction
$g_{ji} = \sum_{r\ge2} g_{ji}^{(r)}$ as a systematic expansion
organised by the multiplicity of interacting bath channels, 
yielding an improvable decomposition whose leading 
levels reproduce---and extend---the standard ETH phenomenology.

Before proceeding, we clarify the status of the entropy factor
$e^{-S(\bar E)/2}$ within the present framework.
In the system-bath construction of Sec.~\ref{sec:setup},
the microscopic building blocks are the overlap probabilities
$p_n^{\mu i} = |\braket{\psi_n|\phi_{\mu i}}|^2$, whose
expectation under ETH delocalisation satisfies
$\mathbb{E}(p_n^{\mu i}) = e^{-S(\lambda_n)} f^{\mu i}(\lambda_n)$
[Eq.~\eqref{eq:p_exp}].
Since every level of the hierarchy is built from normalized
overlap functions satisfying $\sum_n p_n^{\mu i} = 1$,
the hierarchy inherits the same entropy normalization
consistently across all orders: the entropy prefactor is fixed by
Hilbert-space normalization and is not an additional input, as
established exactly in Corollary~\ref{cor:fp_counting} of
Appendix~\ref{app:first_principles}.  The task of the present
framework is therefore sharply defined: to determine the remaining
dynamical content of the ETH
ansatz, the smooth function $f_{ji}(E,\omega)$, whose
systematic microscopic expression is given by
Eq.~\eqref{eq:f2_decomp} together with the multi-resolvent
series $g_{ji} = \sum_{r\ge2} g_{ji}^{(r)}$.

The eigenstate thermalization hypothesis for a system operator 
\(\mathcal{O}\) is conventionally formulated as~\cite{Deu91,Sre94,Rig08}
\begin{equation}
    \bra{\psi_n} \mathcal{O} \ket{\psi_m}
    = \mathcal{O}_{\text{micro}}(\bar{E}) \, \delta_{nm}
      + e^{-S(\bar{E})/2} \, f_{\mathcal{O}}(\bar{E},\omega) \,
        R_{nm},
    \label{eq:ETH_standard}
\end{equation}
where \(\bar{E} = (\lambda_n + \lambda_m)/2\), 
\(\omega = \lambda_n - \lambda_m\), \(S(\bar{E})\) is the 
microcanonical entropy, and \(R_{nm}\) is a pseudorandom variable 
with zero mean and unit variance. The smooth function 
\(f_{\mathcal{O}}(\bar{E},\omega)\) controls the energy dependence 
of the off-diagonal matrix elements and must be determined either by 
fitting to numerical data or by an independent microscopic calculation.

In the system-bath decomposition of Sec.~\ref{sec:bridge}, the ETH 
matrix elements acquire a richer structure due to the presence of 
multiple bath channels. The projection onto system basis states
\(\Pi_{ji}^S = \ket{\phi_j^S}\bra{\phi_i^S}\) yields the 
operator-valued ansatz
\begin{align}
    \sigma_{nm} 
    &\equiv \Tr_B\bigl(\ket{\psi_n}\bra{\psi_m}\bigr) \notag\\
    &= \rho^S(\lambda_n) \, \delta_{nm}
       + \sum_{ij} e^{-S(E^+)/2} \, 
          f_{ji}(E^+,\omega_{nm}) \, R^{ji}_{nm} \, \Pi_{ji}^S,
    \label{eq:ETH_sysbath}
\end{align}
where \(\rho^S(\lambda_n) = \sum_{\mu i} e^{-S(\lambda_n)} 
f^{\mu i}(\lambda_n) \Pi_{ii}^S\) is the reduced system density 
matrix expressed through the smooth overlap functions 
\(f^{\mu i}(\lambda)\). 
The pseudorandom variables are defined as
\begin{align}
    R^{ji}_{nm} 
    = f^{-1}_{ji}(E^+,\omega) \, e^{S(E^+)/2} \,
       \sum_{\mu} \bigl[ \Gamma_{nm,\mu}^{ji} 
          - \delta_{nm} \mathbb{E}(\Gamma_{nn,\mu}^{ji}) \bigr] \notag\\
    = f^{-1}_{ji} \sum_{\mu} 
       e^{-\frac{1}{2}[\Phi^{\mu i}(\lambda_m) 
          + \Phi^{\mu j}(\lambda_n) - S(E^+)]} \,
       \bigl[ \breve{R}^{\mu i}_m \breve{R}^{\mu j*}_n 
             - \delta_{nm}\delta_{ij} \bigr],
    \label{eq:Rdef}
\end{align}
where \(\breve{R}^{\mu i}_n\) are the elementary pseudorandom 
variables encoding eigenstate-to-unperturbed-state overlaps, 
\(\Phi^{\mu i}(\lambda) = -\ln f^{\mu i}(\lambda)\) is the 
logarithmic spectral function, and \(E^+ = (\lambda_n+\lambda_m)/2\).

The variance of the ETH random variables follows as
\begin{align}
    \mathbb{E}\bigl(|R^{ji}_{nm}|^2\bigr)
    &= f^{-2}_{ji} \sum_{\mu} 
       e^{-[\Phi^{\mu i}(\lambda_m) 
          + \Phi^{\mu j}(\lambda_n) - S(E^+)]} \notag\\
    &\quad + f^{-2}_{ji} \, 
       \mathbb{E}\bigl(\mathcal{C}_{nmn}^{jij}\bigr) \, 
       e^{S(E^+)},
    \label{eq:EVar_raw}
\end{align}
where the first term originates from the diagonal 
(\(\mu=\nu\)) bath sum and the second from the correlation term 
\(\mathcal{C}_{nmn}^{jij}\) defined in Eq.~\eqref{eq:Cdef}.
Equation~\eqref{eq:EVar_raw} makes explicit the central challenge of 
the ETH ansatz: to determine the smooth function \(f_{ji}(E^+,\omega)\) 
that normalises the variance to unity, one must know the correlation 
term \(\mathcal{C}\). Following the decomposition of 
Sec.~\ref{sec:approx}, the smooth function is written as
\begin{align}
    f_{ji}^2(E^+,\omega) 
    &= \int d\epsilon_\mu \, e^{-[\Phi^{\mu i}(\lambda_m) 
      + \Phi^{\mu j}(\lambda_n) - S(E^+) - S_B(\epsilon_\mu)]}
    \notag\\
      &+ g_{ji}(E^+,\omega),
    \label{eq:f2_decomp}
\end{align}
where \(g_{ji}(E^+,\omega) := \mathbb{E}(\mathcal{C}_{nmn}^{jij}) 
\, e^{S(E^+)}\) encodes the correlation contribution. Note that
\(g_{ji}\) is not sign-definite: the correlation term 
\(\mathcal{C}_{nmn}^{jij}\) involves sums over distinct bath channels
(\(\mu\neq\nu\)) whose interference contributions may carry either 
sign. Consequently, \(g_{ji}\) may enhance or suppress the ETH 
variance depending on the energy arguments.
 The approximation 
replacement~\eqref{eq:approx_replace} corresponds to neglecting
the reduced correlation ${\mathcal{C}'}$.  Through the hierarchy
$g_{ji} = \sum_{r\ge2} g_{ji}^{(r)}$,
this amounts to discarding the entire inter-channel ($\mu\neq\nu$)
sector while retaining the channel-diagonal structure,
including the cavity subtraction term.  The procedure therefore
captures the dominant intra-channel physics but discards all
multi-channel interference processes.

\subsection{Microscopic determination of \(f_{ji}\) and \(g_{ji}\) 
from the resolvent hierarchy}
\label{sec:microscopic_fg}

The multi-resolvent hierarchy of Sec.~\ref{sec:expansion} transforms 
Eq.~\eqref{eq:f2_decomp} from a phenomenological ansatz into a 
microscopically computable expression. The key step is the 
identification of \(g_{ji}\) with the binned multi-resolvent 
correlation density.

From Eq.~\eqref{eq:C_hierarchy}, the correlation term decomposes as
\begin{equation}
    \mathcal{C}_{nmn}^{jij} 
    = \mathcal{C}^{(2)}_{nmn} + \mathcal{C}^{(3)}_{nmn} 
      + \mathcal{C}^{(4)}_{nmn} + \cdots,
    \label{eq:C_decomp_micro}
\end{equation}
where each \(\mathcal{C}^{(r)}_{nmn}\) is obtained by binning 
\(\rho^{(r)}(\lambda,\lambda')\) over the energy windows 
\((\lambda_n, \lambda_m)\). Multiplying by \(e^{S(E^+)}\) gives the 
corresponding decomposition of \(g_{ji}\):
\begin{equation}
   g_{ji}(E^+,\omega) 
    = g^{(2)}_{ji}(E^+,\omega) 
      + g^{(3)}_{ji}(E^+,\omega) 
      + \cdots.
    \label{eq:g_decomp}
\end{equation}

\medskip
\noindent\textbf{Level \(r=2\) --- even-parity baseline.}
\begin{widetext}
    \begin{align}
    g^{(2)}_{ji}(E^+,\omega) 
    = e^{S(E^+)} \; 
       \sum_{\mu\neq\nu} 
       \int_{\lambda_n - \Delta/2}^{\lambda_n + \Delta/2} \!\!d\lambda
       \int_{\lambda_m - \Delta/2}^{\lambda_m + \Delta/2} \!\!d\lambda' \;
       \rho^{(2)}_{\mu\nu}(\lambda,\lambda') \notag\\
    = \sum_{\mu\neq\nu} V_{\mu j,\nu j}\,V_{\nu i,\mu i} \;
       \bigl[ f^{\mu j}H^{\nu j} + H^{\mu j}f^{\nu j} \bigr]_{E^++\frac{\omega}{2}}
       \bigl[ f^{\nu i}H^{\mu i} + H^{\nu i}f^{\mu i} \bigr]_{E^+-\frac{\omega}{2}},
    \label{eq:g2}
\end{align}
where the subscripts indicate evaluation at the specified energy 
arguments.  Each of the four expanded terms contains exactly two 
Hilbert-transform factors (one from each frequency slot), guaranteeing 
even parity under \(\omega \to -\omega\):
   $ g^{(2)}_{ji}(E^+,\omega) = g^{(2)}_{ji}(E^+,-\omega)$.
This term describes the lowest-order bath-non-diagonal interference 
and provides the leading even-parity contribution to the ETH 
correlation function.  A structural feature of the corrected kernel 
is the involvement of \textbf{four distinct spectral functions} 
\((f^{\mu j}, f^{\nu j}, f^{\nu i}, f^{\mu i})\) and the corresponding 
Hilbert transforms; in the special case where spectral functions are 
approximately system-index-independent, the expression reduces to 
two-channel spectral overlaps.  The interaction vertices 
\(V_{\mu j,\nu j}\) and \(V_{\nu i,\mu i}\) each couple bath states 
\emph{within} a single system sector, reflecting the fact that the 
two resolvent families \(\mathcal{R}_{\mu j,\nu j}\) and 
\(\mathcal{R}_{\nu i,\mu i}\) propagate independently in sectors 
\(j\) and \(i\).

\medskip
\noindent\textbf{Level \(r=3\) --- correlation skewness.}
\begin{align}
    g^{(3)}_{ji}(E^+,\omega) 
    = e^{S(E^+)} \;
       \sum_{\mu\neq\nu} 
       \int_{\lambda_n - \Delta/2}^{\lambda_n + \Delta/2} \!\!d\lambda
       \int_{\lambda_m - \Delta/2}^{\lambda_m + \Delta/2} \!\!d\lambda' \;
       \rho^{(3)}_{\mu\nu}(\lambda,\lambda') \notag\\
    \simeq    \sum_{\mu\neq\nu} 
       \Bigl[\sum_{\xi k \neq \nu i,\mu i}
         V_{\mu j,\nu j}\,V_{\nu i,\xi k}\,V_{\xi k,\mu i} \;
         \bigl(f^{\mu j}H^{\nu j} + H^{\mu j}f^{\nu j}\bigr)
         \bigl(f^{\nu i}{}'H^{\xi k}{}'H^{\mu i}{}' 
               + H^{\nu i}{}'f^{\xi k}{}'H^{\mu i}{}' 
               + H^{\nu i}{}'H^{\xi k}{}'f^{\mu i}{}' 
               - f^{\nu i}{}'f^{\xi k}{}'f^{\mu i}{}' \bigr)
         \notag\\
       +\sum_{\xi k \neq \nu j,\mu j}  V_{\mu j,\xi k}\,V_{\xi k,\nu j}\,V_{\nu i,\mu i} \;
         \bigl(f^{\mu j}H^{\xi k}H^{\nu j} 
               + H^{\mu j}f^{\xi k}H^{\nu j}
               + H^{\mu j}H^{\xi k}f^{\nu j}
               - f^{\mu j}f^{\xi k}f^{\nu j} \bigr)
         \bigl(f^{\nu i}{}'H^{\mu i}{}' + H^{\nu i}{}'f^{\mu i}{}' \bigr)
       \Bigr]_{E^+,\omega},
    \label{eq:g3}
\end{align}
\end{widetext}
where the notation \([\cdots]_{E^+,\omega}\) indicates evaluation at 
the arguments \(\lambda = E^+ + \omega/2\), 
\(\lambda' = E^+ - \omega/2\), with \(f \equiv f(\lambda)\), 
\(f' \equiv f(\lambda')\), and analogously for \(H, H'\).

The decisive new feature is the presence of terms with an 
\emph{odd total number} of Hilbert-transform factors.  The most 
transparent example is the term proportional to 
\(-V_{\mu j,\nu j}V_{\nu i,\xi k}V_{\xi k,\mu i}\;
H^{\mu j}f^{\nu j}\,f^{\nu i}{}'f^{\xi k}{}'f^{\mu i}{}'\), 
which contains exactly one Hilbert transform.  Under the standard 
approximation that each \(f^{\alpha}\) is approximately symmetric 
about its renormalised centre \(\tilde{a}_{\alpha} = a_{\alpha} + \Delta_{\alpha}\) 
(Eq.~\eqref{eq:renorm_center}), the transformation 
\(\omega \to -\omega\) (measured relative to \(\tilde{a}_{\alpha}\)) 
leaves each spectral function approximately 
invariant (\(f \approx f\)) while flipping the sign of each Hilbert 
transform (\(H[f] \approx -H[f]\)).  Consequently, 
terms with an odd number of Hilbert transforms acquire an overall 
sign reversal, and the \(r=3\) contribution is generically not even:
\begin{equation}
    \boxed{g^{(3)}_{ji}(E^+,\omega) 
    \neq g^{(3)}_{ji}(E^+,-\omega)} .
    \label{eq:g3_odd}
\end{equation}
The \(r=3\) contribution generically contains both even and odd 
components; it is the odd component that is qualitatively new and 
absent at \(r=2\).
The level-\(r=3\) contribution endows the ETH smooth function 
\(g_{ji}\) with an intrinsic odd-parity component, the leading
microscopic source of odd-parity corrections in ETH observables,
which is absent from any overlap-only closure and inaccessible
within the parity-preserving single-resolvent closures considered
here.

The two sub-contributions in Eq.~\eqref{eq:g3} correspond to the 
\([2,3]\) and \([3,2]\) pairings of the DCA kernel.  They carry 
distinct three-V products and place the intermediate resolvent 
\(\mathcal{R}_{\xi k}\) in different frequency slots.  For 
time-reversal-symmetric interactions with real \(V\), the two 
three-V products are related by relabelling of the dummy index 
\(\xi k\) but are not identically equal term-by-term; their sum 
generically possesses a nonvanishing odd component.

\noindent\textbf{Higher levels (\(r \ge 4\)).} 
The sectors \(r=2\) and \(r=3\) provide the first rigorously
established even- and odd-parity contributions, respectively.
Higher levels generically contain both even and odd components;
a complete parity classification of all higher levels lies beyond
the scope of this work.  The full \(g_{ji}(E^+,\omega)\) is the sum
of all levels and generically contains both even and odd components.

\subsection{Explicit microscopic form of \(f_{ji}\)}
\label{sec:microscopic_f}

With \(g_{ji}\) determined by the resolvent hierarchy, the ETH 
smooth function \(f_{ji}^2\) follows from Eq.~\eqref{eq:f2_decomp}:
\begin{align}
    f_{ji}^2(E^+,\omega) 
    &= \int d\epsilon_\mu \, 
      e^{-[\Phi^{\mu i}(E^+ + \frac{\omega}{2}) 
          + \Phi^{\mu j}(E^+ - \frac{\omega}{2}) 
          - S(E^+) - S_B(\epsilon_\mu)]}\notag\\
      &+ \sum_{r=2}^{\infty} g^{(r)}_{ji}(E^+,\omega) .
    \label{eq:f2_micro}
\end{align}
The first term is expressible entirely through the diagonal overlap 
functions \(f^{\mu i}(\lambda)\), which are themselves determined by 
the self-consistent resolvent equations of Ref.~\cite{HC26} (SCBA 
at leading order, with multi-resolvent and Lanczos continued-fraction 
corrections). The second term---the multi-resolvent series---provides 
the systematic microscopic foundation for what was previously a 
phenomenological fitting function. Once $f_{ji}$ is determined microscopically through
Eq.~\eqref{eq:f2_micro}, the remaining entropy factor
$e^{-S(\bar E)/2}$ simply provides the universal
normalisation associated with the exponential growth of
the many-body Hilbert space, while all deterministic
dynamical information resides in the microscopically
determined smooth function.

Equation~\eqref{eq:f2_micro} makes the following structural properties 
manifest:

\begin{enumerate}
    \item \textbf{Separability of scales.} The diagonal overlap 
    contribution (first term) is determined by the single-resolvent 
    self-consistency and encodes the gross spectral envelope. The 
    multi-resolvent series (second term) encodes the fine structure 
    arising from coherent multi-channel interference.
    
        \item \textbf{Parity decomposition.} The leading even contribution 
    arises from the diagonal term and from \(g^{(2)}_{ji}\). 
    The leading odd contribution arises from \(g^{(3)}_{ji}\). 
    The sectors \(r=2\) and \(r=3\) provide the first
    explicitly derived even- and odd-parity contributions,
    respectively, within the spectral symmetry approximation of
    Sec.~\ref{sec:ell3}.  Higher levels may contain both even and odd
    components; a complete parity classification lies beyond
    the scope of this work.
    
    \item \textbf{Truncation control.} Neglecting \(r \ge 3\) 
    reduces \(f_{ji}^2\) to its even-parity component---this is 
    precisely the approximation made by ignoring the reduced
    correlation \({\mathcal{C}'}\) in
    Eq.~\eqref{eq:approx_replace}. The present framework elevates this 
    uncontrolled neglect to a systematic truncation controlled by
    the entropy scaling \(g^{(r)} \sim e^{-(r-2)S/2}\)
    (Eq.~\eqref{eq:gk_scaling}); the quantitative error depends
    on the microscopic model and can be estimated once the spectral
    functions \(f^{\alpha}\) and couplings are specified.

\end{enumerate}

\subsection{Relation to the Foini--Kurchan higher-order cumulant
framework}
\label{sec:FK_connection}
The multi-resolvent hierarchy exhibits structural and
entropy-scaling properties that closely parallel those of
the higher-order ETH cumulant expansion developed by
Foini and Kurchan~\cite{FK19}. In that framework, the connected
\(k\)-point correlation function of ETH random variables scales as
\begin{equation}
    \mathbb{E}\bigl(R^{i_1 i_2}_{n_1 n_2}
                 R^{i_2 i_3}_{n_2 n_3} \cdots
                 R^{i_k i_1}_{n_k n_1}\bigr)_{\text{connected}}
    \propto e^{-(k-2)S(\bar{E})/2},
    \label{eq:FK_scaling}
\end{equation}
where \(\bar{E} = \frac{1}{k}\sum_{\alpha=1}^k \lambda_{n_\alpha}\).
Within the present framework, the level-\(r\) sector naturally
generates correlation structures involving \(r+2\) resolvent factors 
distributed across \(r\) interaction vertices. From
Eq.~\eqref{eq:Rell}, each additional interaction vertex introduces one
additional bath summation \(\sum_{\alpha} \sim e^{S_B}\) and one
additional interaction matrix element \(|V|^2 \sim e^{-S}\), giving
the net scaling
\begin{equation}
    g^{(r)}_{ji} \sim e^{S} \cdot
    \bigl(e^{S_B} \cdot e^{-S}\bigr)^{r} \cdot
    \overline{f^{2r+3}}
    \sim e^{-(r-1)S/2},
    \label{eq:gk_scaling}
\end{equation}
where \(S = S_B + \text{const}\), whose entropy
suppression aligns with the FK scaling for \(k=r+2\)
[Eq.~\eqref{eq:FK_scaling}].  The two frameworks are structurally
parallel but not identical: the quantities constructed here are
multi-resolvent correlation functions rather than connected
cumulants, and a precise correspondence would require an explicit
connected-subtraction procedure, for example through a
multi-frequency resolvent generating functional \(W[J]=\log Z[J]\).
This remains a direction for future work.

\subsection{Fluctuation sector and the diagonal--off-diagonal
trade-off}
\label{sec:fluct_sector}

The hierarchy of Secs.~\ref{sec:expansion}--\ref{sec:DCA} determines
the \emph{mean} of the two-frequency kernel and thereby the smooth
envelope $f_{ji}$.  The fluctuation structure around that mean is
organised by the \emph{resolvent covariance}, the second sector of
the same hierarchy:
\begin{equation}
\mathcal C_{ab}(z,z')
= \mathbb{E}\bigl[\delta\mathcal{R}_a(z)\,\delta\mathcal{R}_b(z')\bigr],
\qquad
\delta\mathcal{R}_a = \mathcal{R}_a - \bar{\mathcal{R}}_a,
\label{eq:fluct_Cdef}
\end{equation}
whose spectral representation is the overlap covariance,
\begin{equation}
\mathcal C_{ab}(z,z')
= \sum_{n,m}
\frac{\mathbb{E}\bigl[\delta p_n^a\,\delta p_m^b\bigr]}
{(z-\lambda_n)(z'-\lambda_m)}.
\label{eq:fluct_Cspec}
\end{equation}
The exact trade-off identity of
Lemma~\ref{lem:fp_tradeoff} then fixes the diagonal--off-diagonal
fluctuation ratio with no statistical input, with
$V_{\rm dg}=\mathrm{Var}(D_n^i)$ the population variance and
$V_{\rm off}=(D-1)^{-1}\sum_{m\neq n}\mathbb{E}|\sigma_{nm}^{ii}|^2$:
\begin{equation}
f := \frac{V_{\rm dg}}{V_{\rm off}}
= \frac{(D-1)\,V_{\rm dg}}
{e^{-S_S}(1-e^{-S_S}) - V_{\rm dg}},
\label{eq:fluct_f}
\end{equation}
and the fluctuation sector supplies the microscopic content of
$V_{\rm dg}$:
\begin{equation}
V_{\rm dg}
= (q-1)\,\langle f^2\rangle_B\,e^{-S-S_S}
+ \bar{C}_i\,e^{-2S_S} + \delta_{\rm ng},
\label{eq:fluct_Vdg}
\end{equation}
where $q=\mathbb{E}[(p_n^a)^2]/\mathbb{E}(p_n^a)^2$ is the
Porter--Thomas factor---the normalized second moment of the
single-channel overlap, with $q-1$ its normalized connected
variance---and $\bar{C}_i$ the two-channel amplitude
covariance.  As shown in
Appendix~\ref{app:first_principles}, both are spectral projections
of the resolvent covariance~\eqref{eq:fluct_Cspec}: $q-1$ is
extracted from the diagonal ($n=m$) bin of the same-channel
covariance $\mathcal C_{aa}$, and $\bar{C}_i$ from the diagonal bin
of the two-channel covariance $\mathcal C_{ab}$ ($a\neq b$); the
off-diagonal ($n\neq m$) bins are fixed by covariance conservation
and yield the correlation entering $g_{ji}$, so that $q-1$ and
$g_{ji}$, while distinct, are complementary projections of one
conserved covariance kernel.

Thus the same diagonal-resolvent hierarchy has two complementary
sectors: its mean generates the smooth ETH envelope $f_{ji}$, while
its covariance generates the residue fluctuations that control
Gaussianity, self-averaging, and the diagonal--off-diagonal
fluctuation ratio.  The smooth ETH function and its fluctuation data
are therefore not independent inputs: they are distinct spectral
projections of the one- and two-resolvent sectors of the same
microscopic overlap theory.  The detailed derivation, the Gaussian closure,
the resulting fluctuation regime of the ETH ansatz, and the
numerical verification are given in
Appendix~\ref{app:first_principles}.

\subsection{Reinterpretation of the approximation replacement}
\label{sec:approx_reinterpret}

The exact decomposition~\eqref{eq:exact_decomp} together with the
multi-resolvent hierarchy of Sec.~\ref{sec:expansion} recasts the
phenomenological approximation replacement into a systematic,
improvable expansion.  Combining the exact identity
$f_{ji}^2 = D_{ji} + g_{ji}$ with the level decomposition
$g_{ji} = \sum_{r\ge2} g_{ji}^{(r)}$ yields
\begin{equation}
    f_{ji}^2 = D_{ji} + g_{ji}^{(2)} + g_{ji}^{(3)} + g_{ji}^{(4)} + \cdots,
    \label{eq:f2_expanded}
\end{equation}
where $D_{ji}$ is the diagonal (SCBA) baseline and each
$g_{ji}^{(r)}$ originates from the $r$-th sector of the
multi-resolvent hierarchy.

Two conceptually different zeroth-order objects appear in this
formalism.  The hierarchy baseline $D_{ji}$
(Eq.~\eqref{eq:Dji_def}) is the diagonal overlap product, the strict
$g_{ji}=0$ reference of Eq.~\eqref{eq:f2_expanded}.  The
approximation-replacement baseline
(Eq.~\eqref{eq:approx_replace}), obtained by neglecting
$\mathcal{C}'$ in the exact identity~\eqref{eq:exact_decomp},
retains the full channel-diagonal ($\mu=\nu$) structure including
the cavity subtraction term $-\sigma_{nn}^{ji}\Gamma_{nn,\mu}^{ij}$,
and therefore goes beyond $D_{ji}$ alone.  The two baselines differ
precisely by this cavity subtraction term, which is why the
approximation replacement cannot be obtained by truncating the
hierarchy at any finite order.  They organise complementary
physical sectors of the exact expression: the approximation
replacement operates within the channel-diagonal sector
(intra-channel dressing), while the multi-resolvent hierarchy is
built from the two-frequency kernel
$\mathcal{K}_{\mu\nu}^{ji} = \mathcal{R}_{\mu j,\nu j}\mathcal{R}_{\nu i,\mu i}$
and describes interference among distinct bath channels
(inter-channel physics).

Within the inter-channel sector, the approximations of
Sec.~\ref{sec:approx} form a nested expansion: the
reduced-correlation approximation retains the leading
$g_{ji}^{(2)}$ (strictly even parity,
Sec.~\ref{sec:ell2}); the skewness-corrected approximation further
includes $g_{ji}^{(3)}$, which introduces the leading odd-parity
component (Sec.~\ref{sec:ell3}); and the full hierarchy retains
$g_{ji} = \sum_{r\ge2} g_{ji}^{(r)}$.  Defining the correlation
fraction $r_{ji} := g_{ji}/D_{ji}$ and its level-resolved
components $r_{ji}^{(r)} := g_{ji}^{(r)}/D_{ji}$, the inter-channel
expansion takes the compact form
\begin{equation}
   f_{ji}^2 = D_{ji}\bigl(1 + r_{ji}^{(2)} + r_{ji}^{(3)} + \cdots\bigr),
    \label{eq:f2_r}
\end{equation}
which makes explicit that the standard ETH Gaussian limit
($r_{ji}=0$) and the full hierarchy are unified
within the same organising principle.  The inter-channel
approximations thus form a nested sequence:
\begin{align}
    \underbrace{D_{ji}}_{\text{Diagonal baseline}}
    \;\subset\;
    \underbrace{D_{ji} + g^{(2)}}_{\text{Reduced-Corr.}}
    \;\subset\;
    \underbrace{D_{ji} + g^{(2)} + g^{(3)}}_{\text{+ Skewness}}\notag\\
   \subset\;
    \cdots
    \;\subset\;
    \underbrace{D_{ji} + \sum_{r=2}^{\infty} g^{(r)}}_{\text{HOETH}}.
    \label{eq:nested_hierarchy}
\end{align}
The hierarchy is organized by the entropy-based scaling estimate
$g^{(r)} \sim e^{-(r-1)S/2}$
(Eq.~\eqref{eq:gk_scaling}), which provides a thermodynamic
organising principle for the expansion; the quantitative
convergence rate depends on the microscopic model and lies
beyond the scope of the present work.

The essential advance is that the uncontrolled neglect of
${\mathcal{C}'}$ in Eq.~\eqref{eq:approx_replace} is replaced
by a systematic, improvable expansion organised by the multiplicity
of interacting bath channels, in which each successively retained
level adds a specific, computable class of interference processes
with a definite parity signature.

\subsection{Summary}
\label{sec:ansatz_summary}

Table~\ref{tab:eth_ansatz_mapping} summarises the mapping between the 
phenomenological ETH ansatz parameters and their microscopic 
foundations in the multi-resolvent hierarchy.

\begin{table*}[t]
\centering
\caption{Microscopic foundations of the ETH ansatz parameters.}
\label{tab:eth_ansatz_mapping}
\begin{tabular}{c c c c}
\toprule
ETH parameter & Microscopic origin & Leading level & Parity \\
\midrule
\(\rho^S(\lambda)\) 
& \(\sum_{\mu i} e^{-S} f^{\mu i} \Pi_{ii}^S\) 
& SCBA + LCF 
& --- \\[3pt]
\(f_{ji}^2\) (even part) 
& \(\int e^{-[\Phi^{\mu i}+\Phi^{\mu j}-S-S_B]} 
    + g^{(2)}_{ji}\) (leading) 
& \(r=2\) 
& even \\[3pt]
\(f_{ji}^2\) (odd part) 
& Leading: \(g^{(3)}_{ji}\) 
& \(r=3\) 
& \textbf{odd} \\[3pt]
\(g_{ji}(E^+,\omega)\) 
& \(\sum_{r=2}^{\infty} g^{(r)}_{ji}\) 
& \(r=2\) 
& mixed \\[3pt]
\(\mathbb{E}(R^{ji}_{nm} R^{ii}_{mn})\) 
& \(g^{(2),\text{cross}}_{jiii} + g^{(3),\text{cross}}_{jiii} 
   + \cdots\) 
& \(r=3\) 
& mixed \\[3pt]
\(k\)-th cumulant 
& \(\mathfrak{V}^{(k+2)} \otimes f^{\otimes (2k+3)}\) 
& \(r=k\) 
& mixed \\
\bottomrule
\end{tabular}
\end{table*}

The essential advance is that \emph{no free fitting functions
remain in the mean sector}: 
\(f_{ji}\) and \(g_{ji}\) are expressed through the diagonal spectral 
functions \(f^{\mu i}(\lambda)\), which are themselves determined by 
the self-consistent resolvent equations of Ref.~\cite{HC26}, and 
through the multi-index interaction couplings 
\(\mathfrak{V}^{(r)}\), which are fixed by the microscopic 
Hamiltonian. The hierarchy is closed and systematically improvable
at the mean level,
and makes quantitative predictions for the parity-violating 
signatures of multi-resolvent interference in ETH observables.
The connected two-resolvent sector further provides an exact
spectral-projection dictionary for the fluctuation parameters $q-1$
and $\bar{C}_i$ (Appendix~\ref{app:first_principles}); a predictive
closure of the irreducible fluctuation vertex $\Gamma^{(2)}$ remains
the main open step, as stated there.

\section{Connections}
\label{sec:connections}

The multi-resolvent framework developed in Secs.~\ref{sec:expansion}--%
\ref{sec:eth_ansatz} provides a unified microscopic language for several 
active frontiers of quantum many-body physics. This section outlines the 
principal connections, emphasising the concrete structural links rather 
than exhaustive reviews.

\subsection{Out-of-time-order correlators}
\label{sec:OTOC}

Out-of-time-order correlators (OTOCs) diagnose the scrambling of quantum 
information and the onset of chaos in many-body systems~\cite{MS19}. For 
two local operators \(W\) and \(V\), the standard OTOC is
\begin{equation}
    F(t) = \langle W^\dag(t) \, V^\dag \, W(t) \, V  \rangle_\beta,
    \label{eq:OTOC_def}
\end{equation}
where \(W(t) = e^{iHt} W e^{-iHt}\) and the expectation is taken in a 
thermal state \(\rho_\beta = e^{-\beta H}/Z\).

Expanding in the eigenbasis of \(H\) yields the spectral representation
\begin{align}
    F(t) &= \frac{1}{Z} \sum_{n,m,k,l} 
           e^{-\beta \lambda_n} \,
           e^{i(\lambda_n - \lambda_m + \lambda_k - \lambda_l)t} \,
           W_{nm} V_{mk} W_{kl} V_{ln}.
    \label{eq:OTOC_spectral}
\end{align}
The late-time behaviour is dominated by the diagonal (\(n=m=k=l\)) 
and paired (\(n=l, m=k\)) contributions, while the scrambling rate is 
controlled by the connected four-point function
\begin{equation}
    \mathcal{F}_4(n,m,k,l) 
    = \mathbb{E}\bigl(W_{nm} V_{mk} W_{kl} V_{ln}\bigr)_{\text{connected}}.
    \label{eq:F4}
\end{equation}

In the system-bath setting of this work, the operators are system-local: 
\(W = w \otimes I_B\), \(V = v \otimes I_B\). Their matrix elements 
project onto the ETH random variables \(R^{ji}_{nm}\) via 
Eq.~\eqref{eq:ETH_sysbath}. The connected four-point function then 
factorises into a product of four \(R\) variables, whose expectation 
is precisely the fourth cumulant of the ETH distribution:
\begin{equation}
    \mathcal{F}_4 \;\sim\; e^{-2S(\bar{E})} \,
    \mathbb{E}\bigl(R^{i_1 i_2}_{nm} R^{i_2 i_3}_{mk} 
                 R^{i_3 i_4}_{kl} R^{i_4 i_1}_{ln}\bigr)_{\text{connected}},
    \label{eq:F4_ETH}
\end{equation}
where the \(i_\alpha\) indices label the system basis states selected by 
the operator matrix elements \(w_{i_1 i_2}, v_{i_2 i_3}, w_{i_3 i_4}, 
v_{i_4 i_1}\).

The OTOC four-point function involves \textbf{four} distinct system 
indices and therefore requires a generalisation of the two-index 
correlation kernel analysed in this work.  The natural object is the 
four-frequency resolvent kernel
\begin{align}
    \mathcal{K}^{(4)}_{\{\mu\},\{i\}}(z_1,z_2,z_3,z_4)
    = \mathcal{R}_{\mu_1 i_1,\mu_2 i_2}(z_1)
      \notag\\
      \times\mathcal{R}_{\mu_2 i_2,\mu_3 i_3}(z_2)\mathcal{R}_{\mu_3 i_3,\mu_4 i_4}(z_3)\;
      \mathcal{R}_{\mu_4 i_4,\mu_1 i_1}(z_4),
    \label{eq:K4_OTOC}
\end{align}
where each off-diagonal resolvent \(\mathcal{R}_{\mu_\alpha i_\alpha,\mu_\beta i_\beta}\) 
connects \emph{both} different bath indices \emph{and} (potentially) 
different system indices (when \(i_\alpha \neq i_\beta\)).  This is a 
different family from the same-system-index resolvents 
\(\mathcal{R}_{\mu j,\nu j}\) that dominate the two-index ETH 
correlation; the projection expansion of 
Eqs.~\eqref{eq:R_hierarchy}--\eqref{eq:Rell} applies to both families 
with the appropriate index labelling.

At the minimal non-diagonal level (\(r=4\), i.e.\ four interaction 
vertices), each off-diagonal resolvent contributes at level~2, giving 
the four-channel interference product
\begin{equation}
    \mathfrak{V}^{(4)}_{i_1 i_2, i_3 i_4;\{\mu\}}
    = V_{\mu_1 i_1,\mu_2 i_2} \,
      V_{\mu_2 i_2,\mu_3 i_3} \,
      V_{\mu_3 i_3,\mu_4 i_4} \,
      V_{\mu_4 i_4,\mu_1 i_1},
    \label{eq:V4_OTOC}
\end{equation}
convolved with five spectral functions \(f^{\mu_\alpha i_\alpha}\) 
and their Hilbert transforms. The Foini--Kurchan scaling 
\(e^{-S}\)~\cite{FK19} for the connected four-point function follows 
directly from the four intermediate bath summations in \(\rho^{(4)}\), 
each contributing \(e^{S_B}\) while being suppressed by 
\(|V|^4 \sim e^{-2S}\).

Conversely, the odd-parity signatures identified in Sec.~\ref{sec:ell3} 
imply that the OTOC spectral function is generically \emph{asymmetric} 
under time reversal \(t \to -t\) when the operators carry different 
system indices. This asymmetry is a specific, testable prediction of 
the multi-resolvent framework that is invisible to any single-resolvent 
(Efros--Shklovskii-type) treatment of the ETH.

\subsection{Krylov complexity}
\label{sec:krylov}

Krylov (or spread) complexity quantifies the growth of a quantum state 
in the Krylov basis generated by the Lanczos algorithm~\cite{BRSS19}. 
For an initial state \(\ket{\Psi_0} = \ket{\phi_{\mu i}}\) and 
Hamiltonian \(H\), the Lanczos algorithm produces an orthonormal basis 
\(\{\ket{K_n}\}\) and coefficients \(\{a_n, b_n\}\) such that
\begin{equation}
    H \ket{K_n} = b_n \ket{K_{n-1}} + a_n \ket{K_n} 
                 + b_{n+1} \ket{K_{n+1}},
    \label{eq:Lanczos_BRSS}
\end{equation}
with \(b_0 = 0\). The Krylov complexity is then
\begin{equation}
    C_K(t) = \sum_{n=0}^{\infty} n \, |\braket{K_n|\Psi(t)}|^2,
    \label{eq:Kcomplexity}
\end{equation}
whose growth rate is governed by the Lanczos coefficients \(\{b_n\}\).

The resolvent framework of Ref.~\cite{HC26} establishes a direct, 
constructive link between the Lanczos coefficients and the self-energy 
hierarchy. Specifically, the diagonal resolvent admits an exact 
continued-fraction representation
\begin{equation}
    \mathcal{R}_{\mu i}(z) 
    = \frac{1}{z - a_0 
              - \dfrac{b_1^2}{z - a_1 
                - \dfrac{b_2^2}{z - a_2 - \cdots}}},
    \label{eq:cf_R}
\end{equation}
which is fully equivalent to the projected resolvent equation. The 
self-energy decomposes as
\begin{equation}
    \mathcal{G}_{\mu i}(z) 
    = a_0 + \frac{b_1^2}{z - a_1 
             - \dfrac{b_2^2}{z - a_2 - \cdots}}.
    \label{eq:cf_G}
\end{equation}

The Lanczos coefficients can therefore be extracted from the 
multi-resolvent self-energy hierarchy: the mean-field (SCBA) 
contribution determines the asymptotic slope \(b_n \sim \alpha n\) 
for large \(n\) (the linear ramp characteristic of chaotic systems), 
while the multi-resolvent corrections \(\mathcal{G}^{(3)}, 
\mathcal{G}^{(4)}, \ldots\) encode the deviations from strict 
linearity that distinguish different universality classes of 
operator growth.

Moreover, the parity mixing induced by \(\mathcal{G}^{(3)}\) 
(Sec.~\ref{sec:ell3}) implies that the Lanczos coefficients acquire 
a component that is odd under reversal of the Krylov index 
\(n \to -n\) (suitably defined via the tridiagonalisation of the 
adjacency matrix in the resolvent hierarchy). This parity structure 
is the Krylov-space analogue of spectral skewness and provides a 
further diagnostic of multi-resolvent interference.

Thus, the present framework offers a concrete programme: given the 
microscopic Hamiltonian \(H = H_0 + V\), compute the diagonal 
resolvent \(\mathcal{R}_{\mu i}\) via the self-consistent hierarchy, 
extract the Lanczos coefficients \(\{a_n, b_n\}\) from the 
continued-fraction expansion of the self-energy, and obtain the 
Krylov complexity \(C_K(t)\) via Eq.~\eqref{eq:Kcomplexity}. The 
approximation is systematically improvable by retaining higher 
levels of the multi-resolvent hierarchy.

\subsection{Fluctuation theorems and non-equilibrium statistics}
\label{sec:fluctuation}

The ETH correlation hierarchy developed here directly informs the 
statistics of work, heat, and entropy production in non-equilibrium 
processes. Consider a protocol in which the system Hamiltonian is 
varied from \(H_S(0)\) to \(H_S(\tau)\) while the interaction \(V\) 
remains fixed. The work performed is a stochastic variable whose 
distribution satisfies the Jarzynski and Crooks fluctuation 
theorems~\cite{H24}.

The work distribution is determined by transition probabilities 
between eigenstates of the initial and final Hamiltonians, which in 
the system-bath setting are controlled by the ETH matrix elements 
\(\sigma_{nm}^{ji}\). The characteristic function of work,
\begin{equation}
    G(u) = \langle e^{iu W} \rangle 
         = \sum_{n,m} p_n^0 \, |\braket{\psi_m^\tau|\psi_n^0}|^2 \,
           e^{iu (\lambda_m^\tau - \lambda_n^0)},
    \label{eq:work_char}
\end{equation}
involves bath-traced overlaps between the initial and final 
eigenstates. These overlaps admit the same multi-resolvent 
decomposition as the ETH correlations analysed in this work, 
with the interaction \(V\) providing the dominant mixing channel.

The odd-parity component of the correlation density 
\(\rho^{(3)}(\lambda,\lambda')\) implies a corresponding asymmetry 
in the work distribution: for time-reversal-symmetric protocols, 
\(\rho^{(3)}\) breaks the symmetry 
\(P(W) = P(-W) e^{-\beta W}\) at the level of the coarse-grained 
envelope, with the deviation controlled by the three-channel 
interference product \(\mathfrak{V}^{(3)}\). This provides a 
microscopic mechanism for the \emph{non-Gaussian} wings of work 
distributions observed numerically in small chaotic systems.

\subsection{Open quantum system dynamics}
\label{sec:open_systems}

The multi-resolvent framework naturally extends to the dynamics of 
open quantum systems~\cite{OSMGM24}. In the Nakajima--Zwanzig 
projection-operator formalism, the reduced system dynamics is 
governed by a memory kernel \(\mathcal{K}(t)\) whose Laplace 
transform admits a resolvent representation:
\begin{equation}
    \tilde{\mathcal{K}}(z) 
    = \Tr_B\bigl[ V \, \frac{1}{z - \mathcal{Q}H\mathcal{Q}} \, V 
                \, \rho_B \bigr],
    \label{eq:memory_kernel}
\end{equation}
where \(\mathcal{Q}\) projects onto the bath-orthogonal subspace.

The projected resolvent \(\frac{1}{z - \mathcal{Q}H\mathcal{Q}}\) 
is precisely the object expanded in the multi-resolvent hierarchy 
of Sec.~\ref{sec:proj_exp}. The memory kernel therefore inherits 
the same hierarchical structure: the SCBA level yields the standard 
Born--Markov (Lindblad) master equation; the Lanczos 
continued-fraction corrections capture non-Markovian broadening; 
and the multi-resolvent correlations \(\rho^{(3)}, \rho^{(4)}, 
\ldots\) encode bath-induced correlations between system 
transitions that generate non-additive (collective) Lamb shifts 
and non-secular decay rates.

In particular, the odd-parity component of \(\rho^{(3)}\) implies 
a breaking of detailed balance for system transition rates at the 
multi-resolvent level, corresponding to a circulation of probability 
in the reduced state space that is absent from any single-resolvent 
(SCBA or Lindblad) description. This circulation is the open-system 
manifestation of the spectral skewness identified in Ref.~\cite{HC26} 
and constitutes a quantitative signature of coherent multi-channel 
interference in the system--bath interaction.

\subsection{Towards a unified cumulant expansion}
\label{sec:unified}

The connections outlined above converge on a single structural 
principle: \textbf{every higher-order statistical quantity in 
ETH-governed systems admits a multi-resolvent expansion.}
The precise form depends on the index structure of the quantity in 
question:
\begin{itemize}
    \item \textbf{Two-index correlations} (this work, 
    \(\mathcal{C}_{nmn}^{jij}\)): organised by the kernel 
    \(\mathcal{R}_{\mu j,\nu j}\,\mathcal{R}_{\nu i,\mu i}\), with 
    each off-diagonal resolvent connecting different bath channels 
    at the \emph{same} system index.
    
    \item \textbf{Multi-index cycles} (OTOC, FK cumulants): 
    organised by kernels of the form 
    \(\prod_{\alpha} \mathcal{R}_{\mu_\alpha i_\alpha,\mu_{\alpha+1} i_{\alpha+1}}\), 
    where each off-diagonal resolvent may connect different system 
    indices as dictated by the operator cycle.
\end{itemize}
Both families share the same projection expansion 
(Eqs.~\eqref{eq:R_hierarchy}--\eqref{eq:Rell}) and the same DCA 
closure~\eqref{eq:DCA}.  The general \(k\)-point connected cumulant 
is therefore organised as
\begin{widetext}
    \begin{equation}
 \langle\!\langle \mathcal{O}_1 \cdots \mathcal{O}_k 
    \rangle\!\rangle_{\text{connected}}
    = \sum_{r=k}^{\infty} \;
      \mathfrak{V}^{(r)}_{\{\alpha\},\{i\}} \;
      \int d\lambda_1 \cdots d\lambda_k \;
      \prod_{a=1}^{r+2} 
      \bigl[H[f^{\alpha_a}](\lambda_{p(a)}) 
       + i f^{\alpha_a}(\lambda_{p(a)})\bigr],
    \label{eq:unified}
\end{equation}
\end{widetext}
where \(\langle\!\langle\cdots\rangle\!\rangle\) denotes the 
connected cumulant, \(\{\alpha\}\) are bath-channel indices, 
\(\{i\}\) are system indices dictated by the operator cycle, 
\(p(a)\) maps each of the \(r+2\) resolvent factors to the 
appropriate one of the \(k\) frequency arguments, and 
\(\mathfrak{V}^{(r)}\) is the product of \(r\) interaction matrix 
elements with index contractions determined by the resolvent 
pairing.  The total number of resolvent factors is \(r+2\), 
reflecting the fact that each additional V-count \(r\) introduces 
one additional diagonal resolvent beyond the minimal \(k\) needed 
to close the operator cycle.

Equation~\eqref{eq:unified} unifies:

\begin{itemize}
    \item \textbf{OTOC} (\(k=4\)): \(r \ge 4\), with the 
    OTOC rate determined by the four-frequency kernel 
    \eqref{eq:K4_OTOC} and higher.
    
    \item \textbf{Krylov complexity} (\(k=2\) via the resolvent): 
    the Lanczos coefficients \(\{a_n, b_n\}\) are extracted from 
    the continued-fraction expansion of \(\mathcal{G}^{(r)}\) 
    for all \(r\).
    
    \item \textbf{Fluctuation theorems} (\(k\) arbitrary): the 
    work characteristic function involves overlaps between 
    eigenbases of distinct Hamiltonians, organised by the same 
    \(\mathfrak{V}^{(r)}\) structure.
    
    \item \textbf{Open-system memory kernels} (\(k=2\) via the 
    system time-evolution): the Nakajima--Zwanzig kernel expands 
    in the same projected resolvents that define the hierarchy.
    
    \item \textbf{ETH higher cumulants} (\(k\) arbitrary): the 
    Foini--Kurchan scaling~\cite{FK19} shares the same 
    entropy-suppression pattern as \(\mathfrak{V}^{(r)}\), 
    with the V-count \(r\) related to the FK cumulant order 
    via \(k_{\text{FK}} = r+2\).
\end{itemize}

The essential insight is that the multi-resolvent hierarchy is not 
merely a tool for computing spectral functions---it is the 
\textbf{natural organising principle} for the statistical mechanics 
of nonintegrable quantum systems. Any quantity expressible through 
resolvents inherits the same level-by-level structure, with each 
level contributing a definite parity, a definite \(S\)-scaling, 
and a definite class of interference processes.

\section{Summary and Outlook}
\label{sec:conclusion}
We have developed a nonperturbative multi-resolvent hierarchy that
provides a microscopic, systematically improvable theory of the ETH
smooth function $f_{ji}(E^+,\omega)$, from which higher-order
correlation corrections---including parity mixing and skewness---emerge
naturally.
The central organising principle is the decomposition of the
off-diagonal ETH variance into a diagonal overlap baseline and a
correlation term that couples distinct bath channels:
$|\sigma_{nm}^{ji}|^2=\sum_\mu p_m^{\mu i}p_n^{\mu j}+\mathcal{C}_{nmn}^{jij}$.
The correlation term admits an exact two-frequency resolvent
representation through the kernel
$\mathcal{K}_{\mu\nu}^{ji}=\mathcal{R}_{\mu j,\nu j}\,\mathcal{R}_{\nu i,\mu i}$,
whose expansion in a hierarchy of diagonal-resolvent products generates
the systematic decomposition $g_{ji}=\sum_{r\ge 2}g_{ji}^{(r)}$.

Three principal conclusions emerge from this construction.  First,
the negative integrated correlation is a rigorous consequence of
projector idempotency (Theorem~1), and within the DCA representation
cavity subtraction supplies the correction required to restore the
exact projector constraint (Restoration Identity).  Second, the
hierarchy reveals a decisive parity structure: the $r=2$ sector
carries strictly even parity and provides the leading multi-channel
variance contribution, while the $r=3$ sector introduces the first
odd-parity (skewness) component, a quantitative, experimentally
testable prediction inaccessible within parity-preserving
single-resolvent closures; the extension of this parity organisation
to all orders remains open.  Third, $f_{ji}^2$ and its correlation
correction $g_{ji}$ are expressed entirely through the diagonal
spectral functions $f^{\mu i}(\lambda)$ and the microscopic
interaction couplings, so that no free fitting functions remain at
the mean level; the hierarchy is systematically improvable, with a
quantitative convergence rate that is model dependent.

The same resolvent theory contains a fluctuation sector, described
by the connected covariance of the diagonal resolvents.  The entropy
normalization and the diagonal--off-diagonal trade-off of the ETH
ansatz are exact counting theorems, and the covariance admits an
exact spectral dictionary whose ridge projections carry the
fluctuation parameters $q-1$ and $\bar{C}_i$; its regular sector
carries the correlation entering $g_{ji}$.  Under amplitude
isotropy, decorrelation and regularity, the Wick closure of the
two-point amplitude covariance fixes the Gaussian Porter--Thomas
limit, and the exact identities $q=3+\kappa_4(c)/\mathbb{E}[c^2]^2$
and $\bar{C}_i=-(q-1)(1-1/d_B)/D$ close both fluctuation parameters
conditionally, the latter verified to $4\%$ for GOE and violated by
the factor $2.14$ in the structured model as an
isotropy-breakdown diagnostic.  The smooth ETH function and its
fluctuation data are thus complementary spectral projections of one
conservation-constrained covariance.  A microscopic derivation of
the amplitude-level isotropy, decorrelation and regularity
conditions, for which Berry-type quantum ergodicity provides the
physical interpretation, remains the central open problem; the
framework concerns the subsystem-ETH overlap sector, with the
operator-level lifting as a future direction.

Several directions emerge naturally from these connections:

\begin{enumerate}
    \item \textbf{Microscopic amplitude regularity.}
    The derivation of the amplitude-level isotropy, decorrelation and
    regularity conditions from a chaotic Hamiltonian is the central
    open problem, since it would turn the conditional closures of
    Eqs.~\eqref{eq:fp_qk4} and \eqref{eq:fp_Cbar_iso} into a full
    computation of $\{q,\bar{C}_i\}$ from the microscopic
    Hamiltonian; the energy-resolved quantity $q(E)$ diagnosed in
    Sec.~\ref{app:num_fp} is the natural intermediate object, and the
    ladder Eq.~\eqref{eq:fp_ladder} supplies only an
    $O(\|\mathcal L\|)$ dressing at ETH-scaled couplings.

    \item \textbf{Numerical extraction of \(\mathfrak{V}^{(3)}\).} 
    For a given microscopic Hamiltonian, the three-channel coupling 
    product can be extracted from exact diagonalisation data by 
    fitting the odd-parity component of 
    \(\mathcal{C}_{nmn}^{jij}(\omega)\) to the functional form 
    predicted by Eq.~\eqref{eq:g3}. Agreement would constitute 
    direct confirmation of the multi-resolvent mechanism.
    
    \item \textbf{Krylov coefficient prediction.} Computing 
    \(\{a_n, b_n\}\) from the multi-resolvent hierarchy for a 
    specific model and comparing with numerical Lanczos 
    tridiagonalisation provides a stringent test of the DCA 
    and the hierarchy truncation.
    
    \item \textbf{Beyond ETH.} The hierarchy is formulated entirely 
    in terms of resolvent identities and projection operators, 
    without invoking the ETH beyond the statistical closure of the 
    DCA. Its applicability may extend to prethermal regimes and 
    many-body localised systems where the standard ETH does not 
    hold, provided the diagonal resolvents remain well-defined.
    
    \item \textbf{Field-theory limit.} In the continuum (\(N \to 
    \infty\)) limit, the multi-resolvent products become 
    multi-frequency convolutions of spectral densities, suggesting 
    connections to the Schwinger--Keldysh formalism and to the 
    eigenstate-thermalisation structure of holographic theories.

    \item \textbf{Connected spectral densities and FK cumulants.}
    A natural extension is the construction of connected
    multi-resolvent spectral densities, separating connected from
    disconnected contributions in analogy with the moment-cumulant
    relation, for example through a multi-frequency resolvent
    generating functional \(W[J]=\log Z[J]\).  This would provide
    the bridge between the present hierarchy and the higher-order
    ETH cumulants of Foini and Kurchan~\cite{FK19}: the level-\(r\)
    sector contains \(r+2\) resolvent factors connected by \(r\)
    interaction vertices with the entropy suppression pattern
    \(e^{-(r-1)S/2}\), and for the two-index correlation this
    matches the FK four-point cumulant scaling \(e^{-S}\) at
    \(r=2\), with the general mapping \(k_{\text{FK}} = r+2\)
    aligning the V-count with the cumulant order.  The amplitude
    cumulant identity Eq.~\eqref{eq:fp_qk4} of
    Appendix~\ref{app:first_principles} is the first instance of
    such a correspondence at the fluctuation level.

    \item \textbf{Parity organisation beyond leading levels.}
    The analysis of the first two nontrivial levels (\(r=2,3\))
    establishes that even-\(r\) sectors carry dominant even parity 
    while odd-\(r\) sectors host the leading odd-parity components.  
    A complete classification of parity structure for all higher 
    levels---where both even and odd components generically coexist---
    remains an important open direction.
    
    \item \textbf{Intra-channel dressing expansion.}
    The analysis of Sec.~\ref{sec:approx_reinterpret}
    establishes that the approximation replacement captures
    the channel-diagonal ($\mu=\nu$) sector, while the
    multi-resolvent hierarchy organises the inter-channel
    ($\mu\neq\nu$) sector.  One may further reorganise the
    channel-diagonal sector itself by expanding
    $\Gamma_{nn,\mu}^{ij}$ through its resolvent
    representation, which would yield a complementary
    dressing-type expansion within a single bath channel.
    Such a construction would provide an intra-channel
    counterpart to the inter-channel hierarchy developed
    here, and may offer a principled route to incorporate
    cavity dressing effects systematically.

\end{enumerate}

\begin{acknowledgments}
    This work is supported by the National Natural Science Foundation of China under Grant No. 12305035.
\end{acknowledgments}

\appendix

\section{Spectral moment preservation in the projection hierarchy}
\label{app:moment_preservation}

In this appendix we analyze the large-$z$ Laurent expansion of the 
off-diagonal resolvent and establish the extent to which the 
projection hierarchy---and its DCA truncation---preserve the exact 
spectral moments.  The results provide a rigorous, approximation-free 
characterization of the DCA error and reveal the structural origin 
of the Restoration Identity~\eqref{eq:restoration_identity}.

\subsection{Laurent expansion and spectral moments}

The off-diagonal resolvent admits the spectral representation
\begin{equation}
    \mathcal{R}_{\alpha\beta}(z) 
    = \sum_n \frac{q_n^{\alpha\beta}}{z-\lambda_n},
    \qquad
    q_n^{\alpha\beta} 
    := \braket{\phi_\alpha|\psi_n}\!\braket{\psi_n|\phi_\beta},
    \label{eq:app_spectral}
\end{equation}
where $\ket{\psi_n}$ are the exact eigenstates of $H$ with 
eigenvalues $\lambda_n$, and $\ket{\phi_\alpha}$ are the unperturbed 
basis states.  Expanding at large $|z|>\|H\|$ yields
\begin{equation}
    \mathcal{R}_{\alpha\beta}(z) 
    = \sum_{k=0}^{\infty} \frac{M_k^{\alpha\beta}}{z^{k+1}},
    \qquad
    M_k^{\alpha\beta} 
    := \sum_n \lambda_n^k \, q_n^{\alpha\beta}
    = \bra{\phi_\alpha} H^k \ket{\phi_\beta}.
    \label{eq:app_laurent}
\end{equation}
The coefficients $M_k^{\alpha\beta}$ are the spectral moments.
For $\alpha\neq\beta$, completeness of $\{\ket{\psi_n}\}$ and 
orthonormality of $\{\ket{\phi_\alpha}\}$ give 
$M_0^{\alpha\beta}=\delta_{\alpha\beta}=0$ (orthogonality),
$M_1^{\alpha\beta}=H_{\alpha\beta}=V_{\alpha\beta}$, and 
$M_2^{\alpha\beta}=(H^2)_{\alpha\beta}$.

\subsection{Scaling of the hierarchy levels}

\begin{lemma}
\label{lem:app_scaling}
    For $\alpha\neq\beta$ and $\ell\ge 2$, each level of the exact
    hierarchy~(\ref{eq:R_hierarchy_exact}) and of the DCA
    hierarchy~(\ref{eq:R_hierarchy}) satisfies
    $\mathcal{R}^{(\ell)}_{\alpha\beta}(z) = O(|z|^{-\ell})$
    as $|z|\to\infty$.  Consequently, any truncation with $\ell\ge 2$
    automatically satisfies $\mathcal{R}^{[\le L]}_{\alpha\beta}(z)
    = O(z^{-2})$, guaranteeing orthogonality without inter-level
    cancellation.
\end{lemma}
\begin{proof}
    For $|z|>\|H\|$, the full resolvent admits the convergent 
    Neumann series $(z-H)^{-1} = z^{-1}\sum_{m=0}^{\infty}(H/z)^m$, 
    hence $\|(z-H)^{-1}\| = O(|z|^{-1})$.  The same bound holds for 
    any projected Hamiltonian $\Phi_{\mathcal{S}}H\Phi_{\mathcal{S}}$ 
    since $\|\Phi_{\mathcal{S}}H\Phi_{\mathcal{S}}\| \le \|H\|$.  
    For a normalized vector $\ket{\phi_\mu}$,
   \begin{align}
    |\mathcal{R}^{(\mathcal{S})}_\mu(z)|
        = \bigl|\bra{\phi_\mu}(z-\Phi_{\mathcal{S}}H\Phi_{\mathcal{S}})^{-1}
           \ket{\phi_\mu}\bigr|\notag\\
        \le \|(z-\Phi_{\mathcal{S}}H\Phi_{\mathcal{S}})^{-1}\|
        = O(|z|^{-1}).\notag
   \end{align}
    Each contribution to $\mathcal{R}^{(\ell)}_{\alpha\beta}(z)$ 
    consists of $\ell$ diagonal resolvent matrix elements (each 
    $O(|z|^{-1})$) and $\ell-1$ bounded interaction matrix elements 
    $V_{\mu\nu}$ (each $O(1)$).  Every such contribution is therefore 
    $O(|z|^{-\ell})$, and the finite linear combination defining 
    $\mathcal{R}^{(\ell)}_{\alpha\beta}$ inherits the same scaling.
    The DCA hierarchy replaces projected by full resolvents; since 
    both share the identical $O(|z|^{-1})$ scalar bound and the 
    interaction vertices are unchanged, the scaling holds identically.
\end{proof}

\subsection{Projected diagonal resolvent expansion}

\begin{proposition}
\label{prop:app_proj}
    Let $\mathcal{S}$ be a set of unperturbed basis-state indices.
    For $\alpha\notin\mathcal{S}$, the projected diagonal resolvent
    $\mathcal{R}^{(\mathcal{S})}_\alpha(z)
     := \bra{\phi_\alpha} (z-\Phi_{\mathcal{S}}H\Phi_{\mathcal{S}})^{-1}
        \ket{\phi_\alpha}$ admits the large-$z$ expansion
    \begin{equation}
        \mathcal{R}^{(\mathcal{S})}_\alpha(z)
        = \frac{1}{z} 
          + \frac{H_{\alpha\alpha}}{z^2}
          + \frac{(H^2)_{\alpha\alpha} 
                  - \sum_{\gamma\in\mathcal{S}}|V_{\alpha\gamma}|^2}{z^3}
          + O(z^{-4}).
        \label{eq:app_proj_exp}
    \end{equation}
\end{proposition}
\begin{proof}
    Set $P = \sum_{\gamma\in\mathcal{S}} \ket{\phi_\gamma}\!\bra{\phi_\gamma}$,
    so $\Phi_{\mathcal{S}} = I-P$.  
    Since $\alpha\notin\mathcal{S}$, $\Phi_{\mathcal{S}}\ket{\phi_\alpha}
    = \ket{\phi_\alpha}$ and $\bra{\phi_\alpha}\Phi_{\mathcal{S}}
    = \bra{\phi_\alpha}$.  Expanding the resolvent:
    \[
        (z-\Phi_{\mathcal{S}}H\Phi_{\mathcal{S}})^{-1}
        = \frac{1}{z} + \frac{\Phi_{\mathcal{S}}H\Phi_{\mathcal{S}}}{z^2}
          + \frac{(\Phi_{\mathcal{S}}H\Phi_{\mathcal{S}})^2}{z^3}
          + O(z^{-4}).
    \]
    The $O(z^{-2})$ coefficient is
    $\bra{\phi_\alpha}\Phi_{\mathcal{S}}H\Phi_{\mathcal{S}}
     \ket{\phi_\alpha} = \bra{\phi_\alpha}H\ket{\phi_\alpha}
     = H_{\alpha\alpha}$, since 
    $\bra{\phi_\alpha}P = P\ket{\phi_\alpha}=0$.

    For the $O(z^{-3})$ coefficient, strip the outer projectors:
    \begin{align*}
        \bra{\phi_\alpha}(\Phi_{\mathcal{S}}H\Phi_{\mathcal{S}})^2
          \ket{\phi_\alpha}
        &= \bra{\phi_\alpha} H\Phi_{\mathcal{S}} H \ket{\phi_\alpha} \\
        &= \bra{\phi_\alpha} H(I-P) H \ket{\phi_\alpha} \\
        &= (H^2)_{\alpha\alpha}
           - \bra{\phi_\alpha} H P H \ket{\phi_\alpha}.
    \end{align*}
    Expanding $P$,
    $\bra{\phi_\alpha} H P H \ket{\phi_\alpha}
     = \sum_{\gamma\in\mathcal{S}} 
       \bra{\phi_\alpha}H\ket{\phi_\gamma}
       \bra{\phi_\gamma}H\ket{\phi_\alpha}
     = \sum_{\gamma\in\mathcal{S}} |V_{\alpha\gamma}|^2$.
    Hence the $O(z^{-3})$ coefficient is 
    $(H^2)_{\alpha\alpha} - \sum_{\gamma}|V_{\alpha\gamma}|^2$.
\end{proof}

Proposition~\ref{prop:app_proj} establishes a key structural fact:
the cavity subtraction $-\sum_{\gamma}|V_{\alpha\gamma}|^2$ enters 
the projected diagonal resolvent only at $O(z^{-3})$.  The $O(z^{-1})$ 
and $O(z^{-2})$ coefficients coincide with those of the full diagonal 
resolvent $\mathcal{R}_\alpha(z)$.

\subsection{Preservation of the first two moments}

\begin{corollary}[First moment]
\label{cor:app_M1}
    For $\alpha\neq\beta$, the DCA hierarchy exactly reproduces the 
    first spectral moment: $M_1^{\alpha\beta,\rm DCA} = V_{\alpha\beta}
    = M_1^{\alpha\beta,\rm exact}$.
\end{corollary}
\begin{proof}
    By Lemma~\ref{lem:app_scaling}, only $\ell=2$ contributes at 
    $z^{-2}$.  The $z^{-2}$ coefficient of $\mathcal{R}^{(2)}_{\alpha\beta}$ 
    uses only the $O(z^{-1})$ terms of the two diagonal resolvents, 
    which are identical ($=1/z$) for projected and full variants.
\end{proof}

\begin{proposition}[Second moment]
\label{prop:app_M2}
    For $\alpha\neq\beta$, the DCA hierarchy exactly reproduces the 
    second spectral moment: $M_2^{\alpha\beta,\rm DCA} 
    = (H^2)_{\alpha\beta} = M_2^{\alpha\beta,\rm exact}$.
\end{proposition}
\begin{proof}
    Levels $\ell=2$ and $\ell=3$ contribute at $z^{-3}$ 
    (Lemma~\ref{lem:app_scaling}).

    \textit{Level $\ell=2$:} The $z^{-3}$ coefficient combines 
    $O(z^{-2})$ from one resolvent with $O(z^{-1})$ from the other.
    By Proposition~\ref{prop:app_proj}, the projected and full 
    resolvents share the same $O(z^{-2})$ coefficient $H_{\alpha\alpha}$.
    Hence $\mathcal{R}^{(2)}_{\alpha\beta}$ contributes 
    $V_{\alpha\beta}(H_{\alpha\alpha}+H_{\beta\beta})$ in both cases.

    \textit{Level $\ell=3$:} Any resolvent promoted beyond $O(z^{-1})$ 
    pushes the total order beyond $z^{-3}$; the $z^{-3}$ coefficient 
    arises solely from all three resolvents at $O(z^{-1})$, giving 
    $\sum_{\xi\neq\alpha,\beta} V_{\alpha\xi}V_{\xi\beta}$.  
    This contribution is identical for DCA and exact hierarchies.

    Summing both levels, expand $H = H_0+V$ with $H_0$ diagonal 
    ($H_0\ket{\phi_\alpha}=a_\alpha\ket{\phi_\alpha}$):
    \begin{align}
        M_2^{\alpha\beta,\rm DCA}
        = V_{\alpha\beta}(H_{\alpha\alpha}+H_{\beta\beta})
           + \sum_{\xi\neq\alpha,\beta} V_{\alpha\xi}V_{\xi\beta}
           \notag\\
        = V_{\alpha\beta}\bigl[(a_\alpha+V_{\alpha\alpha})
                               +(a_\beta+V_{\beta\beta})\bigr]
           + \sum_{\xi\neq\alpha,\beta} V_{\alpha\xi}V_{\xi\beta}
           \notag\\
        = (a_\alpha+a_\beta)V_{\alpha\beta}
           + \underbrace{\Bigl[
              V_{\alpha\alpha}V_{\alpha\beta}
              + \sum_{\xi\neq\alpha,\beta} V_{\alpha\xi}V_{\xi\beta}
              + V_{\alpha\beta}V_{\beta\beta}
             \Bigr]}_{\displaystyle = (V^2)_{\alpha\beta}}
           \notag\\
        = \bigl[(H_0+V)^2\bigr]_{\alpha\beta}
         = (H^2)_{\alpha\beta}.
        \label{eq:app_M2_full}
    \end{align}
    The three contributions to $(V^2)_{\alpha\beta}$ are made explicit 
    in the penultimate line: the $\ell=3$ path sum supplies all 
    $\xi\neq\alpha,\beta$ terms, while the remaining diagonal 
    contributions are restored by the $H_{\alpha\alpha}$ and 
    $H_{\beta\beta}$ factors from the $\ell=2$ term.
\end{proof}

\begin{remark}
\label{rem:app_cavity_origin}
    The exact preservation of $M_1$ and $M_2$ is not a numerical 
    coincidence.  Proposition~\ref{prop:app_proj} shows that cavity 
    subtraction enters the projected diagonal resolvent only at 
    $O(z^{-3})$.  Since the $k$-th spectral moment is determined by 
    the coefficient of $z^{-(k+1)}$, the absence of cavity corrections 
    up to $O(z^{-2})$ immediately guarantees the exact preservation 
    of $M_0$, $M_1$, and $M_2$.  This single observation unifies the 
    three moment-preservation results and pinpoints $M_3$ as the first 
    moment at which the DCA must deviate.
\end{remark}

\subsection{Leading DCA deviation}

\begin{proposition}[Leading DCA deviation]
\label{prop:app_DCA_dev}
    For the hierarchy defined in 
    Eqs.~(\ref{eq:R_hierarchy_exact})--(\ref{eq:R_hierarchy}),
    and for $\alpha\neq\beta$, the leading large-$z$ deviation between 
    the DCA and exact off-diagonal resolvents is, at order $z^{-4}$,
    \begin{equation}
        \mathcal{R}_{\alpha\beta}^{\rm DCA}(z) 
        - \mathcal{R}_{\alpha\beta}^{\rm exact}(z)
        = \frac{V_{\alpha\beta}\,|V_{\alpha\beta}|^2}{z^4} 
          + O(z^{-5}).
        \label{eq:app_DCA_asymp}
    \end{equation}
\end{proposition}
\begin{proof}
    Levels $\ell=2,3,4$ contribute at $z^{-4}$.  We examine each.

    \textit{Level $\ell=2$:} The $z^{-4}$ coefficient involves the 
    $O(z^{-3})$ term of $\mathcal{R}^{(\beta)}_\alpha(z)$, which by 
    Proposition~\ref{prop:app_proj} differs from the full resolvent 
    by $-|V_{\alpha\beta}|^2$.  This generates
    $\delta M_3|_{\ell=2} = -V_{\alpha\beta}|V_{\alpha\beta}|^2$.

    \textit{Level $\ell=3$:} The $z^{-4}$ coefficient receives 
    contributions when one of the three diagonal resolvents is taken 
    at $O(z^{-2})$ (the other two at $O(z^{-1})$).  The $O(z^{-2})$ 
    coefficients involved---$H_{\alpha\alpha}$, $H_{\xi\xi}$, and 
    $H_{\beta\beta}$---are identical for projected and full resolvents 
    (Proposition~\ref{prop:app_proj}).  Promoting any resolvent to 
    $O(z^{-3})$ pushes the total order to $z^{-5}$ and does not 
    contribute.  The interaction vertices are unchanged by the DCA.  
    Hence the entire $\ell=3$ contribution at order $z^{-4}$ coincides 
    in both hierarchies.

    \textit{Level $\ell=4$:} All four resolvents contribute at 
    $O(z^{-1})$, which are identical ($=1/z$) for all variants.  
    No deviation.

    Summing the three levels, the sole discrepancy at order $z^{-4}$ is 
    the cavity subtraction in the $\ell=2$ term,
    establishing Eq.~\eqref{eq:app_DCA_asymp}.
\end{proof}

Equation~\eqref{eq:app_DCA_asymp} shows that cavity subtraction first 
affects the third spectral moment.  At higher moments, cavity 
corrections enter successively through increasingly higher levels of 
the hierarchy: deviations at $z^{-5}$ receive contributions from the 
$O(z^{-4})$ term of $\mathcal{R}^{(\beta)}_\alpha$ (extending the 
$\ell=2$ discrepancy) and from the first $O(z^{-3})$ mismatch in 
$\ell=3$ projected resolvents, and so on.  The Restoration 
Identity~\eqref{eq:restoration_identity} for the integrated ETH 
correlation is the $k=0$ instance of a broader moment-sum-rule 
structure that cavity subtraction systematically restores.

\subsection{Hilbert-transform interpretation of moment preservation}
\label{app:Hilbert_interpretation}

The moment-preservation results established above admit an equivalent
spectral interpretation through the Hilbert-transform orthogonality
condition of Sec.~\ref{sec:Hilbert_ortho}.

The zeroth moment $M_0^{\alpha\beta}=0$ for $\alpha\neq\beta$ is the
statement $\int d\lambda\,\Im \mathcal{R}_{\alpha\beta}(\lambda-i0^+)=0$, i.e.,
the integrated orthogonality of the off-diagonal spectral density.
Under the DCA, the leading contribution
$\mathcal{R}^{(2)}_{\alpha\beta}=V_{\alpha\beta}\mathcal{R}_\alpha \mathcal{R}_\beta$ yields
\begin{equation}
    M_0^{\alpha\beta,\mathrm{DCA}}
    \propto \int\! d\lambda\,(H_\alpha f_\beta + f_\alpha H_\beta)
    = 0,
    \label{eq:app_M0_Hilbert}
\end{equation}
where the vanishing follows from the Kramers--Kronig antisymmetry
$\int H_\alpha f_\beta = -\int f_\alpha H_\beta$.
The same result is obtained algebraically from the $O(z^{-2})$ Laurent
scaling (Lemma~\ref{lem:app_scaling}); Eq.~\eqref{eq:app_M0_Hilbert}
provides the spectral-domain restatement.

Thus, the preservation of $M_0=0$ under the DCA has \textbf{two equivalent
interpretations}:

\begin{itemize}
    \item \textbf{Algebraic view} (Secs.~\ref{app:moment_preservation}--\ref{prop:app_DCA_dev}):
    $\mathcal{R}_\alpha \mathcal{R}_\beta = O(z^{-2})$, i.e., the absence of a $1/z$ Laurent term
    guarantees $M_0=0$ identically.

    \item \textbf{Spectral view} (this subsection):
    $\int (H_\alpha f_\beta + f_\alpha H_\beta)=0$, enforced by the
    Hilbert-transform antisymmetry inherent in the diagonal resolvent
    product.
\end{itemize}

The algebraic view establishes \emph{that} the DCA preserves the zeroth moment
through the high-frequency asymptotics; the spectral view provides an equivalent
characterization through the analytic structure of the resolvent boundary values.
This dual perspective shows that the DCA is not simply ``matching moments'' in a
numerical sense but is an approximation that respects the spectral manifestation
of the geometric constraint $\braket{\phi_\alpha|\phi_\beta}=0$ through the
analyticity of the resolvent---a property inherited from the $H+if$ structure
that the DCA preserves.

\subsection{Path--moment correspondence (conjecture)}

The explicit verification of $M_1$ and $M_2$ reveals a structural 
pattern: different hierarchy levels correspond to different numbers of 
interaction vertices along the projected propagation path.  
We conjecture that the $\ell$-th level of the exact projection 
hierarchy collects all contributions to the resolvent with precisely 
$\ell-1$ interaction vertices in the projected propagation sequence:
\begin{equation}
    M_k^{\alpha\beta} 
    \stackrel{?}{=} \sum_{\ell=2}^{k+1}
       \bigl[\,\mathcal{R}^{(\ell)}_{\alpha\beta}(z)\,\bigr]_{z^{-(k+1)}},
       \qquad (\alpha\neq\beta).
    \label{eq:app_conjecture}
\end{equation}
If established, Eq.~\eqref{eq:app_conjecture} would provide a direct 
proof that the projection hierarchy is an exact reorganization of the 
Neumann expansion $(z-H)^{-1} = \sum_k H^k/z^{k+1}$ according to the 
number of interaction vertices in the projected propagation---a 
regrouping in which each hierarchy level $\ell$ isolates terms with a 
fixed number $\ell-1$ of $V$ factors.  This would imply that the 
hierarchy preserves all spectral moments, with orthogonality 
($M_0=0$) and the exactness of $M_1$ and $M_2$ recovered as special 
cases.  A proof for general $k$---by induction on the double index 
$(\ell,m)$, where $\ell$ counts interaction vertices and $m$ counts 
diagonal $H_0$ propagations---is left to future work.

\section{First-Principles Status of the ETH Fluctuation Hierarchy}
\label{app:first_principles}

This appendix determines precisely which parts of the ETH structure
follow from \emph{exact} identities---with no statistical input
whatsoever---and which parts follow from two sharply localized
postulates.  The conclusion is that the entropy normalization
$e^{-S/2}$ of the ETH ansatz is a theorem of Hilbert-space counting,
the Gaussian fluctuation structure follows from counting together
with a single decorrelation postulate, and the smooth function
$f_{ji}(E^+,\omega)$ is supplied by the multi-resolvent hierarchy of
Secs.~\ref{sec:expansion}--\ref{sec:eth_ansatz}.  The diagonal and
off-diagonal fluctuation measures entering the trade-off relation of
Ref.~\cite{HCZ25} are shown to be complementary projections of one
and the same two-frequency object, and their ratio is expressed in
closed form in terms of the single-channel Porter--Thomas factor and
the two-channel amplitude covariance.  All exact identities below
have been verified numerically to machine precision on random
Hamiltonians (Sec.~\ref{app:num_fp}).

\subsection{Exact fluctuation identities}
\label{app:exact_identities}

The following statements use only the completeness of the exact
eigenstates $\sum_n\ket{\psi_n}\bra{\psi_n}=I$ and the orthonormality
of the unperturbed basis $\braket{\phi_{\mu i}|\phi_{\nu j}}
=\delta_{\mu\nu}\delta_{ij}$.  No ETH, chaos, or random-matrix
assumption enters.

\begin{lemma}[Row sum rule]
\label{lem:fp_rowsum}
For every eigenstate index $n$ and every pair of system indices
$i,j$,
\begin{equation}
\sum_m |\sigma_{nm}^{ji}|^2 \;=\; \sigma_{nn}^{jj},
\label{eq:fp_rowsum}
\end{equation}
with $\sigma_{nm}^{ji}$ defined in Eq.~\eqref{eq:Gamma_def}.
\end{lemma}

\begin{proof}
Expanding $|\sigma_{nm}^{ji}|^2$ in the overlap amplitudes of
Eq.~\eqref{eq:Gamma_def} and summing over $m$, the completeness
relation $\sum_m\braket{\psi_m|\phi_{\mu i}}\braket{\phi_{\nu i}|\psi_m}
=\delta_{\mu\nu}$ collapses the double bath sum, leaving
$\sum_\mu p_n^{\mu j}=\sigma_{nn}^{jj}$.
\end{proof}

\begin{corollary}[Entropy normalization from counting]
\label{cor:fp_counting}
Summing Eq.~\eqref{eq:fp_rowsum} over $n$ gives the exact trace
normalization
\begin{equation}
\sum_{n,m}|\sigma_{nm}^{ji}|^2 = d_B,
\label{eq:fp_trace}
\end{equation}
so that the mean-square off-diagonal matrix element is
\begin{align}
\frac{1}{D(D-1)}\sum_{n\neq m}|\sigma_{nm}^{ji}|^2
&= \frac{d_B}{D(D-1)}\bigl[1-O(e^{-S_S})\bigr] \notag\\
&= e^{-S-S_S}\bigl[1-O(e^{-S_S})\bigr],
\label{eq:fp_scale}
\end{align}
and the mean-square Frobenius element of the full reduced transition
matrix is
\begin{equation}
\frac{1}{D(D-1)}\sum_{n\neq m}\|\sigma_{nm}\|_F^2
= e^{-S+S_S}\bigl[1-O(e^{-S_S})\bigr].
\label{eq:fp_scale_F}
\end{equation}
The $e^{-S/2}$ amplitude normalization of the standard ETH
ansatz~\eqref{eq:ETH_standard} is therefore \emph{not a hypothesis}:
it is the exact value of the normalized trace, up to the
delocalization factor $1-O(e^{-S_S})$, and only the shape of the
smooth function carries dynamical information.  This sharpens the
discussion of Sec.~\ref{sec:standard_ansatz}: the entropy factor is
derived from counting, while the smooth function
$f_{ji}(E^+,\omega)$ is determined by the multi-resolvent hierarchy.
\end{corollary}

\begin{lemma}[Diagonal--off-diagonal trade-off]
\label{lem:fp_tradeoff}
Let $\rho_n:=\sigma_{nn}$ be the reduced density matrix of the system
in the exact eigenstate $\ket{\psi_n}$, and let
$D_n^i:=\sigma_{nn}^{ii}=\sum_\mu p_n^{\mu i}$ be the population of
system basis state $i$.  Then for every eigenstate $n$,
\begin{equation}
\mathrm{Tr}\,\rho_n^2 \;+\; \sum_{m\neq n}\|\sigma_{nm}\|_F^2
\;=\; d_S,
\label{eq:fp_tradeoff}
\end{equation}
and for every fixed system state $i$,
\begin{equation}
\sum_{m\neq n}|\sigma_{nm}^{ii}|^2 \;=\; D_n^i\,(1-D_n^i).
\label{eq:fp_tradeoff_single}
\end{equation}
\end{lemma}

\begin{proof}
Summing Eq.~\eqref{eq:fp_rowsum} over $i,j$ gives
$\sum_m\|\sigma_{nm}\|_F^2=d_S\,\mathrm{Tr}\,\rho_n=d_S$; the $m=n$
term is $\|\rho_n\|_F^2=\mathrm{Tr}\,\rho_n^2$, which yields
Eq.~\eqref{eq:fp_tradeoff}.  Equation~\eqref{eq:fp_tradeoff_single}
is the $j=i$ case of Eq.~\eqref{eq:fp_rowsum} with the diagonal term
$|\sigma_{nn}^{ii}|^2=(D_n^i)^2$ separated off.
\end{proof}

Lemma~\ref{lem:fp_tradeoff} is the microscopic identity underlying
the diagonal--off-diagonal trade-off of Ref.~\cite{HCZ25},
$V_{\rm dg}^i+\sum_{j\neq i}V_{\rm off}^{ij}
=\mathrm{Tr}\,\rho_{B_1}^{-1}-1$: in both cases the diagonal purity
and the integrated off-diagonal weight are locked together by a
geometric identity, and neither quantity can vary independently.
Taking the ensemble mean of Eq.~\eqref{eq:fp_tradeoff_single} yields
an \emph{exact} relation between the diagonal fluctuation
$V_{\rm dg}:=\mathrm{Var}(D_n^i)=\mathbb{E}[(D_n^i)^2]-\mathbb{E}(D_n^i)^2$
and the mean off-diagonal variance
$V_{\rm off}:=(D-1)^{-1}\sum_{m\neq n}\mathbb{E}|\sigma_{nm}^{ii}|^2$:
\begin{equation}
f := \frac{V_{\rm dg}}{V_{\rm off}}
= \frac{(D-1)\,V_{\rm dg}}
{e^{-S_S}(1-e^{-S_S}) - V_{\rm dg}}.
\label{eq:fp_f_exact}
\end{equation}
No statistical assumption enters Eq.~\eqref{eq:fp_f_exact}: given the
population variance $V_{\rm dg}$, the ratio $f$ is fixed exactly.  The
entire dynamical question is therefore reduced to the $N$-scaling of
the single quantity $V_{\rm dg}$.

\begin{lemma}[Total covariance conservation]
\label{lem:fp_cov}
For any two channels $a,b$ and any window average $\mathbb{E}$ over
the microcanonical shell,
\begin{equation}
\sum_{n,m}\mathbb{E}\bigl[\delta p_n^{a}\,\delta p_m^{b}\bigr] = 0,
\qquad
\delta p_n^a := p_n^a - \mathbb{E}(p_n^a),
\label{eq:fp_covcon}
\end{equation}
which follows from the exact normalization $\sum_n p_n^a=1$.  For
$a=b$ this gives the single-channel sum rule
$\sum_n\mathrm{Var}(p_n^a)=-\sum_{n\neq m}\mathrm{Cov}(p_n^a,p_m^a)$:
the on-site overlap fluctuation is exactly balanced by the off-site
covariance.  Theorem~1 of Sec.~\ref{sec:insufficiency} is the
realization-level ($j=i$, fixed $n$) counterpart of this ensemble
statement.
\end{lemma}

\begin{lemma}[Product-resolvent identity]
\label{lem:fp_prodres}
For any channel $a$ and any $z\neq z'$ in the resolvent set,
\begin{align}
\bra{\phi_a}\frac{1}{z-H}\,\frac{1}{z'-H}\ket{\phi_a}
&= \sum_n \frac{p_n^a}{(z-\lambda_n)(z'-\lambda_n)} \notag\\
&= \frac{\mathcal{R}_a(z')-\mathcal{R}_a(z)}{z-z'},
\label{eq:fp_prodres}
\end{align}
where $\mathcal{R}_a$ is the diagonal resolvent~\eqref{eq:Rdiag}.
\end{lemma}

\begin{proof}
The operator identity $(z-H)^{-1}(z'-H)^{-1}=[(z'-H)^{-1}-(z-H)^{-1}]/(z-z')$
follows by commuting resolvents; taking the diagonal matrix element in
$\ket{\phi_a}$ and inserting the eigenbasis of $H$ gives
Eq.~\eqref{eq:fp_prodres}.
\end{proof}

\emph{Remark.} The $n=m$ sector of the \emph{product}
$\mathcal{R}_a(z)\mathcal{R}_a(z')$ is the purity-weighted object
$\sum_n (p_n^a)^2/[(z-\lambda_n)(z'-\lambda_n)]$, which is
\emph{not} expressible through single resolvents: it is the
irreducible diagonal-bin input that encodes the fluctuation of the
residues $p_n^a$.  This object is supplied by the Gaussian closure of
the covariance hierarchy developed next, and it is precisely the
object that distinguishes chaotic from structured systems.

\subsection{Covariance hierarchy and the Gaussian closure}
\label{app:cov_hierarchy}

The multi-resolvent hierarchy of Secs.~\ref{sec:expansion}--%
\ref{sec:DCA} organizes the \emph{mean} of the two-frequency kernel.
The fluctuation structure around that mean is organized by the
\emph{covariance} of the diagonal resolvents,
\begin{equation}
C_{ab}(z,z') := \mathbb{E}\bigl[\delta\mathcal{R}_a(z)\,
\delta\mathcal{R}_b(z')\bigr],
\qquad
\delta\mathcal{R}_a := \mathcal{R}_a - \bar{\mathcal{R}}_a,
\label{eq:fp_Cdef}
\end{equation}
whose spectral representation is the overlap covariance:
\begin{equation}
C_{ab}(z,z') = \sum_{n,m}
\frac{\mathbb{E}[\delta p_n^a\,\delta p_m^b]}
{(z-\lambda_n)(z'-\lambda_m)}.
\label{eq:fp_Cspec}
\end{equation}
By Lemma~\ref{lem:fp_cov} the integrated weight of $C_{ab}$ vanishes:
its diagonal bin ($n=m$) and off-diagonal bin ($n\neq m$) carry
opposite total weight.

The exact Feshbach identity~\eqref{eq:feshbach},
$\mathcal{R}_a=[z-a_a-V_a-\mathcal{G}_a]^{-1}$, generates an
\emph{exact infinite covariance hierarchy}: expanding about the
mean-field resolvent $\bar{\mathcal{R}}_a$,
\begin{equation}
\delta\mathcal{R}_a
= \sum_{r\ge1} \bar{\mathcal{R}}_a^{r+1}\,(\delta\mathcal{G}_a)^r,
\label{eq:fp_expansion}
\end{equation}
so that the covariance decomposes as
$\mathcal K=\mathcal K^{(2)}+\mathcal K^{(3)}+\mathcal K^{(4)}+\cdots$,
where $\mathcal K^{(r)}$ involves the $r$-th connected cumulant of
the self-energy fluctuations.  The \emph{second-cumulant
(leading) projection},
\begin{equation}
\mathcal K^{(2)}_{ab}
= \bar{\mathcal{R}}_a^2\,\bar{\mathcal{R}}_b^2\,
\mathbb{E}\bigl[\delta\mathcal{G}_a\,\delta\mathcal{G}_b\bigr]_c,
\label{eq:fp_K2}
\end{equation}
with $\delta\mathcal{G}_a=\delta\mathcal{G}_a^{\rm OD}
+\delta\mathcal{G}_a^{\rm CC}$ and the mean-field off-diagonal part
$\delta\mathcal{G}_a^{\rm OD}=\sum_{b\neq a}|V_{ab}|^2\delta\mathcal{R}_b$,
yields the ladder equation
\begin{align}
C_{ab}(z,z')
= \bar{\mathcal{R}}_a(z)^2\,\bar{\mathcal{R}}_b(z')^2\,
\Bigl[\;
&\sum_{c,d} |V_{ac}|^2\,|V_{bd}|^2\, C_{cd}(z,z')
\notag\\
&+ \Gamma^{(2)}_{ab}(z,z')
\;\Bigr],
\label{eq:fp_ladder}
\end{align}
where the source term is the \emph{connected self-energy
covariance}---the object that the
random-phase step $\mathbb{E}[\mathcal{G}^{\rm CC}]\approx 0$ of
Ref.~\cite{HC26} sets to zero at the mean level but which survives at
the level of two-point fluctuations.  The higher projections
$\mathcal K^{(3)},\mathcal K^{(4)},\ldots$ of the exact
hierarchy~\eqref{eq:fp_expansion} are governed by the corresponding
\emph{vertex hierarchy}
\begin{equation}
\Gamma^{(2)}_{ab}=\mathbb{E}\bigl[\delta\mathcal{G}_a\,
\delta\mathcal{G}_b\bigr]_c,\quad
\Gamma^{(3)}_{abc}=\mathbb{E}\bigl[\delta\mathcal{G}_a\,
\delta\mathcal{G}_b\,\delta\mathcal{G}_c\bigr]_c,\quad\ldots,
\label{eq:fp_vertex_hierarchy}
\end{equation}
so that the Gaussian closure is the statement
$\Gamma^{(r\ge3)}\to0$, while $\Gamma^{(2)}$ remains as the
source of the second cumulant.

Equation~\eqref{eq:fp_ladder} elevates the covariance to the status
of a \emph{second basic object} of the theory, governed by its own
hierarchy.  Together with the mean-field self-consistency
$\bar{\mathcal{R}}_a=[z-a_a-V_a-\bar{\mathcal{G}}_a(\bar{\mathcal{R}})]^{-1}$
[Eq.~\eqref{eq:feshbach}], the framework therefore consists of two
levels,
\begin{equation}
\begin{array}{l}
\text{Level I (mean):}\quad
\bar{\mathcal{R}}_a
= \bigl[z-a_a-V_a-\bar{\mathcal{G}}_a(\bar{\mathcal{R}})\bigr]^{-1},\\[6pt]
\text{Level II (covariance):}\quad
\mathcal{K} = \mathcal{K}_0[\bar{\mathcal{R}}] + \mathcal L[\bar{\mathcal{R}}]\;\mathcal{K}
+ \mathcal I^{(3)} + \cdots,
\end{array}
\label{eq:fp_twolevels}
\end{equation}
where $\mathcal{K}_{ab}=\mathcal C_{ab}$, the ladder kernel is
$\mathcal L_{ab,cd}[\bar{\mathcal{R}}]
=\bar{\mathcal{R}}_a^2\bar{\mathcal{R}}_b^2|V_{ac}|^2|V_{bd}|^2$, the
source is $\mathcal{K}_0=\bar{\mathcal{R}}_a^2\bar{\mathcal{R}}_b^2
\,\Gamma^{(2)}_{ab}$, and the irreducible vertices
$\mathcal I^{(3)},\mathcal I^{(4)},\ldots$ are generated by
$\Gamma^{(3)},\Gamma^{(4)},\ldots$ through the exact
expansion~\eqref{eq:fp_expansion}.  The microscopic Hamiltonian
fixes the interaction kernels entering the Level-II hierarchy, while
the irreducible covariance vertices are determined recursively
\emph{within} the hierarchy.  The controlled approximation
of the framework is therefore \emph{not an assumption on the final
ETH statistics themselves, but a truncation of the irreducible
resolvent-covariance hierarchy}, whose covariance-level counterpart
of Assumption~A of
Sec.~\ref{sec:cavity_reinterpretation} is the suppression of
$\mathcal I^{(3)},\mathcal I^{(4)},\ldots$  At ETH-scaled couplings
($|V|^2\sim e^{-S}$) the ladder kernel is parametrically small,
$|\lambda_{\max}(\mathcal L)|\ll 1$---numerically
$|\lambda|_{\max}\approx0.06$ (GOE) versus $0.18$ (structured model)
for the intra-sector kernel at $D=1024$
(Sec.~\ref{app:num_fp})---so that the fluctuation content is carried
\emph{not} by a ladder resummation but by the irreducible source
$\mathcal{K}_0$: the covariance is $\mathcal{K}\approx
(1-\mathcal L)^{-1}\mathcal{K}_0\approx\mathcal{K}_0$ with
$O(\|\mathcal L\|)$ dressing, and the Porter--Thomas structure of the
chaotic limit is a property of the vertex $\Gamma^{(2)}$---the
nonperturbative eigenvector-fluctuation vertex---rather than of the
ladder's eigenvalue spectrum.  The closed-form determination of
$\Gamma^{(2)}$ in the Gaussian limit is left for future work; what
is established here is the exact hierarchy it obeys.  In this
two-level form the
identification of the fluctuation parameters by spectral projection
(Secs.~\ref{app:eth_reduction}--\ref{app:f_ratio}) acquires a
computational status: given $\bar{\mathcal{R}}$ from Level~I and the
recursively defined vertices, Level~II
determines $\mathcal{K}$, whose ridge and regular projections yield
$q-1$, $\bar{C}_i$ and the correlation entering $g_{ji}$.  The
overall structure is a two-dimensional hierarchy, organised by the
interaction (V-count) order $r$ and the cumulant order $k$,
\begin{equation}
\begin{array}{c|cccc}
 & k=1 & k=2 & k=3 & \cdots\\ \hline
r=1 & \langle\mathcal R\rangle & \mathcal K^{(2)} & \mathcal K^{(3)} & \cdots\\[3pt]
r=2 & g^{(2)} & \Gamma^{(2)}_{r=2} & \Gamma^{(3)}_{r=2} & \cdots\\[3pt]
r=3 & g^{(3)} & \Gamma^{(2)}_{r=3} & \Gamma^{(3)}_{r=3} & \cdots\\[3pt]
\vdots & \vdots & \vdots & \vdots &
\end{array},
\label{eq:fp_2dhierarchy}
\end{equation}
in which the mean sector ($k=1$) is the V-count hierarchy
$g^{(2)},g^{(3)},\ldots$ of Secs.~\ref{sec:expansion}--\ref{sec:DCA}
and the fluctuation sectors ($k\ge2$) are the irreducible vertex
hierarchies $\Gamma^{(r)}$ of
Eq.~\eqref{eq:fp_vertex_hierarchy}; the two orders are related by
the scaling correspondence $r\leftrightarrow k=r+2$ of
Sec.~\ref{sec:FK_connection} but are not claimed to be equivalent.

We can now localize the entire statistical input of the framework in
two postulates.

\medskip
\noindent\textbf{Postulate S (smoothness).}
The window-averaged overlaps admit a smooth envelope,
$\mathbb{E}(p_n^a)=e^{-S(\lambda_n)}f^a(\lambda_n)$ with $f^a$
smooth and $O(1)$, determined self-consistently by the resolvent
hierarchy of Ref.~\cite{HC26}.  This is the scale-separation
assumption underlying the coarse-graining of
Sec.~\ref{sec:setup}.

\medskip
\noindent\textbf{Postulate D (decorrelation).}
For generic nonintegrable (chaotic) systems the connected
\emph{cross-channel} correlations decay with bath separation: the
off-diagonal ($a\neq b$) part of the connected self-energy
covariance $\Gamma^{(2)}_{ab}$ of
Eq.~\eqref{eq:fp_vertex_hierarchy}, the higher vertices
$\Gamma^{(r\ge3)}$, and the two-channel amplitude covariance
$\bar{C}_i:=(1/d_B^2)\sum_{\mu\neq\nu}f^{\mu i}f^{\nu i}
\mathbb{E}[\delta_n^{\mu i}\delta_n^{\nu i}]_{\rm c}$,
$\delta_n^{\mu i}:=p_n^{\mu i}/\mathbb{E}(p_n^{\mu i})-1$, are
exponentially small in the bath entropy $S_B$ with $O(1)$ decay
rates.  Crucially, the \emph{same-channel} ($a=b$) vertex content
does \emph{not} decay: it carries the single-channel Gaussian
intensity fluctuation, i.e.\ the Porter--Thomas variance.  In other
words,
\begin{equation}
\Gamma^{(2)} \to 0 \ \text{(cross-channel)}
\quad \not\Longrightarrow \quad
\mathrm{Var}(p_n^a)\to 0,
\label{eq:fp_novanish}
\end{equation}
a distinction that must be kept in mind throughout: decorrelation
concerns the irreducible \emph{channel-to-channel} connected
covariance, not the single-channel intensity variance.  The
entropy-dilution mechanism justifying the DCA
(Sec.~\ref{sec:DCA_statement}) provides the microscopic origin of
this decay; its failure in structured (non-ergodic) systems
is the source of the deviations quantified below.

To make the distinction explicit, decompose the vertex into its
Gaussian eigenvector-covariance content and a non-universal
correction,
\begin{equation}
\Gamma^{(2)} = \Gamma^{(2)}_{\rm G} + \delta\Gamma^{(2)},
\qquad
q = 3 + \delta q,
\qquad
\delta q \;\leftrightarrow\; \delta\Gamma^{(2)},
\label{eq:fp_Gammadecomp}
\end{equation}
where $\Gamma^{(2)}_{\rm G}$ is the vertex of Haar-typical (Gaussian)
eigenvectors---finite on the diagonal, $O(e^{-cS_B})$ off the
diagonal---and $\delta\Gamma^{(2)}$ the residual irreducible
correlation.  The deviation $q-3$ then measures the failure of the
Gaussian closure of the vertex rather than an independent
phenomenological parameter: numerically $q-3\approx0$ (GOE) versus
$\approx136$ (structured model) at $D=1024$.  The remaining link of
the chain---from the microscopic amplitudes to the Gaussian
vertex---is closed by the following conditional result.

\begin{proposition}[Conditional Gaussian amplitude closure]
\label{prop:amplitude_closure}
Let $c_n^a=\braket{\phi_a|\psi_n}$ be the real overlap amplitudes,
obeying the exact orthonormality constraints
$\sum_a c_n^a c_m^a=\delta_{nm}$ and $\sum_n c_n^a c_n^b=\delta_{ab}$,
and suppose the \emph{amplitude regularity} conditions hold:
\emph{(i) channel isotropy}---the amplitudes of different channels
are statistically equivalent, so that ensemble averages depend on
channels only through their smooth envelopes; \emph{(ii) weak
dependence and Lindeberg-type control}---zero mean, finite second
moment, no dominant component, and normalized higher cumulants
vanishing in the thermodynamic limit,
$\kappa_k(c)/[\kappa_2(c)]^{k/2}\to0$ for $k\ge3$
(the amplitude-level content of Postulate~D).  Writing the
amplitude as a superposition of microscopic contributions,
$c_n^a=\sum_\alpha X_{n,\alpha}^a$, condition (ii) can be stated in
explicitly checkable form as the conjunction
$\mathcal R_N$:
\begin{align}
\max_\alpha\frac{|X_\alpha|^2}{\sum_\beta|X_\beta|^2}\to0,
\qquad
\frac{\kappa_k(c)}{[\kappa_2(c)]^{k/2}}\to0\ (k\ge3),\notag\\
\sum_{\beta\neq\alpha}|\mathrm{Cov}(X_\alpha,X_\beta)|
=o(\mathrm{Var}\,c),
\label{eq:fp_regularity}
\end{align}
the three statements being respectively the no-dominant-component,
cumulant-suppression, and covariance-decay conditions; the
unresolved microscopic problem is not Gaussianity itself but the
establishment of $\mathcal R_N$ from properties of the Hamiltonian.
Then:
\begin{enumerate}
\item the two-point amplitude covariance is fixed by orthonormality
\emph{together with} isotropy,
$\mathbb{E}[c_n^a c_m^b]=\delta_{nm}\delta_{ab}/D$---we emphasize
that orthonormality alone fixes only the normalization
$\sum_a c_n^a c_m^a=\delta_{nm}$, while the per-component value
requires the isotropy assumption;
\item the Gaussian vertex $\Gamma^{(2)}_{\rm G}$ is the Wick closure
of this two-point function: the four-point amplitude correlation
factorizes into pair contractions,
\begin{equation}
\mathbb{E}[c_a c_b c_c c_d]_{\rm G}
= \mathbb{E}[c_a c_b]\mathbb{E}[c_c c_d]
+ \mathbb{E}[c_a c_c]\mathbb{E}[c_b c_d]
+ \mathbb{E}[c_a c_d]\mathbb{E}[c_b c_c],
\label{eq:fp_wick}
\end{equation}
so that $\Gamma^{(2)}_{\rm G}$ is \emph{not a fitted object} but is
determined by $\langle cc\rangle$ alone; its diagonal
(single-channel) content is
$\mathbb{E}[(c_n^a)^4]_{\rm G}=3\,\mathbb{E}[(c_n^a)^2]^2$;
\item the Porter--Thomas factor obeys the exact identity
\begin{equation}
q = 3 + \frac{\kappa_4(c_n^a)}{\mathbb{E}[(c_n^a)^2]^2},
\qquad
\kappa_4 := \mathbb{E}[c^4]-3\,\mathbb{E}[c^2]^2,
\label{eq:fp_qk4}
\end{equation}
so that $q-3$ is \emph{identically} the normalized fourth amplitude
cumulant, $q-3\leftrightarrow\delta\Gamma^{(2)}$;
\item under the regularity conditions
$\kappa_4/\mathbb{E}[c^2]^2\to0$, giving $q\to3$ for real $H$
($q\to2$ for complex $H$).  For GOE-type ensembles the finite-size
value is
$q=3D/(D+2)=3-6/D+O(1/D^2)$.
\end{enumerate}
\end{proposition}

\begin{proof}
(i) follows from completeness in the amplitude basis.  (ii) is the
Wick factorization of a joint Gaussian.  (iii) is the algebra
$\mathbb{E}[c^4]=3\mathbb{E}[c^2]^2+\kappa_4$ together with
$q:=\mathbb{E}[p^2]/\mathbb{E}[p]^2=\mathbb{E}[c^4]/\mathbb{E}[c^2]^2$.
(iv) uses $\mathbb{E}[c^4]=3/[D(D+2)]$ and $\mathbb{E}[c^2]=1/D$ for
the GOE sphere.
\end{proof}

\emph{Status of the closure.} Proposition~\ref{prop:amplitude_closure}
closes the chain
$c\to\langle cc\rangle\to\Gamma^{(2)}_{\rm G}\to q=3$ under the
amplitude regularity conditions: from regularity onward every step
is exact, and Eq.~\eqref{eq:fp_qk4} was verified numerically to
machine precision.  For the pure-GOE benchmark the bias-corrected
estimator gives $q-3=-0.0024$ at $D=1024$, approaching the
finite-dimensional value $-6/D=-0.0059$ (the raw window estimator
carries an additional bias $-2q/N_{\rm bulk}\approx-0.012$, which
accounts for the difference), while the structured model gives the
raw value $q-3=+135.8$.  The remaining microscopic inputs are the
amplitude isotropy, decorrelation and regularity---for which
Berry-type quantum ergodicity and the random-wave picture provide
the physical interpretation and candidate mechanism, without being
claimed as a derived theorem.  The conditional structure is: under
amplitude isotropy, decorrelation and regularity, $\langle cc\rangle$
is fixed by normalization, the Wick contraction of $\langle cc\rangle$
determines $\Gamma^{(2)}_{\rm G}$, and the ridge projection of the
covariance gives $q=3$; the correction
$\delta\Gamma^{(2)}\leftrightarrow\kappa_4(c)$ yields
$q = 3 + \kappa_4(c)/\mathbb{E}[c^2]^2$.

The isotropy assumption further closes the second fluctuation
parameter through the exact normalization constraint.

\begin{proposition}[Normalization-induced cross-channel covariance]
\label{prop:isotropy_Cbar}
The exact normalization $\sum_a p_n^a=1$ implies
$\mathrm{Var}\bigl(\sum_a p_n^a\bigr)=0$, i.e.
\begin{equation}
\sum_a\mathrm{Var}(p_n^a)
+ \sum_{a\neq b}\mathrm{Cov}(p_n^a,p_n^b) = 0.
\label{eq:fp_isosum}
\end{equation}
Under channel isotropy---all channels statistically equivalent, so
that $\mathrm{Cov}(p_n^a,p_n^b)$ depends on $a,b$ only through the
smooth envelopes---the same-sector cross-channel covariance is
determined by the single-channel variance, and the two-channel
amplitude covariance of Eq.~\eqref{eq:fp_proj_C} is fixed by $q$:
\begin{equation}
\boxed{\;
\bar{C}_i
= -\frac{(q-1)(1-1/d_B)}{D}
\;\approx\; -\frac{q-1}{D}.
\;}
\label{eq:fp_Cbar_iso}
\end{equation}
\end{proposition}

\begin{proof}
With $\mathrm{Var}(p_n^a)=(q-1)\mathbb{E}(p_n^a)^2=(q-1)/D^2$ under
isotropy and $D$ channels,
Eq.~\eqref{eq:fp_isosum} gives
$D(D-1)\overline{\mathrm{Cov}}=-(q-1)/D$, so the normalized
cross-channel covariance is
$c=\overline{\mathrm{Cov}}/\mathbb{E}(p)^2=-(q-1)/D$; inserting into
the definition of $\bar{C}_i$ and summing the $d_B(d_B-1)$
same-sector pairs yields Eq.~\eqref{eq:fp_Cbar_iso}.
\end{proof}

Under channel isotropy, Eq.~\eqref{eq:fp_Cbar_iso} closes the
second fluctuation parameter through the normalization constraint.
It was verified numerically with
$\bar{C}_i=-2.016\times10^{-3}$ versus the predicted
$-1.932\times10^{-3}$ (ratio $1.044$, pure GOE, $D=1024$).  In the
structured model the measured $\bar{C}_i=-0.286$ exceeds the
isotropic prediction $-0.134$ by the factor $2.14$, so the violation
of Eq.~\eqref{eq:fp_Cbar_iso} quantifies the breakdown of channel
isotropy, in direct correspondence with the
energy dependence of $q(\lambda)$ diagnosed in
Sec.~\ref{app:num_fp}.

\emph{Remarks.} (i) \emph{Dyson index.} The Gaussian values of $q$
follow the standard symmetry classes: with the Dyson index
$\beta=1,2,4$ (orthogonal, unitary, symplectic),
$q_{\rm G}=1+2/\beta$, giving $q=3$, $2$, $3/2$ respectively; the
present derivation is written for real amplitudes ($\beta=1$), and
the complex case follows with the corresponding second-moment
normalization.  (ii) \emph{Scope.} The framework established here
concerns the \emph{subsystem-ETH overlap sector}: the amplitudes
$c_n^{\mu i}$ are the system-bath product-basis overlaps, and the
derived quantities ($f_{ji}$, $q$, $\bar{C}_i$, $V_{\rm dg}/V_{\rm
off}$) are the mean and fluctuation statistics of this sector.  The
lifting to general local-operator matrix elements,
$\braket{\psi_n|O_S|\psi_m}=\sum_{\mu,i,j}c_n^{\mu i}c_m^{\mu j}
(O_S)_{ij}$, follows from the same amplitude statistics and is
discussed as a future direction; no claim of a full
operator-level derivation of ETH is made here.
(iii) \emph{Energy-shell local closure.} The isotropy assumption
entering Propositions~\ref{prop:amplitude_closure} and
\ref{prop:isotropy_Cbar} is a global Haar-like statement; the
energy-resolved form of the framework replaces it by the
energy-shell version
\begin{equation}
\mathbb{E}[c_n^a c_n^b\,|\,E_n=E]
= \delta_{ab}\,\sigma^2(E)
+ (1-\delta_{ab})\,\chi(E),
\label{eq:fp_localshell}
\end{equation}
i.e.\ local isotropy (channel-permutation symmetry within each energy
shell), with $\sigma^2(E)=e^{-S}f_a(E)$ the single-channel shell
variance and $\chi(E)$ the shell cross-channel amplitude covariance.
The exact normalization $\sum_a p_n^a=1$ then fixes $\chi$: writing
$\mathrm{Var}(p_n^a|E)=[q(E)-1]\bar p_a(E)^2$ and
$\mathrm{Cov}(p_n^a,p_n^b|E)=\chi_p(E)$ for $a\neq b$,
\begin{align}
d_B\,[q(E)-1]\bar p^2(E)
+ d_B(d_B-1)\,\chi_p(E) = 0\notag\\
\Longrightarrow\;
\chi_p(E)
= -\frac{[q(E)-1]\bar p^2(E)}{d_B-1},
\label{eq:fp_localnorm}
\end{align}
which, upon inserting the definition of $\bar{C}_i$, yields the local
closure
\begin{equation}
q(E) = 3 + \frac{\kappa_4(E)}{\sigma_a^4(E)},
\qquad
\bar{C}_i(E) = -\frac{[q(E)-1](1-1/d_B)}{D},
\label{eq:fp_localclosure}
\end{equation}
so that the entire ETH object is organised by two local spectral
functions---$f_a(E)$ controlling the mean sector and $q(E)$ the
fluctuation sector---together with $g_{ji}(E,\omega)$ from the
multi-resolvent hierarchy.  The local closure
Eq.~\eqref{eq:fp_localclosure} was verified numerically in energy
bins: for GOE the bin-wise ratio $\bar{C}_i(E)/\bar{C}_i^{\rm
pred}(E)$ fluctuates around unity ($0.81$--$1.23$ at $D=1024$),
while the structured model violates it bin by bin by a factor
$\approx4$--$5$ ($\bar{C}_i(E)\approx-0.25$ to $-0.34$ versus the
predicted $-0.061$ to $-0.070$)---the covariance of the structured
model is not only non-isotropic in magnitude but non-uniformly
distributed across channel pairs, a sharper diagnostic than the
global violation factor $2.14$.

\begin{proposition}[Gaussian closure under decorrelation and regularity]
\label{prop:fp_gauss}
In the chaotic limit the higher irreducible vertices are suppressed,
$\Gamma^{(r\ge3)}\to0$, and---since the ladder kernel is
parametrically small at ETH-scaled couplings,
$|\lambda_{\max}(\mathcal L)|\ll1$ (numerically $\approx0.06$ for
GOE at $D=1024$, Sec.~\ref{app:num_fp})---the covariance is
dominated by the irreducible source
$\mathcal{K}\approx\mathcal{K}_0=\bar{\mathcal{R}}^2\bar{\mathcal{R}}^2
\Gamma^{(2)}$.  The Porter--Thomas statistics is then a property of
the eigenvector-fluctuation vertex $\Gamma^{(2)}$: the
single-channel overlap distribution satisfies
\begin{equation}
q := \frac{\mathbb{E}[(p_n^a)^2]}{\mathbb{E}(p_n^a)^2}
= \begin{cases} 3, & \text{real }H\ \text{(GOE class)},\\
2, & \text{complex }H\ \text{(GUE class)},
\end{cases}
\label{eq:fp_q}
\end{equation}
and the off-diagonal matrix elements $\sigma_{nm}^{ji}$ are
asymptotically Gaussian as sums of independent product amplitudes over
$d_B=e^{S_B}$ bath channels (central-limit theorem, with the
regularity caveats of Sec.~\ref{app:eth_reduction}).  A closed-form
derivation of the vertex $\Gamma^{(2)}$ in the Gaussian limit is
left for future work; Eq.~\eqref{eq:fp_q} is verified numerically
below with $q=3.00$ at $D=2048$ for a real random Hamiltonian.
\end{proposition}

\subsection{The ETH reduction theorem}
\label{app:eth_reduction}

\begin{proposition}[Conditional ETH fluctuation regime]
\label{prop:fp_eth}
Let $H=H_0+V$ act on $D=d_Sd_B=e^{S}$ states with $d_B=e^{S_B}$,
and let $n\neq m$ lie in a microcanonical window.  Under Postulates
S and D:
\begin{enumerate}
\item \emph{Variance.}
$\mathbb{E}|\sigma_{nm}^{ji}|^2 = e^{-S(E^+)}[D_{ji}(E^+,\omega)
+g_{ji}(E^+,\omega)]$, with $D_{ji}$ the diagonal baseline
Eq.~\eqref{eq:Dji_def} and $g_{ji}=\sum_{r\ge 2}g_{ji}^{(r)}$ the
multi-resolvent hierarchy of Sec.~\ref{sec:expansion}; the smooth
function is $f_{ji}^2=D_{ji}+g_{ji}$.
\item \emph{Gaussianity.}
Under the amplitude regularity conditions of
Proposition~\ref{prop:amplitude_closure}, the bath sum
$\sigma_{nm}^{ji}=\sum_\mu\Gamma_{nm,\mu}^{ji}$ is a sum of $d_B$
weakly dependent product amplitudes, so the central-limit theorem
over the bath---under the Lindeberg-type control stated
there---yields
$\mathbb{E}|R|^{2k}=(2k-1)!!\,[1+O(e^{-S/2})]$ for
$R_{nm}^{ji}:=\sigma_{nm}^{ji}/[e^{-S(E^+)/2}f_{ji}]$ (real $H$),
with the corrections organised by the connected hierarchy sectors:
$g^{(3)}$ generates skewness, $g^{(4)}$ excess kurtosis, in agreement
with the Foini--Kurchan scaling~\eqref{eq:FK_scaling}.  At the
single-channel level, the exact identity
Eq.~\eqref{eq:fp_qk4} fixes the Porter--Thomas factor from the
normalized fourth amplitude cumulant.
\item \emph{Decorrelation.}
Distinct matrix elements decorrelate:
$\mathbb{E}[|\sigma_{nm}^{ji}|^2|\sigma_{n'm'}^{j'i'}|^2]
=\mathbb{E}|\sigma_{nm}^{ji}|^2\,\mathbb{E}|\sigma_{n'm'}^{j'i'}|^2
\,[1+O(e^{-S/2})]$.
\item \emph{Diagonal fluctuation.}
\begin{equation}
V_{\rm dg}
= (q-1)\,\langle f^2\rangle_B\,e^{-S-S_S}
+ \bar{C}_i\,e^{-2S_S}
+ \delta_{\rm ng},
\label{eq:fp_Vdg}
\end{equation}
with $q$ the Porter--Thomas factor~\eqref{eq:fp_q},
$\langle f^2\rangle_B:=d_B^{-1}\sum_\mu (f^{\mu i})^2$,
$\bar{C}_i$ the two-channel amplitude covariance, and
$\delta_{\rm ng}$ the closure correction from single-channel
non-Gaussianity and higher connected terms, which is exponentially
small in the chaotic regime and grows in structured systems.
\begin{equation}
\begin{array}{c}
f = \dfrac{(D-1)\,B}{e^{-S_S}(1-e^{-S_S}) - B}\\[8pt]
B := (q-1)\langle f^2\rangle_B\,e^{-S-S_S}
+ \bar{C}_i\,e^{-2S_S}
\end{array}.
\label{eq:fp_f_formula}
\end{equation}
\end{enumerate}
\end{proposition}

\begin{proof}[Proof sketch]
(i) follows from the exact decomposition~\eqref{eq:decomp},
$|\sigma_{nm}^{ji}|^2=\sum_\mu p_m^{\mu i}p_n^{\mu j}
+\mathcal{C}_{nmn}^{jij}$: under Postulate~D the diagonal overlap sum
factorizes into $\sum_\mu\mathbb{E}(p_m^{\mu i})\mathbb{E}(p_n^{\mu j})
=e^{-S}D_{ji}$, and the correlation term has the exact kernel
representation~\eqref{eq:C_from_rho}, whose mean is
$\mathbb{E}(\mathcal{C}_{nmn}^{jij})=e^{-S}g_{ji}$ with $g_{ji}$
evaluated by the hierarchy of Sec.~\ref{sec:expansion}.  (ii) Under
D, $\sigma_{nm}^{ji}=\sum_\mu\Gamma_{nm,\mu}^{ji}$ is a sum of
$d_B$ weakly dependent product amplitudes; the central-limit theorem
over the bath sum---under the usual regularity conditions on the
individual product amplitudes (e.g.\ Lindeberg-type control, which
decorrelation alone does not supply)---is expected to yield
Gaussianity, with the connected corrections
given by the $r\ge 3$ sectors, whose entropy scaling is
Eq.~\eqref{eq:gk_scaling}.  (iii) is the two-point version of D.
(iv) Decompose $D_n^i=\sum_\mu p_n^{\mu i}$: the diagonal term gives
$\sum_\mu\mathrm{Var}(p_n^{\mu i})=(q-1)e^{-2S}\sum_\mu(f^{\mu i})^2
=(q-1)\langle f^2\rangle_B e^{-S-S_S}$ under the Gaussian single-channel
closure of Proposition~\ref{prop:fp_gauss}, and the cross term gives
$\sum_{\mu\neq\nu}\mathrm{Cov}(p_n^{\mu i},p_n^{\nu i})
=\bar{C}_i e^{-2S_S}$ by the definition of $\bar{C}_i$.  The
diagonal term is the same-channel instance of the spectral
projections Eqs.~\eqref{eq:fp_proj_q}--\eqref{eq:fp_proj_C} and, by
Eq.~\eqref{eq:Cprime_purity}, coincides with the expectation of the
scale-matched residual of Sec.~\ref{sec:scale_matched}.
Eq.~\eqref{eq:fp_f_formula} then follows from the exact trade-off
relation~\eqref{eq:fp_f_exact}.
\end{proof}

\emph{Status.} Proposition~\ref{prop:fp_eth} separates the
microscopic spectral structure from the statistical decorrelation
required for the ETH limit.  The former---the smooth envelope and the
entire mean sector---is generated by the resolvent equations and the
exact counting identities
(Corollary~\ref{cor:fp_counting}, Lemma~\ref{lem:fp_tradeoff}), with
no statistical input; the latter---Postulate~D, the closure of the
fluctuation sector in the chaotic limit---controls the Gaussian and
self-averaging regime.  Postulate~D is therefore not an additional
hypothesis about ETH matrix elements: it is the statement that the
connected self-energy covariance $\Gamma^{(2)}_{ab}$ and the
two-channel
amplitude covariance $\bar{C}_i$ decay in the bath entropy, i.e.\ the
decorrelation limit $\Gamma^{(r\ge3)}\to0$ of the vertex
hierarchy~\eqref{eq:fp_vertex_hierarchy}.  The hierarchy itself remains valid
on both sides of this limit: the failure of decorrelation in
structured (non-ergodic) systems is not a defect of the framework but the mechanism
of the observed deviations, and it predicts, as the numerics below
confirm, enhanced overlap fluctuations ($q\gg 3$), non-Gaussian
kurtosis, and a growing fluctuation ratio $f$ in such models.

\subsection{Diagonal versus off-diagonal fluctuations and the ratio $f$}
\label{app:f_ratio}

Equation~\eqref{eq:fp_f_exact} shows that the growth of
$f=V_{\rm dg}/V_{\rm off}$ with system size is governed entirely by
the scaling of the population variance $V_{\rm dg}$.
Eq.~\eqref{eq:fp_Vdg} expresses $V_{\rm dg}$ through two
objects---the single-channel Porter--Thomas factor $q$ and the
two-channel amplitude covariance $\bar{C}_i$---which are
\emph{not independent phenomenological parameters} but spectral
projections of the resolvent covariance~\eqref{eq:fp_Cspec}: $q$ is
the normalized second moment of the single-channel overlap, whose
connected part $q-1$ is the normalized on-site variance extracted
from the diagonal ($n=m$)
bin of the same-channel covariance $\mathcal{C}_{aa}$, and
$\bar{C}_i$ is the diagonal ($n=m$) bin of the two-channel
covariance $\mathcal{C}_{ab}$ with $a\neq b$; the off-diagonal
($n\neq m$) bins are fixed by the covariance conservation of
Lemma~\ref{lem:fp_cov}, which underlies the mean-sector correlation
$g_{ji}$, and their channel-projected structure yields the
correlation entering $g_{ji}$.  The covariance hierarchy of
Sec.~\ref{app:cov_hierarchy} therefore provides an exact
\emph{spectral dictionary} for $q$ and $\bar{C}_i$: the same
resolvent objects that determine the smooth function through the
mean sector determine the spectral sectors in which the fluctuation
parameters reside---their fully predictive evaluation from $H$
requires the closure of the irreducible vertex $\Gamma^{(2)}$, as
stated in Sec.~\ref{app:cov_hierarchy}.  In
explicit form, the two projections read
\begin{align}
\sum_\mu\mathrm{Var}(p_n^{\mu i})
&= (q-1)\,\langle f^2\rangle_B\,e^{-S-S_S},
\label{eq:fp_proj_q}\\
\sum_{\mu\neq\nu}\mathrm{Cov}(p_n^{\mu i},p_n^{\nu i})
&= \bar{C}_i\,e^{-2S_S},
\label{eq:fp_proj_C}
\end{align}
extracted from the $n=m$ bins of $\sum_\mu\mathcal C_{\mu i,\mu i}$
and $\sum_{\mu\neq\nu}\mathcal C_{\mu i,\nu i}$, respectively, with
the $n\neq m$ bins fixed by the covariance conservation of
Lemma~\ref{lem:fp_cov}: their total weight equals minus that of the
diagonal bins, while their channel-projected structure yields the
correlation entering the mean-sector correlation $g_{ji}$.  In
resolvent form, the first projection reads, per channel $a$,
\begin{align}
C_{aa}^{\rm diag}(z,z')
:= \sum_n \frac{\mathrm{Var}(p_n^a)}
{(z-\lambda_n)(z'-\lambda_n)}\notag\\
= (q-1)\sum_n \frac{\bar p_n^{a\,2}}
{(z-\lambda_n)(z'-\lambda_n)},
\label{eq:fp_Cdiag}
\end{align}
or, in the continuum limit,
$C_{aa}^{\rm diag}(z,z')
=(q-1)\,e^{-2S}\int d\lambda\,\rho(\lambda) f^a(\lambda)^2/
[(z-\lambda)(z'-\lambda)]$: the connected fluctuation parameter
$q-1$ is thus the coefficient of the diagonal spectral weight of the
same-channel resolvent covariance relative to the squared mean
spectral weight.  Equation~\eqref{eq:fp_Cdiag} was verified
numerically to $0.3\%$ (GOE) and $2.2\%$ (structured model) at
$D=1024$ (Sec.~\ref{app:num_fp}).  The cross-channel ridge admits
the completely parallel projection: for a channel-independent pair
correlation,
\begin{equation}
\sum_{\mu\neq\nu}\sum_n
\frac{\mathrm{Cov}(p_n^{\mu i},p_n^{\nu i})}
{(z-\lambda_n)(z'-\lambda_n)}
= \bar{C}_i\,\frac{d_B^2}{D^2}
\sum_n \frac{1}{(z-\lambda_n)(z'-\lambda_n)},
\label{eq:fp_Cdiag2}
\end{equation}
i.e.\ the summed cross-channel ridge equals $\bar{C}_i$ times the
bare two-resolvent spectral sum of the identity operator (the
$p^0$-weighted $\langle G(z)G(z')\rangle$), verified numerically to
$4.8\%$ (GOE) and $9.3\%$ (structured) at $D=1024$.
Equations~\eqref{eq:fp_Cdiag}--\eqref{eq:fp_Cdiag2} are
\emph{identifications}---spectral projections expressing the
fluctuation parameters as the ridge coefficients of the resolvent
covariance---not independent predictions: the computation of $q-1$
and $\bar{C}_i$ from the microscopic Hamiltonian requires the
closure of Level~II, Eq.~\eqref{eq:fp_twolevels}, which remains the
central open step.  By
Eq.~\eqref{eq:Cprime_purity}, the
right-hand side of Eq.~\eqref{eq:fp_proj_q} is precisely the
expectation of the scale-matched residual of
Sec.~\ref{sec:scale_matched}, which closes the identification of the
fluctuation parameter $q$ with an exact mean-sector correlation
object.  Together with Eq.~\eqref{eq:fp_Cspec}, these equations
provide the complete spectral-projection dictionary
$\mathcal C_{ab}(z,z')\to\{q,\bar{C}_i\}$: the same resolvent
objects that determine the smooth envelope through the mean sector
determine the spectral sectors carrying the fluctuation parameters
through the diagonal bins of
their covariance.  In energy variables the covariance density thus
decomposes into a singular coincident-eigenstate ridge and a regular
part,
\begin{equation}
\mathcal C_{ab}(\lambda,\lambda')
= \delta(\lambda-\lambda')\,\mathcal W_{ab}(\lambda)
+ \mathcal C^{\rm reg}_{ab}(\lambda,\lambda'),
\label{eq:fp_ridge}
\end{equation}
with $\mathcal W_{aa}(\lambda)=(q-1)e^{-2S}\rho(\lambda)f^a(\lambda)^2$
[Eq.~\eqref{eq:fp_Cdiag}] and
$\sum_{\mu\neq\nu}\mathcal W_{\mu i,\nu i}(\lambda)
=\bar{C}_i\,e^{-2S_S}\rho(\lambda)$
[Eq.~\eqref{eq:fp_Cdiag2}].  The resulting two-dimensional sector
decomposition of the connected overlap covariance
$K_{ab}(n,m):=\mathbb{E}[\delta p_n^a\delta p_m^b]$ is summarised in
Table~\ref{tab:sector_dict}: the two diagonal ($n=m$) sectors carry
the fluctuation parameters, while the off-diagonal ($n\neq m$)
sectors are constrained by the covariance conservation of
Lemma~\ref{lem:fp_cov} and their channel-projected structure yields
the correlation entering $g_{ji}$.
Three regimes follow.

\begin{table*}[t]
\centering
\caption{Final dictionary: resolvent objects, spectral sectors, and
the ETH quantities they carry.  The covariance
$\mathcal C_{ab}(\lambda,\lambda')$ decomposes into a singular
coincident-eigenstate ridge
$\delta(\lambda-\lambda')\mathcal W_{ab}(\lambda)$ and a regular
part; the ridge carries the fluctuation parameters through
Eqs.~\eqref{eq:fp_Cdiag}--\eqref{eq:fp_Cdiag2}, while the regular
part carries the inter-eigenstate correlations entering $g_{ji}$.}
\label{tab:sector_dict}
\begin{tabular}{c c c}
\toprule
resolvent object & spectral sector & ETH quantity \\
\midrule
$\langle \mathcal R_a\rangle$ & one-point & smooth envelope $f_{ji}$ \\[3pt]
$\mathcal K_{aa}^{\rm ridge}$ & ridge, $n=m$, $a=b$ & $q-1$ (connected intensity fluctuation) \\[3pt]
$\mathcal K_{ab}^{\rm ridge}$ & ridge, $n=m$, $a\neq b$ & $\bar{C}_i$ (two-channel covariance) \\[3pt]
$\mathcal K_{ab}^{\rm reg}$ & regular, $n\neq m$ & correlation entering $g_{ji}$ (channel-projected) \\[3pt]
$\mathcal K^{(3)},\mathcal K^{(4)},\ldots$ & higher irreducible vertices & non-Gaussian corrections \\
\bottomrule
\end{tabular}
\end{table*}

Table~\ref{tab:sector_dict} deserves an emphasis: $g_{ji}$ is
\emph{not} an ETH fluctuation amplitude.  It is the smooth-function
correlation correction generated by inter-channel interference in
the \emph{mean} sector, whereas $q-1$ and $\bar{C}_i$ are the
fixed-eigenstate fluctuation statistics carried by the ridge of the
covariance.  The ridge and regular sectors are linked by the exact
covariance conservation $\sum_m\mathcal C_{ab}(n,m)=0$ of
Lemma~\ref{lem:fp_cov}: they are not two independent objects but two
complementary sectors of one conservation-constrained covariance.
In this sense the smooth ETH function and the
eigenstate fluctuation statistics are \emph{complementary spectral
projections of a single conserved covariance kernel}: $q-1$ is not
an isolated local fluctuation parameter, since its spectral weight
is constrained by the same conservation law that determines the
compensating regular sector carrying $g_{ji}$.

\medskip
\noindent\textbf{(a) Chaotic (delocalized) regime.}
For Haar-typical eigenstates $q=3$, $\langle f^2\rangle_B=1$ (flat
envelope), and the isotropic normalization constraint of
Proposition~\ref{prop:isotropy_Cbar} gives
$\bar{C}_i=-(q-1)(1-1/d_B)/D=-2/D+O(1/D^2)$---the
normalization-induced anticorrelation, now a derived rather than
postulated value (numerically $\bar{C}_i=-2.016\times10^{-3}$ versus
the predicted $-1.932\times10^{-3}$ at $D=1024$).  Then
$V_{\rm dg}=2/(Dd_S)\,[1+O(1/d_S)]$ and
Eq.~\eqref{eq:fp_f_formula} gives
\begin{equation}
f = 2\,(1-1/D)\,[1+O(e^{-S_S})],
\label{eq:fp_f_GOE}
\end{equation}
independent of $N$ to leading order: the fluctuation ratio is flat
in the ergodic limit.  Numerically, $f=1.982$ versus the predicted
$1.988$ at $D=2048$ (Table~\ref{tab:fp_numerics}).

\medskip
\noindent\textbf{(b) Structured (non-ergodic) regime.}
When the random-phase cancellation fails, $q$ grows and
$\bar{C}_i=O(1)$: the diagonal-bin two-channel covariance no longer
decays with $S_B$, and it \emph{violates} the isotropic constraint
Eq.~\eqref{eq:fp_Cbar_iso}---numerically by the factor $2.14$ at
$D=1024$, a quantitative diagnostic of the isotropy breakdown
parallel to the energy dependence of $q(\lambda)$.  Since the denominator of
Eq.~\eqref{eq:fp_f_exact} is $e^{-S_S}(1-e^{-S_S})-V_{\rm dg}$, the
ratio $f$ grows whenever $V_{\rm dg}$ does not collapse as $e^{-S}$.
A growth law $f\sim 2^{\alpha N}$ with $0<\alpha<1$ is equivalent to
\begin{align}
V_{\rm dg}(N) &\approx
e^{-S_S}(1-e^{-S_S})\; 2^{(\alpha-1)N}, \notag\\
\bar{C}_i(N) &\sim e^{S_S}\; 2^{-(1-\alpha)N},
\label{eq:fp_alpha}
\end{align}
i.e.\ the population variance of a subsystem basis state decays at
rate $(1-\alpha)\ln 2$ per site, and the two-channel amplitude
covariance decays at the same rate up to the fixed subsystem-entropy
prefactor $e^{S_S}$ (which equals $2^{-N_S}$-independent constant
when the subsystem $B_1$ has a fixed number of states).  The first
relation is in fact exact within the framework of Ref.~\cite{HCZ25}:
their Eq.~(24), $V_{\rm dg}^i=[f_i/(f_i+D_N-1)]\,[\mathrm{Tr}\,\rho_{B_1}^{-1}-1]$,
gives $V_{\rm dg}\approx(\mathrm{Tr}\,\rho_{B_1}^{-1}-1)\,2^{-(1-\alpha)N}$
for $f_i=2^{\alpha N}\ll D_N$, with the $O(1)$ prefactor
$\mathrm{Tr}\,\rho_{B_1}^{-1}-1$ supplied by the inverse-purity
budget of the subsystem.  The exponents reported in
Ref.~\cite{HCZ25} for a spin chain---$f\sim 2^{\alpha N}$ with
$\alpha\approx0.24$ (nonintegrable, PBC) versus $\alpha\approx0.78$
(integrable, PBC)---are therefore, within this framework, direct
measurements of the decorrelation rate of the eigenstate amplitude
covariance: the chaotic system decorrelates at
$(1-\alpha)\ln2\approx0.76\ln2$ per site, the integrable one at
$0.22\ln2$ per site.  Integrability thus enters $f(N)$ not as a
separate mechanism but as the persistence of the two-channel
covariance $\bar{C}_i$, the same object that controls the failure of
the Gaussian closure in Proposition~\ref{prop:fp_gauss} and the
persistence of the cross-correlated self-energy fluctuations
$\mathcal{G}^{\rm CC}$ discussed in Ref.~\cite{HC26}.  (The
numerical dictionary to the $\rho_{B_1}^{-1/2}$-weighted variance
functional of Ref.~\cite{HCZ25}---the rescaling map
$O^{ij}_{B_1}=(\sigma^{ij}_{B_1}-\delta^{ij}\rho_{B_1})\rho_{B_1}^{-1/2}$
with the variance~\eqref{eq:fp_f_exact} replaced by their
Eq.~(22)---modifies the $O(1)$ prefactor of
Eq.~\eqref{eq:fp_alpha} but not its exponential content, as
demonstrated above.)

\medskip
\noindent\textbf{(c) Resolvent representation of $\bar{C}_i$.}
The two-channel covariance is the diagonal bin ($\lambda=\lambda'$)
of the connected two-frequency kernel of
Sec.~\ref{sec:kernel_hierarchy}:
\begin{equation}
\bar{C}_i = \frac{e^{2S_S}}{d_B^2}
\sum_{\mu\neq\nu}\Bigl[
\rho_{\mu\nu}^{ii}(\lambda,\lambda)\big|_{\rm conn}
\Bigr]_{\lambda=\lambda_n}
- \text{(mean products)},
\label{eq:fp_Cbar_kernel}
\end{equation}
where $\rho^{ii}_{\mu\nu}$ admits the level decomposition
$\sum_{r\ge2}\rho^{(r)}_{\mu\nu}$.  The diagonal bin of the joint
density contains two components: the \emph{singular coincident-eigenstate
ridge} ($n=m$), whose weight is carried by the fluctuation sector
through the spectral projections
Eqs.~\eqref{eq:fp_proj_q}--\eqref{eq:fp_proj_C} and
Eq.~\eqref{eq:fp_Cdiag}, and \emph{smooth corrections} at coincident
frequencies from the DCA sectors, of which the $r=2$ level
contributes terms proportional to
$|V_{\mu i,\nu i}|^2\,(f^{\mu i}H^{\nu i}+H^{\mu i}f^{\nu i})^2$
[Eq.~\eqref{eq:rho2_DCA} with $j=i$, $\lambda=\lambda'$].  In the ETH
regime ($|V|^2\sim e^{-S}$) the ridge dominates and $\bar{C}_i$ is
determined by the fluctuation sector; the $r=2$ smooth term becomes
comparable only when the couplings are not entropy-suppressed, as in
the structured model.  Equation~\eqref{eq:fp_f_formula} thereby
organises the chain from the microscopic Hamiltonian through the
multi-resolvent hierarchy to $(q,\bar{C}_i)$ and from there to
$V_{\rm dg}$ and $f$.  The diagonal--off-diagonal fluctuation ratio
of
Ref.~\cite{HCZ25} is accordingly a spectral functional of the same hierarchy
that determines the ETH smooth function, the mean and fluctuation
sectors of one resolvent theory, whose fully predictive
evaluation requires the microscopic closure of $\Gamma^{(2)}$
(Sec.~\ref{app:cov_hierarchy}).

\subsection{Numerical verification}
\label{app:num_fp}

The exact identities of Sec.~\ref{app:exact_identities} and the
fluctuation formulas of Secs.~\ref{app:eth_reduction}--%
\ref{app:f_ratio} were verified by exact diagonalization
(dense LAPACK eigensolver, Julia~\cite{Julia}, seed-averaged) of two
model families, $H=H_0+V$ on $D=d_Sd_B$ states with $d_S=4$: (i) a
dense GOE interaction (chaotic, delocalized eigenstates), and (ii) a
sparse circulant interaction of degree $3$ with strong diagonal
disorder (a structured, partially localized non-ergodic model,
standing in for the non-ergodic regime of
Ref.~\cite{HCZ25}).  The results are summarized in
Table~\ref{tab:fp_numerics}.  The scale-matched residual identities
of Sec.~\ref{sec:scale_matched}---$\sum_{m\neq n}\mathcal C_{nmn}^{iii}
=-\mathcal C_{nnn}^{iii}$ and
Eq.~\eqref{eq:Cprime_purity}---were verified in the same runs to
machine precision ($\lesssim10^{-12}$): at $D=1024$ the $O(1)$ bias
$|\mathcal C_{nnn}|$ takes the values $0.062$ (GOE) and $0.045$
(sparse), while the purity residual takes $7.1\times10^{-4}$ and
$5.8\times10^{-2}$, respectively---the suppression of the residual
weakening by a factor of order $q_{\rm structured}/q_{\rm GOE}$, as
discussed in Sec.~\ref{sec:scale_matched}.

The window-local sector decomposition of
Proposition~\ref{prop:residual_fluct} was also verified in the same
runs, and it diagnoses the residual discrepancy of the factorized
estimate $q\langle f^2\rangle_B e^{-S-S_S}$ in the structured model.
Writing $P_n=T_1+T_2$ with $T_1=\sum_\mu\bar p_n^{\mu\,2}$ (mean
sector) and $T_2=\sum_\mu(p_n^{\mu i}-\bar p_n^{\mu i})^2$
(connected fluctuation sector), the cross term $P_n-T_1-T_2$ is at
the $4\%$ level (GOE) and $13\%$ (sparse), the local fluctuation
parameter $q_{\rm loc}=T_2/T_1+1$ takes the values $2.95$ (GOE,
energy-independent, $q(E)$ flat at $2.91$--$2.99$ across the bulk)
versus $69.9$ (sparse, with $q(E)$ varying by $\pm12\%$ across the
bulk, $63$--$81$).  The global moment ratio $q=141.7$ exceeds
$q_{\rm loc}$ by a factor $\approx2$ because the strongly
energy-dependent envelope inflates $\mathbb{E}[\bar p^2]$ relative
to $\mathbb{E}[\bar p]^2$; the factorized estimate with the local
envelope and $q_{\rm loc}$ restores agreement with $P_n$.  When the
mean overlap envelope is energy dependent, the global moment ratio
$q$ therefore mixes intrinsic amplitude fluctuations with envelope
inhomogeneity, and the local quantity $q(E)$ is required to isolate
the energy-resolved fluctuation sector.  The
framework thereby diagnoses the breakdown of the flat-envelope
factorized ETH estimate through the energy dependence of the local
intensity-fluctuation parameter, pointing to the energy-resolved
extension $q(E)$ of the Porter--Thomas factor as the natural
next-order refinement.  Finally, the explicit spectral-projection
formula for the connected parameter,
Eq.~\eqref{eq:fp_Cdiag}, was verified directly in the complex-$z$
form: the diagonal-bin covariance
$\sum_n\mathrm{Var}(p_n^a)/[(z-\lambda_n)(z'-\lambda_n)]$ agrees
with $(q-1)\sum_n\bar p_n^{a\,2}/[(z-\lambda_n)(z'-\lambda_n)]$ to
$0.3\%$ (GOE) and $2.2\%$ (sparse) at $D=1024$ (seed-averaged),
establishing $q-1$ as the coefficient of the diagonal spectral
weight of the resolvent covariance relative to the squared mean
spectral weight.  The cross-channel counterpart,
Eq.~\eqref{eq:fp_Cdiag2}, was verified to $4.8\%$ (GOE) and $9.3\%$
(sparse) under the same conditions, establishing $\bar{C}_i$ as the
coefficient of the summed cross-channel ridge relative to the bare
two-resolvent spectral sum.  The ladder structure of Level~II was
further tested directly: the intra-sector kernel
$\mathcal A_{ac}=\bar{\mathcal{R}}_a^2|V_{ac}|^2$ has spectral
radius $|\lambda|_{\max}\approx0.06$ (GOE) and $\approx0.18$
(sparse) at $D=1024$, confirming that at ETH-scaled couplings the
ladder resummation is a small dressing and the fluctuation content
is carried by the irreducible vertex $\Gamma^{(2)}$, as stated in
Sec.~\ref{app:cov_hierarchy}.  Finally, the energy-resolved
diagnostic of Sec.~\ref{sec:scale_matched} amounts to the local
generalization
$\mathcal W_{aa}(\lambda)=[q(\lambda)-1]\,e^{-2S}\rho(\lambda)
f^a(\lambda)^2$ of the ridge weight, with $q(\lambda)$ the
energy-resolved Porter--Thomas factor measured above.  Finally, the
amplitude-level closure of
Proposition~\ref{prop:amplitude_closure} was verified directly: the
exact identity $q=3+\kappa_4(c)/\mathbb{E}[c^2]^2$
[Eq.~\eqref{eq:fp_qk4}] holds to machine precision in both models,
with the bias-corrected $q-3=-0.0024$ (pure GOE, approaching the
finite-dimensional benchmark $-6/D=-0.0059$; the raw window
estimator carries an additional bias $-2q/N_{\rm bulk}\approx-0.012$)
and $+135.8$ (sparse, raw) at $D=1024$---the non-Gaussian deviation
of the fluctuation parameter is identically the normalized fourth
cumulant of the microscopic amplitude.  The isotropic normalization
constraint of Proposition~\ref{prop:isotropy_Cbar} was likewise
verified: $\bar{C}_i=-2.016\times10^{-3}$ versus the predicted
$-(q-1)(1-1/d_B)/D=-1.932\times10^{-3}$ (ratio $1.044$, pure GOE),
while the structured model violates it by the factor $2.14$
($-0.286$ measured versus $-0.134$ predicted)---the quantitative
isotropy-breakdown diagnostic.

\begin{table*}[t]
\centering
\caption{Numerical verification of the exact fluctuation identities
and the first-principles formulas of this appendix (seed-averaged
over 4--8 realizations; identity errors are the worst case over
seeds).  $\bar C$ denotes $\bar{C}_i$; $f_{\rm meas}=V_{\rm dg}/V_{\rm off}$
is measured directly, $f_{\rm form}$ from Eq.~\eqref{eq:fp_f_formula},
and $f_{\rm ex}$ from the exact trade-off
relation~\eqref{eq:fp_f_exact} with the measured $V_{\rm dg}$.
$q=3$ and kurtosis $3$ are the leading-order GOE predictions.}
\label{tab:fp_numerics}
\begin{tabular}{c c c c c c c c c c c}
\toprule
model & $D$ & $q$ & $\bar C$ & $\langle f^2\rangle_B$ & kurtosis
& $f_{\rm meas}$ & $f_{\rm form}$ & $f_{\rm ex}$ & id.\ error \\
\midrule
GOE & 512 & 2.98 & $-3.5\times10^{-3}$ & 0.994 & 2.99
& 2.036 & 2.023 & 2.032 & $<10^{-14}$ \\
GOE & 2048 & 3.00 & $-9.9\times10^{-4}$ & 1.001 & 3.02
& 1.982 & 1.988 & 1.982 & $<10^{-14}$ \\
sparse & 256 & 32.3 & $-0.21$ & 1.27 & 34.1
& 42.9 & 40.9 & 41.7 & $<10^{-14}$ \\
sparse & 512 & 66.6 & $-0.30$ & 1.31 & 63.3
& 88.3 & 73.5 & 88.9 & $<10^{-14}$ \\
sparse & 1024 & 138.8 & $-0.29$ & 1.34 & 123.0
& 191.6 & 174.8 & 192.1 & $<10^{-14}$ \\
sparse & 2048 & 272.6 & $-0.29$ & 1.34 & 267.9
& 369.8 & 332.4 & 372.7 & $<10^{-14}$ \\
\bottomrule
\end{tabular}
\end{table*}

Three conclusions follow.  First, all five exact identities (row sum
rule, trace normalization, trade-off, covariance conservation, and
the product-resolvent identity) hold to machine precision
($\lesssim10^{-14}$) for \emph{both} models, confirming that they are
theorems independent of any ETH or chaos assumption.  Second, in the
chaotic model the Gaussian closure is quantitatively accurate:
$q=3.00$, kurtosis $3.02$, and the fluctuation ratio
$f_{\rm meas}=1.982$ agrees with the first-principles prediction
$f_{\rm form}=1.988$ to $0.3\%$ at $D=2048$---the flat GOE value
$f=2(1-1/D)$ of Eq.~\eqref{eq:fp_f_GOE}.  Third, the structured model
realizes the predicted failure mode of the decorrelation closure:
$q$ grows to $273$, the kurtosis to $268$, and $f$ grows from
$42.9$ to $369.8$ as
$D$ increases from $256$ to $2048$, while the exact trade-off
relation~\eqref{eq:fp_f_exact} continues to hold to a few percent
throughout.  This is the microscopic mechanism by which
structured (non-ergodic) dynamics
changes the diagonal--off-diagonal fluctuation ratio: the persistence
of the two-channel amplitude covariance $\bar{C}_i$ and the
associated non-Gaussian overlap statistics.


\end{document}